\documentclass[aps,twocolumn,pra,superscriptaddress,nofootinbib]{revtex4-2}

\usepackage{amsmath}

\usepackage{amssymb}
\usepackage{amsthm}
\usepackage{thmtools}
\usepackage{mathrsfs}
\usepackage{bm, dsfont}
\usepackage[usenames,dvipsnames]{color}

\usepackage{graphicx}
\usepackage{natbib}
\usepackage[colorlinks=true,linkcolor=blue,citecolor=blue,urlcolor=blue]{hyperref}
\usepackage[capitalise]{cleveref}
\usepackage{hypcap}
\usepackage{verbatim, float}
\usepackage{psfrag}
\usepackage{tikz}
\usetikzlibrary{shapes.geometric, positioning, arrows.meta, decorations.pathmorphing, shadows,shapes.symbols,shadows.blur}
\usepackage[normalem]{ulem}
\usepackage{physics}		% for braket, etc.

\usepackage{mathtools}
\usepackage{relsize}

\newcommand{\id}{\mathds{1}}

\newcommand{\cA}{\mathcal{A}}

\newcommand{\cE}{\mathcal{E}}

\newcommand{\cX}{\mathcal{X}}

\newcommand{\be}{\begin{equation}}
\newcommand{\ee}{\end{equation}}

\newcommand{\ChJa}{\text{Choi-Jamio{\l}kowski }}

\newcommand{\SWAP}{\textsc{Swap} }

\newtheorem{theorem}{Theorem}%[section]

\newtheorem{lemma}[theorem]{Lemma}

\definecolor{darkblue}{RGB}{50,10,180}
\definecolor{darkgreen}{RGB}{10,180,50}
\definecolor{orangef}{RGB}{210,100,20}

\makeatletter 
    
\renewcommand\onecolumngrid{% <<<<<<
\do@columngrid{one}{\@ne}%
\def\set@footnotewidth{\onecolumngrid}% <<<<<<<<<<<<<<<<
\def\footnoterule{\kern-6pt\hrule width 1.5in\kern6pt}%
}

\renewcommand\twocolumngrid{% <<<<<<
        \def\footnoterule{% restore rule
        \dimen@\skip\footins\divide\dimen@\thr@@
        \kern-\dimen@\hrule width.5in\kern\dimen@}
        \do@columngrid{mlt}{\tw@}
}%

\makeatother    
\begin{document}

\title{Self-testing Quantum Supermaps}

\author{Victor Barizien}
\affiliation{Department of Applied Physics, University of Geneva, Geneva, Switzerland}

\author{Cyril Branciard}
\affiliation{Université Grenoble Alpes, CNRS, Grenoble INP, Institut Néel, 38000 Grenoble, France}

\author{Alastair A.\ Abbott}
\affiliation{Université Grenoble Alpes, Inria, 38000 Grenoble, France}

\author{Jean-Daniel Bancal}
\affiliation{Université Paris-Saclay, CEA, CNRS, Institut de Physique Théorique, 91191, Gif-sur-Yvette, France}

\author{Pavel Sekatski}
\affiliation{Department of Applied Physics, University of Geneva, Geneva, Switzerland}

\begin{abstract}
By certifying quantum operations from measurement statistics directly, without any assumption on the internal workings of the devices involved, self-testing enables a uniquely reliable identification of quantum objects. While such device-independent characterization has been shown to be possible for states, measurements and channels, it has so far not been extended to quantum supermaps -- operations that act on quantum channels themselves and can combine them in either a well-defined causal order or also, remarkably, in an indefinite causal order. Here we show that quantum supermaps can be identified device-independently. Specifically, we obtain two levels of certification, depending on the network structure of the experiment: when each slot of the supermap accepts a single uncharacterized black box, identification up to local embedding combs is obtained; when several black boxes are inserted within each slot, identification up to local extracting and injecting maps is achieved. We illustrate our approach on four examples -- the identity comb, a bit-flip error-correcting comb, the comb describing Grover's algorithm, and the quantum switch -- providing in particular the first self-test of both a quantum algorithmic comb and a causally indefinite quantum process.  Notably, in the latter case, this provides a new way to certify causal indefiniteness in a device-independent manner.
\end{abstract}
\maketitle

%%%%%%%%%%%%%%%%%%%%%%%%%%%%%%%%%%%%%%%%%%%%%%%%%%%%%%%%
\section{Introduction}
\label{sec:introduction}
%%%%%%%%%%%%%%%%%%%%%%%%%%%%%%%%%%%%%%%%%%%%%%%%%%%%%%%%

The characterization and certification of quantum resources is a central problem in quantum information science. Using one of the preeminent features of quantum theory, namely its ability to generate non-classical correlations, has become a powerful tool for obtaining such characterization. One of the starkest examples of this is the concept of self-testing~\cite{Mayers04, Supic20}. Self-testing exploits quantum correlations -- typically the maximal violation of a Bell inequality -- to provide the most complete quantum description of a system under study that is achievable in a black-box setting, i.e.~in a situation where neither an internal description of the system is available, nor of the additional devices used to analyse it.
The canonical example of this is in self-testing states, where the maximal violation of the CHSH inequality~\cite{Clauser69} can be used to certify the fact that a device prepares a state that is, up to local transformations, equivalent to the maximally entangled Bell state~\cite{Kaniewski16}. As a result, self-testing is sometimes considered the strongest form of certification of a quantum resource: it maximally characterises the resource in a device-independent setting, directly from observable correlations. Self-testing has been demonstrated for several key quantum resources, including quantum states~\cite{Coladangelo_2017, balanzo2026all} and measurements~\cite{Chen_2024,storz2025complete}, as well as quantum channels~\cite{Magniez05,Sekatski18} and instruments~\cite{wagner2020device}. 
Here, we extend the self-testing framework to the so-called higher-order quantum operations, also known as quantum supermaps or processes~\cite{chiribella2008transforming}. 

Quantum supermaps are general operations that map completely positive (CP) maps to CP maps. Prominent examples of quantum supermaps are so-called quantum combs, which can be realized as quantum circuits which apply the input CP maps in a fixed order.  
Quantum algorithms based on oracles~\cite{dewolf2023quantum} can for instance be described as quantum combs, such as Grover's algorithm~\cite{grover1996fast}, which can be viewed as a quantum comb operating on $M$ copies of an unknown unitary channel to produce a classical guess. Particular interest has also been sparked by the fact that some supermaps cannot be described as combs, but may combine operations in a dynamically established order or even without any well-defined causal order~\cite{Oreshkov12,costa26}. The archetypal supermap of this kind is the quantum switch~\cite{Chiribella13}, which uses a qubit to control the order in which two parties' quantum operations are performed. Supermaps have found numerous applications in computation~\cite{araujo14,Abbott24}, communication~\cite{guerin16}, or describing strategies in channel discrimination~\cite{chiribella08} and quantum metrology~\cite{liu23,bavaresco2024designing}; see Ref.~\cite{taranto2025} for a review.

In this contribution, after reviewing the necessary background on self-testing of states, measurements and channels~(Sec.~\ref{sec:ST_quantum_operations}), we propose a formal definition of self-testing for quantum supermaps. Considering first supermaps acting on quantum channels consisting of single black boxes, we show that they can be identified only up to local quantum combs, or ``\textit{embedding combs}''~(Sec.~\ref{sec:ST_approach1}). However, we argue that network connectivity assumptions make it possible to place several (uncharacterized) separate black boxes within each slot of a supermap, and in particular access the input and output systems of a slot independently. This allows the identification of supermaps with more precision, namely up to local maps at the input and output of each slot (Sec.~\ref{sec:ST_approach2}), similar to the injection and extraction maps used in the self-testing of quantum channels~\cite{Sekatski18}.

In practice, our method relies on the certification of the reference quantum channel obtained by plugging a swap gate into each slot of the target quantum process. We illustrate the difference between the certifications obtained with different network connectivity assumptions on several examples~(Sec.~\ref{sec:examples}). We show that our framework allows for general self-testing of a variety of quantum processes by providing explicit self-tests for the identity comb, the bit-flip error-correcting comb, the comb describing Grover's algorithm, and the quantum switch. These provide in particular the first self-test of an algorithmic comb and of a causally indefinite quantum process. 
We finish by discussing several open questions regarding the extent to which the assumptions of our approach can further be relaxed or simplified, and its generalizability to other supermaps beyond the quantum switch (Sec.~\ref{sec:discussion}).

%%%%%%%%%%%%%%%%%%%%%%%%%%%%%%%%%%%%%%%%%%%%%%%%%%%%%%%%
\section{Self-testing quantum operations}
\label{sec:ST_quantum_operations}
%%%%%%%%%%%%%%%%%%%%%%%%%%%%%%%%%%%%%%%%%%%%%%%%%%%%%%%%

Generally speaking, self-testing refers to the possibility of \emph{identifying} the quantum model of different elements of an experimental setup as precisely as possible given only the knowledge of observed correlations and general assumptions on the \emph{network structure} describing the experiment. Both questions of what it means to identify precisely a physical model and to assume a network structure require some clarification, which is easier to vehicle through concrete examples as we do below. In the following, we also distinguish \emph{self-testing} from \emph{rigidity} statements. Specifically, we refer to self-testing as the fact that certain sets of conditions on the experimental statistics obtained in given network structures imply a rigidity statement, i.e.~a characterization of part (or all) of the physical devices used in the experiment.

\subsection{Self-testing states and measurements}

First, consider the example of self-testing in bipartite Bell scenarios. Here, a common source prepares two physical systems and distributes them to two parties, Alice and Bob. Upon receiving their systems, Alice and Bob each choose a random measurement, labeled with $x$ and $y$, and perform it on their respective system to obtain an output, labeled $a$ and $b$, respectively. Upon identically repeating the described procedure many times Alice and Bob estimate the conditional probability distributions -- or \emph{correlations} -- $P(a,b|x,y)$ associated to the experiment. In quantum physics these correlations are given by the Born rule
\begin{subequations}
\begin{equation}\label{eq: corr}
    P(a,b|x,y) = {\rm Tr}\left(\rho\,  (A_{a|x} \otimes B_{b|y})\right),
\end{equation}
for some, a priori unknown, bipartite state $\rho$ and measurements $A_{a|x}$, $B_{b|y}$, referred to as a {\it realization}.\footnote{Throughout the paper (including in the diagrams) we use the ``mixed state picture'', with quantum states represented as density matrices, measurements as Positive Operator-Valued Measures (POVMs), and quantum operations as Completely Positive (CP) maps.}

The scenario described in the previous paragraph can be simply summarized by the following diagram 
\begin{equation} \label{Bell network} \begin{aligned}
\includegraphics[height=2.5 cm]{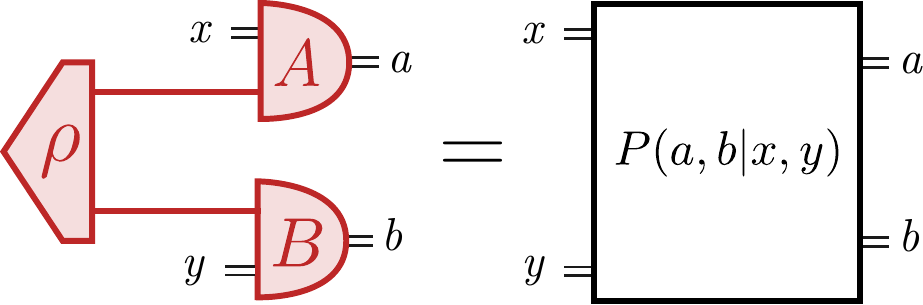}\ \ ,
\end{aligned}
\end{equation}
\end{subequations}
which specifies both the \emph{network structure} underlying the experiment (left hand side) as well as the observed correlations (right hand side). Each diagram is to be read from left to right, similarly to a quantum circuit. Boxes specify different devices in the setup, here (on the lhs above) a source preparing a bipartite state and two local measurements. Lines specify how quantum (single-stroke lines) or classical (double-stroke lines) systems are exchanged between the devices.  We use the red color to signify that the underlying devices or systems are unknown, and no assumption is made on their quantum models nor on the (finite or countably infinite) dimension of the associated Hilbert spaces. Notice that a crucial aspect of the assumed network connectivity is that the sources preparing the quantum state and the classical inputs $x$ and $y$ are independent. Indeed, if they were allowed to share arbitrary correlations, the devices could ``agree'' to display any pattern of inputs and outputs to the experimentalists, which is sometimes referred to as the super-determinism loophole~\cite{Hossenfelder20}. To keep it simple, in the diagrams we do not explicitly represent the sources of classical inputs, rather we associate the inputs with independent open classical wires to signify that they can be chosen freely.

When performing an experiment, one usually aims to implement a specific \emph{reference} -- or \emph{target} -- realization featuring $\bar \Phi$, $\bar A_{a|x}$ and $\bar B_{a|x}$:
\begin{equation} \label{target res} \begin{aligned}
\includegraphics[height=2.5 cm]{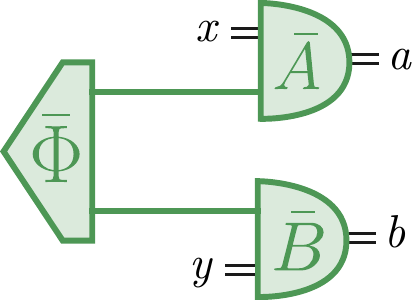}\ \ .
\end{aligned}
\end{equation}
As opposed to the lhs of \cref{Bell network}, here we use the green color to signify that the corresponding elements are known and trusted.
For a given target realization, the expected correlations $P(a,b|x,y)$ are straightforward to compute from the Born rule in Eq.~\eqref{eq: corr}.

Remarkably, for some observed correlations this logical link can be inverted, which is the essence of self-testing: the correlations allow one to essentially uniquely identify the underlying physical realization (here the state and measurements). For example, Alice and Bob observing the maximal quantum score, $2\sqrt{2}$, of the so-called CHSH Bell  test~\cite{Clauser69}, is essentially only compatible with the source preparing a maximally entangled two-qubit state, and the measurements acting as complementary Pauli measurements on these qubits~\cite{Supic20}. Special care is however needed to make this claim precise.

Indeed, it is not clear how a state can be identified within a device with no prior choice of basis or reference frame. Furthermore, one can not possibly know anything about auxiliary degrees of freedom which are prepared by the source but ignored by the measurement devices, and hence do not influence the statistics. This is easy to understand through the following diagram
\begin{equation}\label{physical res} \begin{aligned}
\includegraphics[height=4.5 cm]{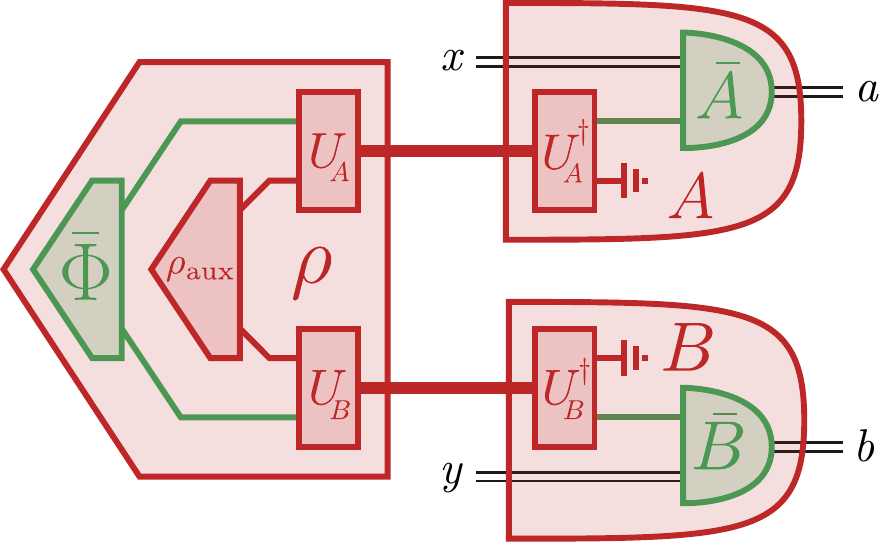} \ \ ,
\end{aligned}
\end{equation}
showing an example of a physical realization that can never be distinguished from the reference one of Eq.~\eqref{target res} in the black-box setting. Here, $U_A$ and $U_B$ are (unknown) unitary maps. Since both realizations in diagrams~\eqref{target res} and~\eqref{physical res} yield the same observed correlations and respect the network connectivity, we can only hope to identify the underlying quantum model up to some equivalence relations.

\begin{figure*}
    \centering
    \includegraphics[width=15.75cm]{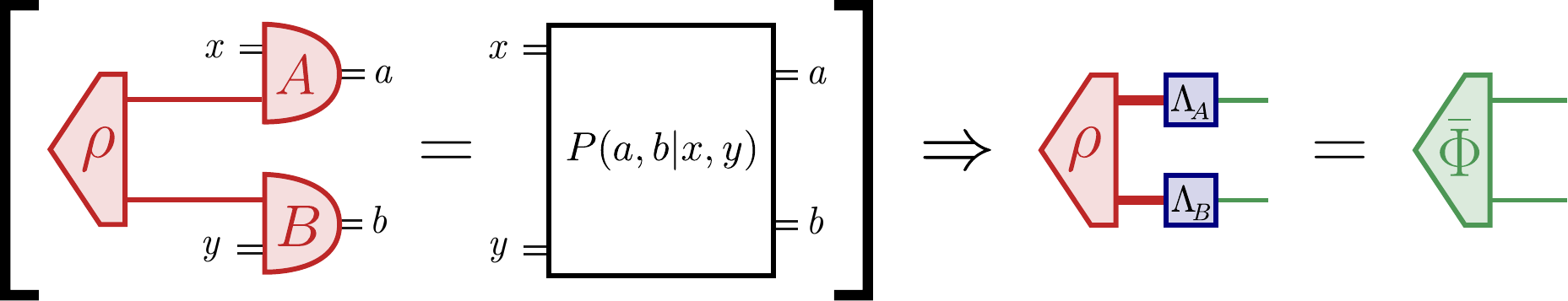}
    \caption{\textbf{Diagrammatic definition of self-testing of a quantum state in a bipartite Bell scenario.} 
    The fact (lhs) that an uncharacterized (red) setup reproduces certain statistics $P(a,b|x,y)$ in a specific network structure implies state rigidity (rhs), i.e.~the existence (blue) of maps extracting $\bar\Phi$ from the measured state $\rho$. The lhs contains the assumptions that the setup $(i)$ is described by a quantum model; $(ii)$ satisfies the required network connectivity; and $(iii)$ generates the specified statistics.}
    \label{fig:selftestingV2}
\end{figure*}

Formally, as summarized in Fig.~\ref{fig:selftestingV2}, we say that the correlations self-test  the target state $\bar \Phi$ if for any physical realization $\rho$, $A_{a|x}$, $B_{b|y}$ compatible with some given statistics $P(a,b|x,y)$, there exist local {\it extraction maps} $\Lambda_A$ and $\Lambda_B$ such that 
\begin{subequations}\label{eq: st states}
\begin{equation}\label{eq: st states eq}
\Lambda_A\otimes \Lambda_B[\rho] = \Bar \Phi.
\end{equation}
To lighten notation, we represent such {\it rigidity statements}\footnote{While we focus here on exact rigidity identities, self-testing statements typically come with some \textit{robustness}: approximate (noisy) statistics lead to approximate rigidity statements, which is however not the focus of this work.} 
graphically as 
\begin{equation}\label{eq: st states diag}\begin{aligned}
\includegraphics[height=1.5 cm]{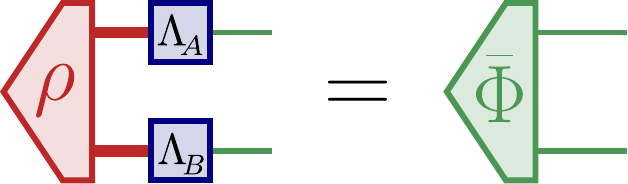} \ \ ,
\end{aligned}
\end{equation}
\end{subequations}
where the blue color means that the existence of an object is guaranteed. In our example, as every realization of the type of \cref{physical res} fulfills \cref{eq: st states} with the choice of extraction maps 
\begin{equation}\begin{aligned}
\includegraphics[height=1 cm]{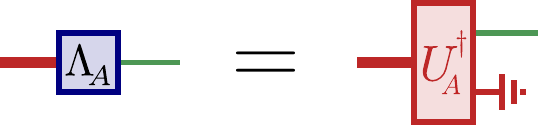}\ \ , \\\includegraphics[height=1 cm]{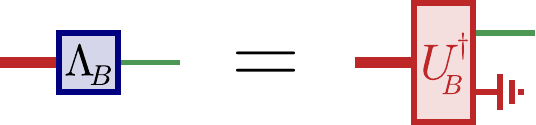}\ \ ,
\end{aligned}
\end{equation}
this captures the discussed freedom. Operationally, the local extraction maps extract the subsystems of desired dimension (when applied on the physical state $\rho$) such that the state they output equals the target one~$\bar \Phi$. They can be understood as a generalized choice of basis for the devices under test.

Things are a little different when it comes to identifying the measurements performed by Alice and Bob: in this case, ignoring the auxiliary degrees of freedom controlled by the source may be detrimental for the proper functioning of the measurement boxes so they cannot be merely ignored. For example, the usual polarization analyser, consisting of a polarizing beamsplitter followed by two single photon detectors, would only work if the incident photon is in the right frequency range. Hence, an identification of a measurement box must be allowed to set the auxiliary degrees of freedom right. Formally, we say that the correlations self-test the target measurement $\bar A_{a|x}$, if for any realization compatible with the observed statistics $P(a,b|x,y)$ there exists a local {\it injection map} $V_A$ (independent of $x$  and $a$) whose input system is of known dimension, such that
\begin{subequations}\label{eq: st meas}
\begin{equation}\label{eq: st meas eq}
V_A^* [A_{a|x}]   =  \bar{A}_{a|x}.
\end{equation}
This rigidity relation can also be represented graphically as
\begin{align}\label{eq: st meas diag}
\begin{aligned}\includegraphics[height=1 cm]{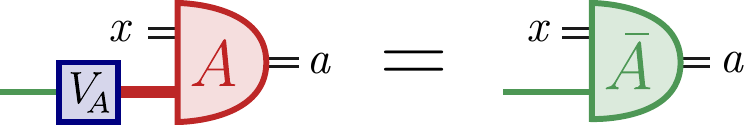} \ \ .
\end{aligned}
\end{align}
\end{subequations}
Again, every realization of the type of \cref{physical res} fulfills \cref{eq: st meas} with the following choice of injection map: 
\begin{equation}\label{eq: inj map VA}\begin{aligned}
\includegraphics[height=2.5 cm]{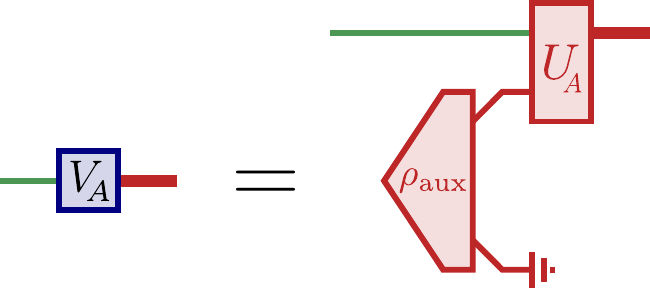}\ \ .
\end{aligned}
\end{equation}

We leave aside the discussion of how exactly self-testing statements are proven for states or measurements, i.e.~how the rigidity statements are implied from the observed statistics. It is nevertheless worth mentioning that today a plethora of results are available~\cite{Supic20}, including self-tests for all pure bipartite entangled states~\cite{Coladangelo_2017}, all pure multipartite entangled qubit states~\cite{balanzo2026all}, and all real projective measurements of arbitrary dimension~\cite{Chen_2024}. In Appendix~\ref{app:self_testing_1comb}, we discuss self-testing statements for products of maximally entangled states. In addition, we derive self-tests for a classical mixture of two GHZ-like states, a global unitary transformation of a  product of $n$ maximally entangled states and a specific 10-qubit state, which will be important for the certification of the supermaps presented as examples of our method in~\cref{sec:examples}. Note that since the latter relies on the certification of multiple states with possibly similar measurements, it is of great interest to make use of methods that allow for deriving self-tests of target states and target measurements, see e.g.~\cite{Barizien24}. Before turning to other types of quantum operations, let us briefly recall some features of measurement self-testing that are specific to the multipartite case.

\subsection{Self-testing multipartite measurements}

\label{sec: self-test measmts}
Multipartite measurements, such as the two-qubit Bell-state measurement (BSM) $\{\bar M_i\}_i =\{\Phi^+,\Phi^-,\Psi^+,\Phi^-\}$ composed of the four Bell-state projectors, are central to many quantum information tasks. Such  measurements can naturally be abstracted in the device-independent framework as a box with several physical inputs that can be manipulated independently, and one classical output. 

In  an experiment where a single global source is used to self-test the BSM, i.e.~both inputs may come from the same physical device, it is crucial to note that the physical realization of the BSM could rely on an auxiliary maximally entangled state distributed to the two inputs by the global source. In particular, this allows for the following realization of a BSM via local operations and classical communication~\cite{Sekatski18}:
\begin{align}\label{BSMbad}
\begin{aligned}\includegraphics[height=4 cm]{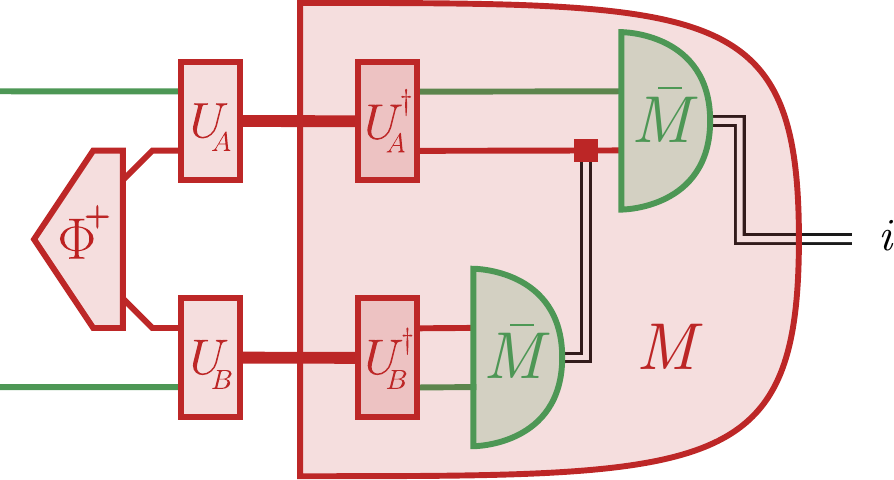},
\end{aligned}
\end{align}
where the classical control inside the box $M$ applies a Pauli correction (as in quantum teleportation). 
In this situation, the best rigidity one can hope for is of the form~\cite{bancal2018noise}
\begin{align}\label{BCM1source}
\begin{aligned}\includegraphics[height=1.5 cm]{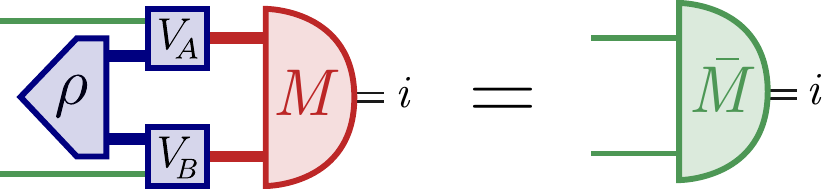}\ \ ,
\end{aligned}
\end{align}
where the injection maps also receive a joint quantum state.

This rigidity relation can however be  strengthened if the two inputs of the BSM are trusted to come from independent sources. In this case, the multipartite measurement can be self-tested by guaranteeing~\cite{bancal2018noise,renou18self} the existence of independent local injection maps for each of the separate inputs such that
\begin{align}
\begin{aligned}\includegraphics[height=1.5 cm]{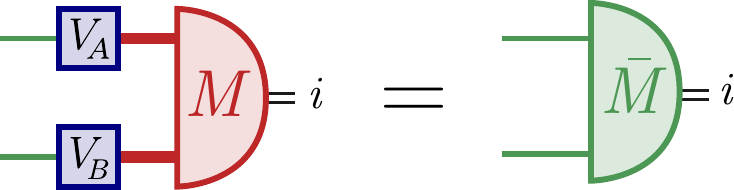}\ \ .
\end{aligned}
\end{align}
Since local injection maps are a special case of the entanglement-assisted injections in \cref{BCM1source}, for which no auxiliary correlations are required, the resulting identification of the measurement is more precise.

This highlights the role of the network structure in self-testing and illustrates the subtle interplay between assumptions on the network structure underlying an experiment and the self-testing possibilities it offers. While already crucial to ensure the independence of measurement choices in bipartite Bell tests, the network structure plays an even more important role in multipartite self-tests~\cite{Supic23quantum}.

\subsection{Self-testing quantum channels} \label{sec:self-test_channels}

To formalize the self-testing of dynamical quantum operations (channels and instruments), which have both quantum inputs and outputs, one essentially combines the definitions for the states~\eqref{eq: st states} and measurements~\eqref{eq: st meas}, which each have respectively a quantum output and a quantum input. Specifically, consider a physical realization $\mathcal{E}$ of an $n$-partite quantum channel, describing a black box that has $n$ well-identified quantum inputs and outputs,\footnote{Note that the number of output systems could be taken different from the number of input systems $n$ without much affecting the following discussion.} that can be manipulated independently. We say that the target channel $\bar \cE$ is self-tested in the experiment (with independent sources), if there exist $n$ independent local injection maps $V_i$ and extraction maps $\Lambda_j$, such that  
\begin{subequations}\label{selftesting_channels}
\be \label{selftesting_channels_eq}
 (\bigotimes_{j=1}^n \Lambda_{j})\circ \cE \circ (\bigotimes_{i=1}^n V_i) = \bar \cE.
\ee
Again, this rigidity statement can be simply given a graphical representation as follows:
\begin{align} \label{selftesting_channels_bis}
\begin{aligned}\includegraphics[height=1.5cm]{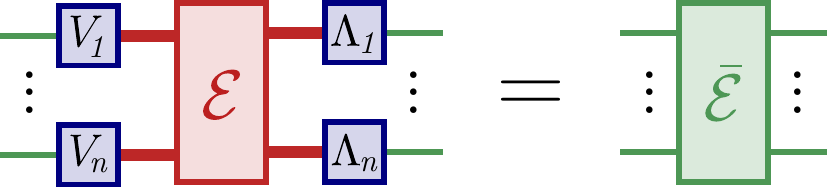} \ \ .
\end{aligned}
\end{align}
\end{subequations}

A very convenient way to represent dynamical quantum operations is offered by the Choi-Jamio\l{}kowski isomorphism~\cite{Jamiolkowski72}, which states that an $n$-partite quantum channel is faithfully represented by its $2n$-partite Choi state
$\bar\cE \simeq \mathcal{C}(\bar\cE ):= ({\rm id^{\otimes n}}\otimes \bar{\cE})[(\Phi^+)^{\otimes n}]$ (where $\Phi^+$ is a maximally entangled state), i.e.
\begin{align} \label{selftesting_channels_choi}
\begin{aligned}\includegraphics[height=2.5cm]{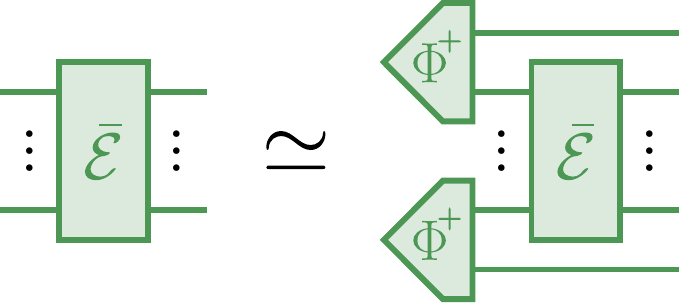} \ \ .
\end{aligned}
\end{align}
The Choi-Jamio\l{}kowski isomorphism gives us yet another way to express the rigidity relation of \cref{selftesting_channels} as
\begin{equation} \label{selftesting_channels_ter} \tag{12c}
\begin{aligned}
\includegraphics[height=2.5cm]{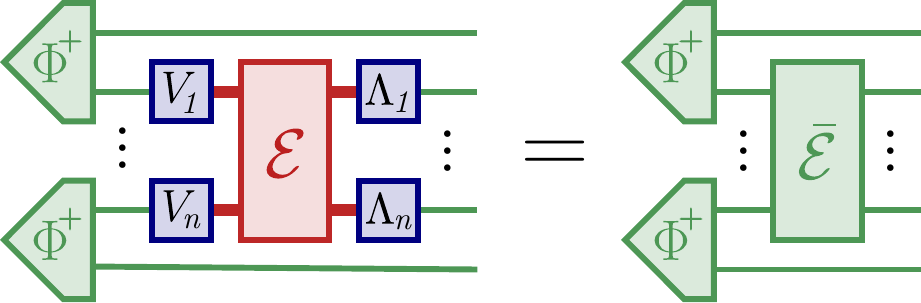} \ \ .
\end{aligned}
\end{equation}
 
The next natural question, then, is how to chose a network structure and devise an experiment featuring the target channel $\bar\cE$ where the latter can be self-tested. This was addressed in Ref.~\cite{Sekatski18}, which provided a generic way to devise such an experiment and combine self-tests for certain well chosen states into a self-test of the target channel. The idea, illustrated in~\cref{fig:selftesting_channel} is to consider an experiment where $n$ independent sources, each ideally preparing maximally entangled states $\Phi^+$ of the desired dimension, are combined with the channel and with local measurements in two different manners.

\begin{figure*}
    \centering
    \includegraphics[width=17.5cm]{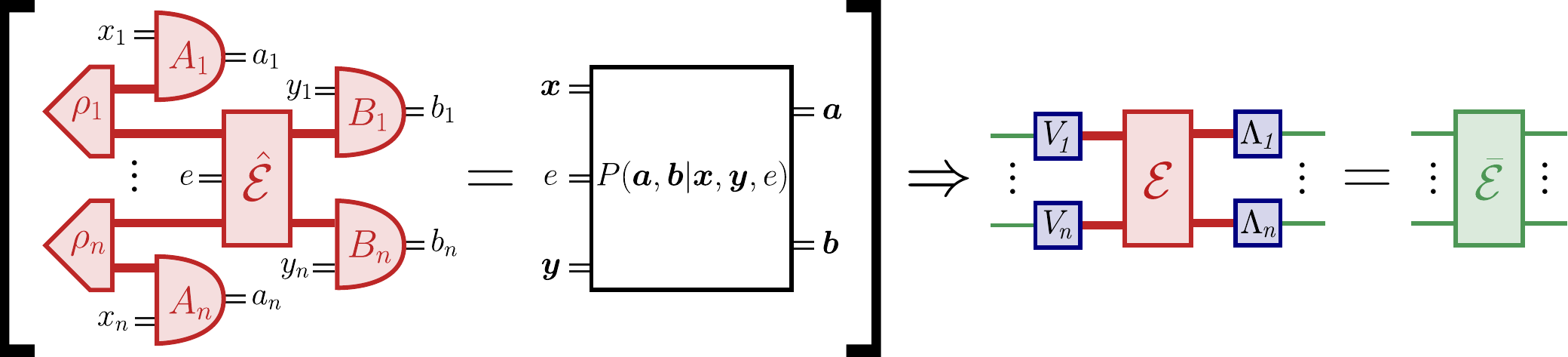}
    \caption{\textbf{Diagrammatic definition of self-testing of an $n$-partite quantum channel $\Bar{\mathcal{E}}$ in a $2n$-partite Bell scenario.} Obtaining the target probability distribution in the physical setup, where $\hat{\mathcal{E}} = \delta_{e,0} \,{\rm id} + \delta_{e,1}\,\mathcal{E}$ (with $\delta$ being the Kronecker delta, ${\rm id}$ denoting the identity channel -- i.e.~the classical bit $e$ decides wether the channel $\mathcal{E}$ is applied or not), implies channel rigidity: the reference (green) channel $\bar{\mathcal{E}}$ can be obtained from the physical (red) channel $\mathcal{E}$ by some local injection and extraction maps. As discussed in Sec.~\ref{sec:self-test_channels}, such a channel self-test requires the following conditions to be met: $(i)$ the target probability distribution $P(\bm a,\bm b|\bm x, \bm y , e=0)$, with $\bm x=(x_1,\ldots,x_n)$ and similarly for $\bm y,\bm a,\bm b$, -- corresponding the case where the channel is not applied -- self-tests (see~\cref{fig:selftestingV2}) each source to prepare a maximally entangled state $\Phi^+$; $(ii)$ the distribution $P(\bm a,\bm b|\bm x, \bm y , e=1)$ -- corresponding to the case where the channel is applied -- self-tests the Choi state $({\rm id}\otimes \bar{\mathcal{E}})[(\Phi^+)^{\otimes n}]$ of the reference channel; and $(iii)$ the two rigidity statements involve the same extraction maps on the ``channel-free'' systems (those on which the measurements $A_i$ are applied in the leftmost diagram).}
    \label{fig:selftesting_channel}
\end{figure*}

First, one performs local measurements on the sources combined with the physical channel in order to self-test the Choi state of the target channel $\mathcal{C}(\bar\cE )$ and obtain the following rigidity statement: 
\begin{align} \label{selftesting_channels_step1}
\begin{aligned}\includegraphics[height=2.5cm]{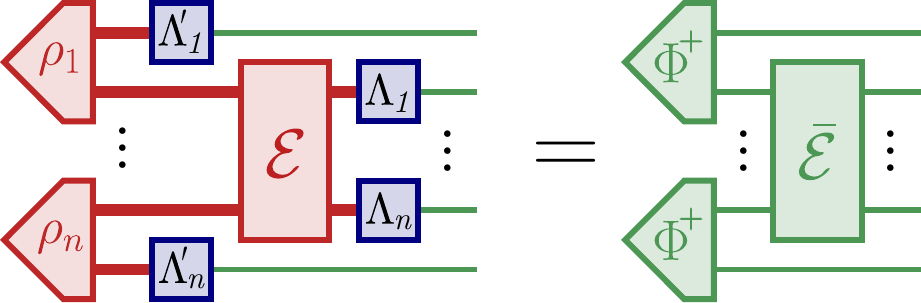} \ \ .
\end{aligned}
\end{align}
This relation exhibits an important difference with the desired one in Eq.~\eqref{selftesting_channels}: here, the physical realizations of the sources on which the physical channel $\cE$ acts are left uncharacterized, and are not guaranteed to prepare maximally entangled states. To close this discrepancy, the devices are combined in a second way to directly perform local measurements on the sources in order to self-test the state $(\Phi^+)^{\otimes n}$ they prepare:
\begin{align} \label{selftesting_channels_step2}
\begin{aligned}\includegraphics[height=2.5cm]{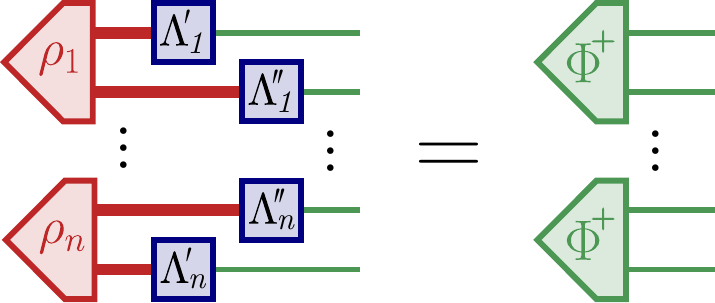}\ \ ,
\end{aligned}
\end{align}
which can be shown~\cite{Sekatski18} to imply
\begin{align} \label{selftesting_channels_step3}
\begin{aligned}\includegraphics[height=2.5cm]{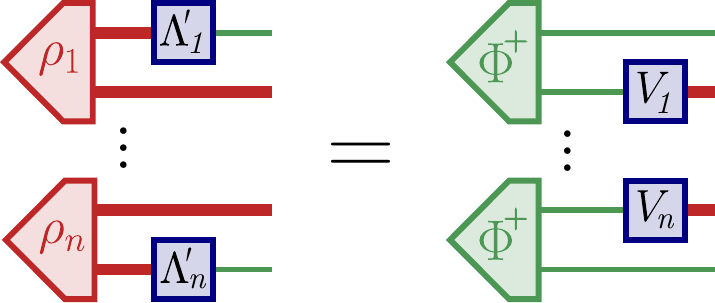}\ \ .
\end{aligned}
\end{align}
In order to finally obtain the desired result, we require the rigidity relations in Eqs.~(\ref{selftesting_channels_step1}) and (\ref{selftesting_channels_step3}) to feature the same extraction maps $\Lambda_i'$ on the ``channel-free'' systems. When this condition is fulfilled, the two rigidity statements can be readily combined to imply the channel rigidity in Eq.~\eqref{selftesting_channels}.

Note that self-testing of both states up to the same extraction maps $\Lambda_i'$ can be achieved by keeping the same measurement boxes on the channel-free system in both steps. This ensures that the extraction can not be influenced by the choice of whether the channel $\cE$ is applied or not on the remaining systems, similar to the assumption that a measurement of one one party has no influence on the measurement of the other party in standard Bell tests.
An explicit construction of such extraction maps is provided in Appendix~\ref{app:extractionmap}.

The rigidity statements we showcase here in Eq.~\eqref{selftesting_channels}, identifying multipartite channels up to local maps, hold when using independent sources to probe the $n$-partite channel. An analysis analogous to the one discussed in the previous section on multipartite measurements holds if independence between the sources is not guaranteed, resulting in a rigidity statements up to correlated injection maps, see~\cite{Sekatski18} for more details.

We also note that the identity channel plays a special role among channels with one input and one output. Indeed, any single-party channel $\bar \cE$ can be obtained by composing suitable injection and extraction maps with the identity channel (represented here by a simple wire)
\begin{align}
\begin{aligned}
    \includegraphics[height=0.8 cm]{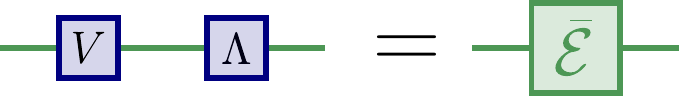}\ \ ,
\end{aligned}
\end{align}
chosen such that $\Lambda \circ V = \bar \cE$. This property extends to unitary channels, whose self-testing is thus equivalent to the one of the identity channel. However, all channels do not share this feature: entanglement-breaking channels, for instance, cannot reproduce all $\bar \cE$ (e.g.~the identity channel) even when supplemented with injection and extraction maps. For this reason, self-testing the identity channel is of particular interest~\cite{Sekatski18}. It has also shown to be practically relevant as quantum memories, frequency converters or transmission lines are all key devices for quantum technologies that ideally realize the identity channel and can be certified device-independently~\cite{sekatski2023toward,bock2024calibration,neves2025experimentally}.

To close this section we note that this framework can also accommodate self-testing of channels with classical outputs, i.e.~quantum instruments~\cite{heinosaari11}. For example, Ref.~\cite{wagner2020device} discussed self-testing of a non-demolition measurement.

We now turn to the self-testing of quantum resources beyond states, measurements and channels. In 
 the following, we provide a definition of self-testing for quantum supermaps, and illustrate it with several examples.

%%%%%%%%%%%%%%%%%%%%%%%%%%%%%%%%%%%%%%%%%%%%%%%%%%%%%%%%
\section{Self-testing quantum supermaps}
\label{sec:ST_quantum_supermaps}
%%%%%%%%%%%%%%%%%%%%%%%%%%%%%%%%%%%%%%%%%%%%%%%%%%%%%%%%

Quantum supermaps are linear transformations that map CP maps to CP maps and, in the case of deterministic supermaps, CP trace-preserving (TP) maps (i.e.~channels) to CPTP maps.%
\footnote{We will mostly consider deterministic supermaps, although our approaches can be extended to probabilistic ones (or superinstruments~\cite{Quintino2019}). E.g.~in the examples of the error-correcting comb (Sec.~\ref{sec:ECcomb}) or of Grover's algorithm (Sec.~\ref{sec:self_test_Grover}), the classical outcomes can be incorporated into classical output systems so as to make the combs (when including these classical systems) deterministic.}
They provide a natural framework for describing transformations of quantum operations, extending the standard formalism of quantum theory to higher-order quantum operations. For a complete review on this topic, see~\cite{taranto2025}. 

The simplest instance is given by one-slot supermaps, which transform a single input channel into another channel. These objects are operationally relevant as they describe the most general ways of embedding an unknown quantum operation into a larger circuit. A fundamental result shows that any one-slot supermap admits a memory-channel realization: it can be decomposed as a sequence of two CPTP maps with an intermediate memory system, such that the input channel is inserted between them~\cite{chiribella2008transforming}. 

The framework extends naturally to multi-slot supermaps, which act on several input channels. These can be thought of as query, or black-box, quantum algorithms, where each slot of the supermap corresponds to a query of an ``oracle channel''. Many such supermaps admit a circuit realization with a definite causal order -- corresponding to plugging the input channels into ordered ``slots'' of a quantum circuit. However, there also exist supermaps that cannot be decomposed in this way. In particular, some processes exhibit indefinite causal order and are termed causally nonseparable~\cite{Oreshkov12}. A paradigmatic example is the quantum switch~\cite{Chiribella13}, which uses a qubit to coherently control the order in which two channels are applied, and will be considered in~\cref{sec:self_test_QS}.

Before discussing a definition of self-testing for quantum supermaps, let us delve a little deeper into how these objects are defined and represented. For instance, consider some process $\bar W$ with one free input system, two open slots for input channels, and one output system:
\begin{align} 
\begin{aligned}\includegraphics[height=2.25cm]{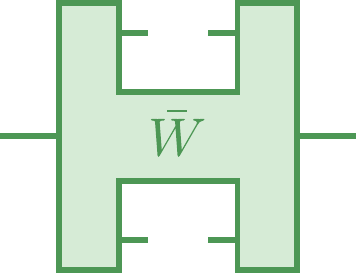}\ \ .
\end{aligned}
\end{align}
This object  is a supermap if, when plugging {\it any channels} $\mathcal{A}_1$ and $\mathcal{A}_2$ in its input slots, each possibly having auxiliary input/output wires, one always obtains a channel $\bar W[\mathcal{A}_1,\mathcal{A}_2]$ (linear in $\mathcal{A}_1$ and~$\mathcal{A}_2$): 
\begin{align} 
\begin{aligned}\includegraphics[height=3.25cm]{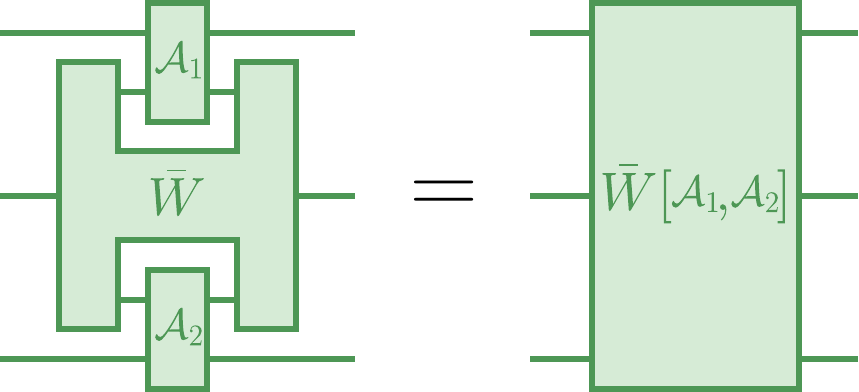}\ \ .
\end{aligned}
\end{align}

To represent supermaps, we will again leverage the Choi–Jamiołkowski isomorphism, that naturally generalizes to higher-order operations~\cite{chiribella2008transforming}. For our purpose it will be important that a supermap is faithfully represented by the channel $\mathcal{J}(\bar{W}):=\Bar W [\SWAP,\SWAP]$, obtained by plugging \SWAP gates in all of it slots, depicted in the right hand side of the following diagram:
\begin{align} \label{diag: Choi chanel}
\begin{aligned}\includegraphics[height=3cm]{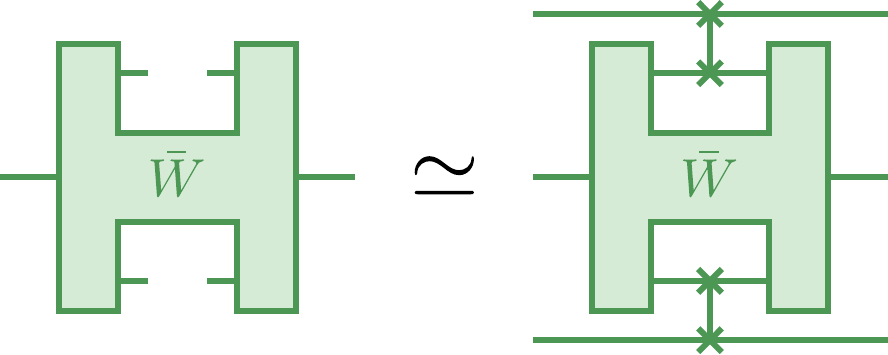} .
\end{aligned}
\end{align}
This representation immediately suggests that one could self-test the supermap $\bar W$ by self-testing its \emph{Choi channel} $\mathcal{J}(\bar{W})$ thus defined. With the help of the method presented in~\cref{sec:self-test_channels}, one can indeed devise an experiment to self-test this reference channel, leading to the following rigidity statement: 
\begin{align} \label{diag: Choi chanel2}
\begin{aligned}\includegraphics[height=3.25cm]{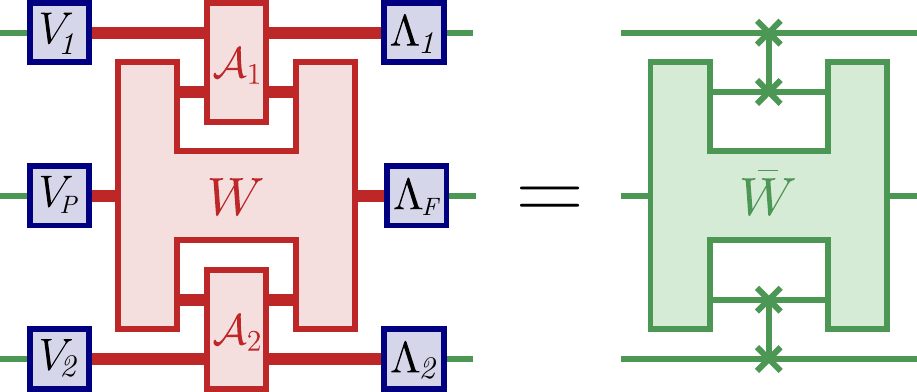}\ \ .
\end{aligned}
\end{align}
This identity, however, does not directly imply a self-test of the target process $\bar W$ itself, since the \SWAP gates introduced in the reference realization cannot be assumed to behave as faithful \textsc{Swap}s in the physical realization, but must also be treated as a priori unknown channels $\mathcal A_1$ and~$\mathcal A_2$.

There are now two ways to proceed that we introduce below, and which are summarized in~\cref{fig:supermaps}. The first approach aims to use the statistics to guarantee that these channels must indeed realize \SWAP gates. The second approach consists of guaranteeing that $\mathcal{A}_1$ and $\mathcal{A}_2$ realize \SWAP gates by ensuring that they consist of separate black boxes accessing the input and output systems of the slot independently.
This is akin to a refinement of the assumptions about the network connectivity used to describe the experiment, and thus leads to a stronger characterization of the quantum process.

\begin{figure*}
    \raggedright a) \\
    \centering
    \includegraphics[height=4.3125cm]{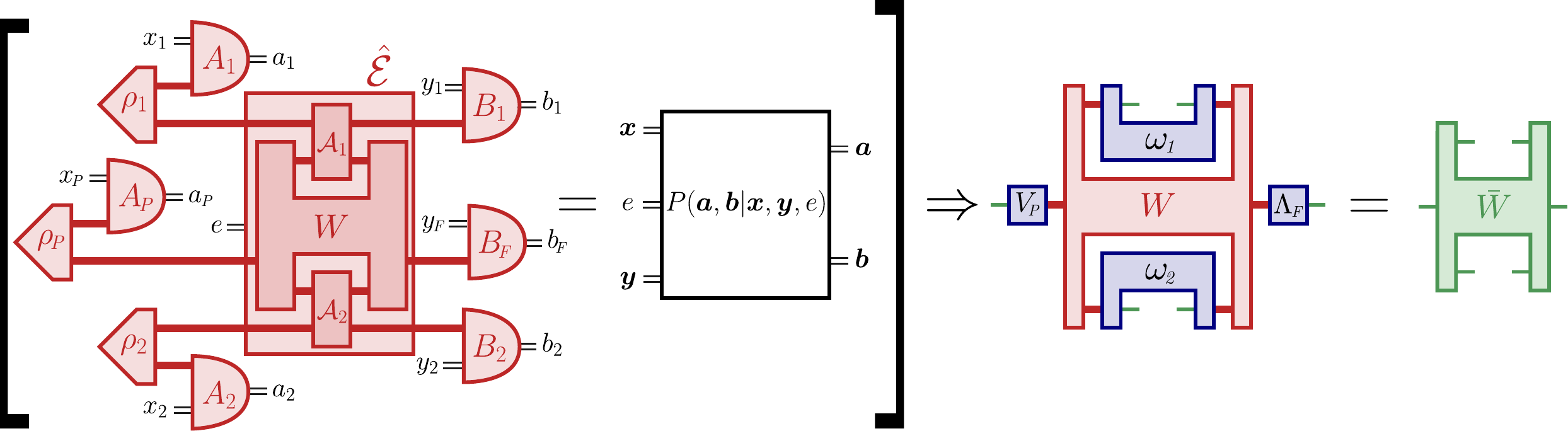}\bigskip\\
    \raggedright b) \\
    \centering
    \includegraphics[height=4.3125cm]{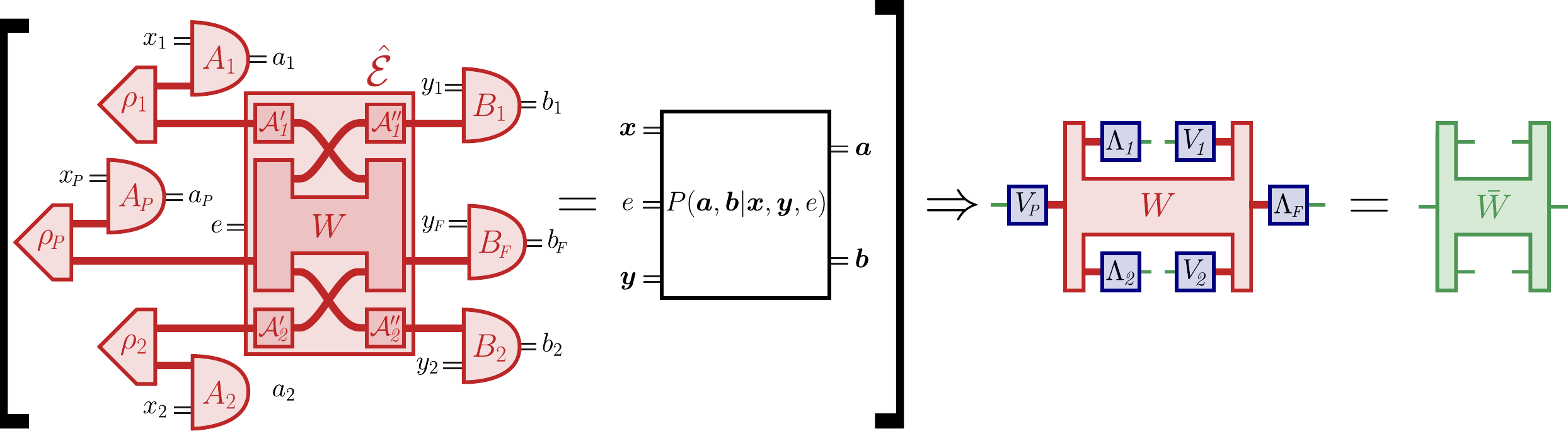}
    \caption{\textbf{Diagrammatic definitions of self-testing for a supermap $\Bar{W}$, illustrated here for a 2-slot supermap with a quantum input and output, requiring one to consider a 6-partite Bell scenario.} Obtaining the target probability distribution in the physical setup, where the classical bit $e$ decides wether the tested supermap and the operations $\mathcal{A}_1, \mathcal{A}_2$ are applied or not, i.e.~$\hat{ \mathcal{E}} = \delta_{e,0} \,{\rm id} + \delta_{e,1}\,W[\mathcal{A}_1, \mathcal{A}_2]$, implies supermap rigidity. In the first approach~(a), where single black boxes are inserted into each slot, the reference (green) supermap $\Bar W$ can be obtained from the physical (red) supermap $W$ by some local injection and extraction maps on the input/outputs systems, and local embedding combs in each slots. In the second approach~(b), with stronger network connectivity assumptions, considering independent black boxes accessing the input and output of each slot, the reference (green) supermap $\Bar W$ can be obtained from the physical (red) supermap $W$ by some local injection and extraction maps on every system involved. As discussed in Sec.~\ref{sec:ST_quantum_supermaps}, the target probability distribution $P(\bm a,\bm b|\bm x, \bm y , e)$, with $\bm x=(x_1,x_2, x_P)$, $\bm y=(y_1,y_2, y_F)$ and similarly for $\bm a,\bm b$, should self-test the Choi channel $\mathcal{J}(\bar{W}):=\Bar W [\SWAP,\SWAP]$ of the reference supermap (see~\cref{fig:selftesting_channel}). In the second approach, this self-testing of the channel directly guarantees the rigidity on the supermap, while in the first approach, additional derivation -- using~\cref{lemma1} -- is needed to ensure that the physical channels $\mathcal{A}_1$, $\mathcal{A}_2$ indeed involve \SWAP operations.}
    \label{fig:supermaps}
\end{figure*}

%%%%%%%%%%%%%%%%%%
\subsection{First approach: Rigidity up to local embedding combs}
\label{sec:ST_approach1}
%%%%%%%%%%%%%%%%%%

Consider a reference experiment where a supermap $\Bar W$ is combined with a \SWAP gate in each slot, and the resulting Choi channel $\mathcal{J}(\Bar W)$ is self-tested (see examples in~\cref{sec:examples} below). For concreteness and simplicity, let us for now consider a one-slot supermap. By definition, the self-test of the associated Choi channel grants the following rigidity identity for any physical realization  of the supermap $W$ and of the bipartite channel $\mathcal{A}$ that is supposed to realize the \SWAP gate: 
\begin{align} \label{diag: Woneslot}
\begin{aligned}\includegraphics[height=2cm]{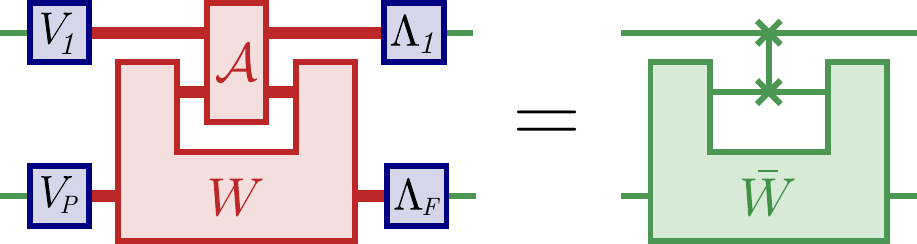} .
\end{aligned}
\end{align}

To further identify both of these operations, we need to again understand which physical realizations thereof are fundamentally indistinguishable from the reference one. To do so, it is enlightening to consider the following alternative physical realization: 
\begin{align}\label{eq: A tilde}
\begin{aligned}\includegraphics[height=1.5cm]{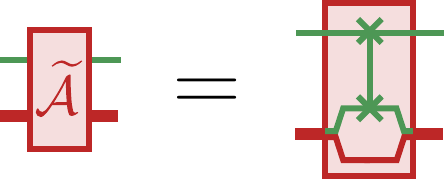}\ \ ,
\end{aligned}
\end{align}
\vspace{-0.2cm}
\begin{align}\label{eq: tilde W}
\begin{aligned}
\includegraphics[height=2.4cm]{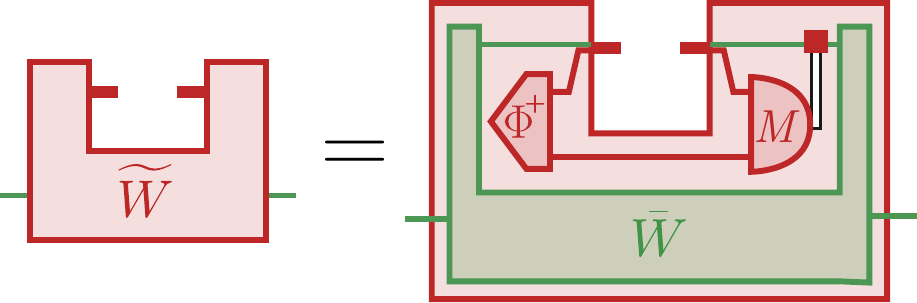}\ \ .
\end{aligned}
\end{align}
Here, $M$ denotes a binary projective measurement that checks if the state $\Phi^+$ is received unchanged, and feeds forwards the classical command to either transmit the state on the green wire, untouched, if $\Phi^+$ is indeed obtained, or replace it by a maximally mixed state otherwise. Combining these operations, one verifies that they indeed realize the same overall channel, i.e.~$\widetilde{W}[\widetilde{\mathcal A}] = \bar W[\SWAP]$, and thus satisfy the channel rigidity in Eq.~\eqref{diag: Woneslot}, with trivial extracting and injecting maps. Nevertheless, one 
also sees that there is in general no way to recover $\bar W$ from $\widetilde{W}$ by inserting local extraction/injection maps into the channel slot, i.e. 
\begin{align} \label{diag: Wid}
\begin{aligned}\includegraphics[height=1.5cm]{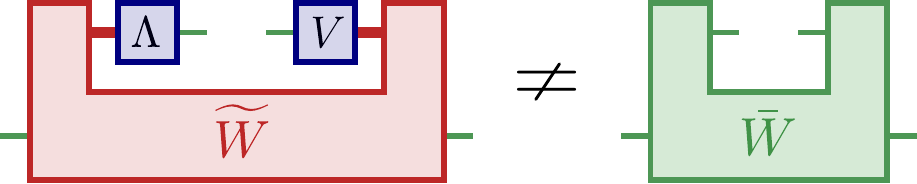}
\end{aligned}\, ,
\end{align}
meaning that no blue boxes acheiving the equality exist.
This is because in general the reference system output by $\Lambda$ (green wire) can not contain both of its inputs: the first output of $\bar W$ (green wire) and half of the entangled state $\Phi^+$ (red wire). Any such compression of two systems into one (of a smaller dimension) performed by~$\Lambda$ is a non-invertible operation, which cannot be undone by $V$.
It follows that, when appended with the maps $\Lambda$ and $V$, the supermap $\widetilde W$ will fail to imitate $\bar W$ in general.

This example shows that the channel rigidity of Eq.~\eqref{diag: Woneslot} is on its own insufficient to guarantee that the physical supermap $W$ can be transformed to the target one $\bar W$ by appending all wires with local extraction and injection maps. 
Moreover, note that this tension cannot be resolved by extending the experiment and combining the reference supermap also with other reference channels $\bar{\mathcal{A}}'$ different from the \SWAP gate. Indeed, the physical realization $\widetilde W$ in Eq.~\eqref{eq: tilde W} remains indistinguishable from the target one $\widetilde W[\tilde{\mathcal{A}}']=\bar W[\bar{\mathcal{A}}']$ upon also adjusting the physical realization of the channels $\tilde{\mathcal{A}}'$ accordingly to Eq.~\eqref{eq: A tilde}.

A way to capture the degrees of freedom discussed above is to allow for identification of supermaps that include some memory effect in each slot. Indeed, from any realization of the form of \cref{eq: tilde W} one can recover the reference supermap $\Bar W$ by plugging the following local comb into the slot:
\begin{align}
\begin{aligned}
    \includegraphics[height=1 cm]{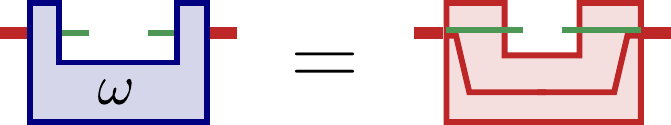}\ \ . 
\end{aligned}
\end{align}
Hence, in the present scenario, this observation supports the fact that the best kind of rigidity statement that one can hope for is to identify the existence of a 1-comb $\omega$ -- that we call an \emph{embedding comb} -- which, composed with the physical supermap $W$, recovers the reference supermap $\bar{W}$. In the case of a process with two slots, this corresponds to the following diagram
\begin{align}
     \begin{aligned}\label{eq:st up to w}
           \includegraphics[height=3.1 cm]{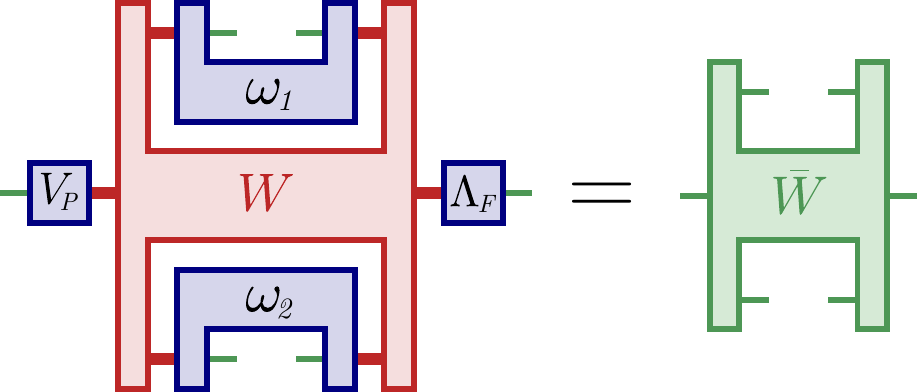} \ ,
       \end{aligned}
\end{align}
which defines the notion of {\bf supermap rigidity up to local embedding combs}. We emphasize that, while depicted here for a supermap with two slots, one input and one output, this form of rigidity naturally generalizes to arbitrary supermaps. 

Remark that any one-slot comb $\bar W$ with a single input and output can be realized by inserting an embedding comb $\omega$ into the identity comb $W_{\rm Id}$, which maps any CP map into itself: 
\begin{align} \label{eq:identityrefcomb}
    \begin{aligned}
        \includegraphics[height=2cm]{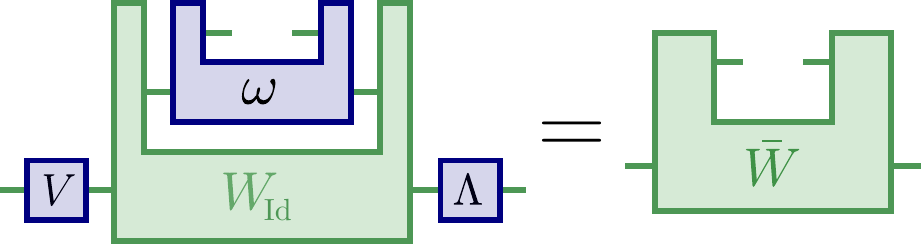} ,
    \end{aligned}
\end{align}
as one can always chose $\omega = \bar W$ and $V=\Lambda={\rm id}$. In this sense $W_{\rm Id}$ can be seen as the reference 1-comb. A practical example of this process is a transducer that transforms an input state into a form, encoded onto different degrees of freedom, that is suitable for the application of an arbitrary gate or channel (in the slot), and transforms the result back. We consider its self-testing in Sec.~\ref{sec: 1-comb}. 
This is analogous to the case of single-system channels, and the self-testing of the identity channel discussed at the end of Sec.~\ref{sec:self-test_channels}. 

Having set the goal, we are left with a natural question:  does the rigidity of the Choi channel in Eq.~\eqref{diag: Choi chanel2} always imply the rigidity of the supermap in the form of Eq.~(\ref{eq:st up to w})? While we leave this general question hanging for future works, we still provide a partial yet helpful result along these lines in the following lemma, where
the ground symbol indicates that the system is discarded.

\setcounter{theorem}{0}

\begin{restatable}{lemma}{LemmaOne}\label{lemma1}
    For any bipartite channel $\mathcal{X}$ and any one-slot comb $\Omega$ such that  
\begin{align}\label{eq: lem11}
\begin{aligned}\includegraphics[height=2cm]{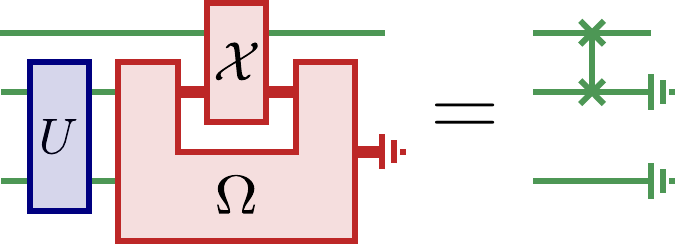}\ \ ,
\end{aligned}
\end{align}
    it holds true that
\begin{align}\label{eq: lem12}
\begin{aligned}\includegraphics[height=2.5cm]{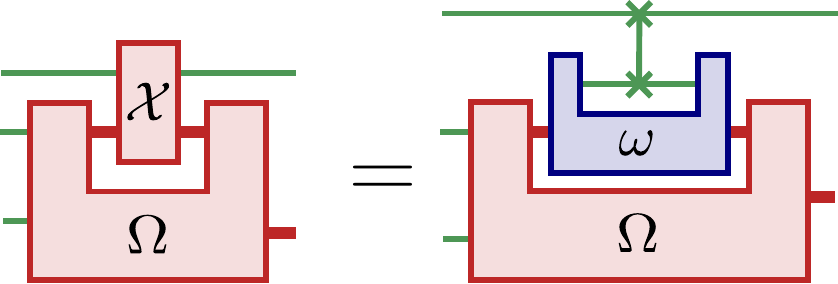}\ \ .
\end{aligned}
\end{align}
\end{restatable}

\noindent The proof of the lemma, presented in Appendix~\ref{app:lemmas}, is based on the following intuition. Eq.~\eqref{eq: lem11} states that the channel on the left hand side perfectly transfers the middle input system to the upper output system. But since the only causal connection between these two systems goes through the channel $\mathcal{X}$, then $\mathcal{X}$ must somehow contain a \SWAP gate, i.e.~it must guarantee the decomposition of this channel in Eq.~\eqref{eq: lem12} for some  $\omega$. 

As we will see with the examples of sections~\ref{sec: 1-comb} and~\ref{sec:self_test_QS},  the lemma is very useful for translating between the rigidity statements in Eqs.~\eqref{diag: Choi chanel2} and~\eqref{eq:st up to w}. This is easier to understand through the examples, however the general idea can be explained already. Consider a Choi channel $\mathcal{J}(\bar{W})$ in Eq.~\eqref{diag: Choi chanel2}, identified by the local injection and extraction maps. Suppose also that with a unitary channel $U$ on the lower two input systems, one can identify a subsystem of these, which the channel $\mathcal{J}(\bar{W})\circ U$ swaps to the upper output system. This condition is equivalent to the premise of the \cref{lemma1} (Eq.~\ref{eq: lem11}). Hence, the lemma guarantees that the channel $\cA_1$ inserted in the upper slot of the process, can be replaced by $\omega_1[\SWAP]$ for some embedding comb $\omega_1$.
This is precisely the logical step required to go from the rigidity statement of Eq.~\eqref{diag: Choi chanel2} to that of Eq.~\eqref{eq:st up to w}; the argument can then be repeated for all slots.

%%%%%%%%%%%%%%%%%%
\subsection{Second approach: Stronger rigidity exploiting network connectivity assumptions}
\label{sec:ST_approach2}
%%%%%%%%%%%%%%%%%%

In the previous subsection we have formalized the notion of a supermap rigidity up to local embedding combs. We argued that this is the best we can hope for in  experiments where the supermap is combined with arbitrary channels and each channel is treated as a single black box. 

We will now show that a sharper notion of supermap self-testing can be achieved  in scenarios where the {\it network connectivity} of the underlying experiment allows one to identify multiple black boxes interacting independently with the input and output systems of each slot.
Indeed, recall that to deduce a self-test as in Eq.~\eqref{diag: Choi chanel2} of the Choi channel associated with a supermap, the latter is combined with bipartite \SWAP gates in all of its slots. A \SWAP gate is not a generic bipartite channel, but can be simply viewed as a specific {\it wiring} between the input and output systems, to the same extent as a bipartite identity channel is represented by two independent wires in the network diagram. Thus, it is in practice possible to devise an experiment whose network structure can be safely assumed to satisfy the independence of the wires inside the \SWAP gates. In this case, one starts with a guarantee that all these operations have the following physical realization: 
\begin{align}\label{eq: SWAP wiring}
\begin{aligned}\includegraphics[height=1.25
cm]{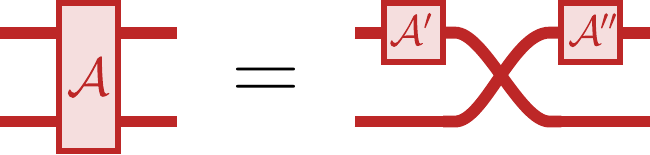}\ \ ,
\end{aligned}
\end{align}
where $\mathcal{A}'$ and $\mathcal{A}''$ identify the uncharacterized parts of channel $\mathcal{A}$, i.e.~they are black boxes.
With this identity in hand, we can immediately rewrite the rigidity of the Choi channel of a generic supermap, Eq.~\eqref{diag: Choi chanel2}, in the following form:
\begin{align}
\begin{aligned}\includegraphics[height=3.75
cm]{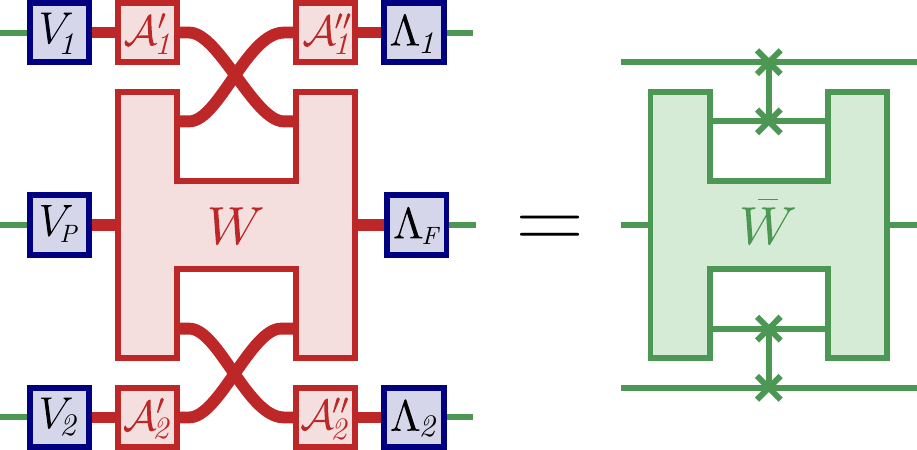}\ .
\end{aligned}
\end{align}
Next, by absorbing the unknown channels $\mathcal{A}_{1(2)}'$ and $\mathcal{A}_{1(2)}''$ into the local injection  $V_{1(2)}$ and extraction $\Lambda_{1(2)}$ maps and invoking once again the supermap-channel isomorphism of \cref{diag: Choi chanel}, we immediately obtain
\begin{equation}\label{eq: Rig up to maps}
       \begin{aligned}
           \includegraphics[height=2.2
cm]{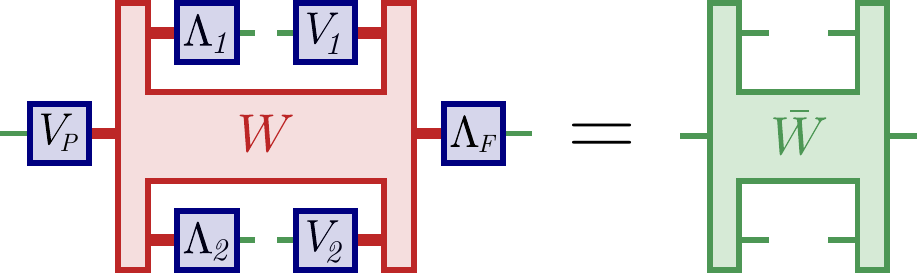} \ ,
       \end{aligned}
  \end{equation}
defining now the notion of {\bf supermap rigidity up to local maps}, in direct analogy with the case of states, measurements and channels we discussed in the background section. Remark again that this naturally generalizes to arbitrary supermaps. 

In summary, we have demonstrated the following general implication: if the network connectivity underlying the experiment is such that the physical realization of the \SWAP gates can be assumed to have the wiring structure of Eq.~\eqref{eq: SWAP wiring}, the rigidity of the Choi channel in Eq.~\eqref{diag: Choi chanel2} automatically implies the rigidity of the supermap  up to local maps in Eq.~\eqref{eq: Rig up to maps}. Thus, in this setting the task of supermap self-testing reduces to that of channel self-testing. 

\section{Examples} \label{sec:examples}

In this section, we present four examples of self-testing statements of quantum supermaps. We begin with the certification of one-combs, showing in~\cref{sec: 1-comb} how the identity comb can be certified using our first approach, up to a local embedding comb and proceeding in~\cref{sec:ECcomb} to self-test a non-trivial one-comb -- a bit-flip error correcting comb -- using our second approach, up to local maps. We then show in~\cref{sec:self_test_Grover} how to self-test (again using our second approach) an $M$-slot comb corresponding to Grover's algorithm. Finally, in~\cref{sec:self_test_QS}, we show, using both approaches, how to self-test the quantum switch, a two-slot comb that allows for coherent superposition of ordering of two quantum channels. \\

\subsection{Self-testing the identity 1-comb up to local embedding combs}

\label{sec: 1-comb}

As the first example of supermap self-testing, we consider the simple ``identity 1-comb'' $\Bar{W}_{\text{Id}}$, which transforms any channel into itself $\Bar{W}_{\text{Id}}[\mathcal{A}] = \mathcal{A}$, and is thus associated with the Choi channel $\mathcal{J}(\bar{W}_{\text{Id}}) = \SWAP$, i.e.
\begin{align} \label{eq: choi swap 1}
\begin{aligned}\includegraphics[height=2cm]{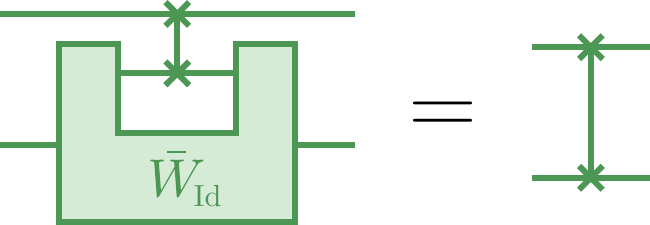} \ \  .
\end{aligned}
\end{align}
As discussed in~\cref{sec:ST_approach1}, all 1-slot combs with single input and output registers can be prepared from the identity comb, making it worth consideration.

A self-test of the \SWAP channel can be easily constructed by combining known self-tests of maximally entangled states, as detailed in Appendix~\ref{app:self_testing_SWAP}. Indeed, as discussed in~\cref{sec:self-test_channels}, a rigidity statement of this channel can be obtained by showing that both Eqs.~(\ref{selftesting_channels_step1}) and (\ref{selftesting_channels_step3}) hold with the same extraction maps on the ``channel-free'' systems. Concretely, to self-test the \SWAP channel it suffices to combine self-tests of the reference Choi state $\mathcal{C}(\SWAP)$ and $(\Phi^+)^{\otimes 2}$, i.e. 
\begin{align}
\begin{aligned}
    \includegraphics[height=2.25cm]{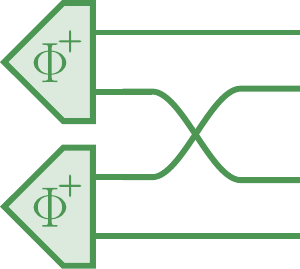} \  \quad \text{and}\quad \ \includegraphics[height=2.25cm]{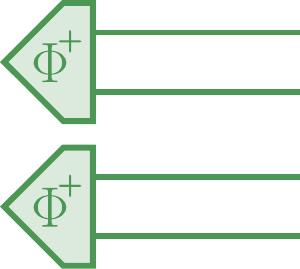}\ \ , 
\end{aligned}
\end{align}
obtained with the same measurements on the external wires. For instance, in the case where green wires represent qubits, this can be done by combining CHSH self-tests~\cite{Kaniewski16} of the maximally entangled two-qubit states on the appropriate systems. Using the same tests guarantees that the extraction maps are identical in both. For arbitrary dimension, see self-tests of maximally entangled qudits provided in e.g.~\cite{coladangelo2016parallel,Sarkar21,Meyer25}.

This leaves us with the following identity valid for any physical realization of the experiment: 
\begin{align}\label{st channel Wid}
\begin{aligned}\includegraphics[height=2cm]{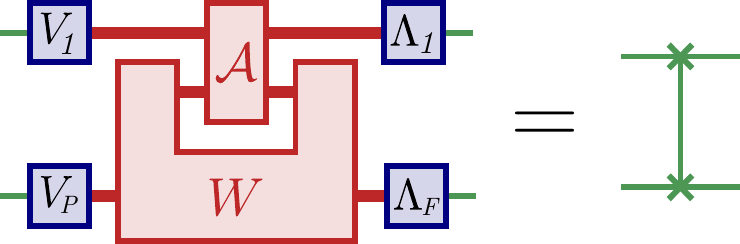}\ \  .
\end{aligned}
\end{align}
On the left hand side, let us apply the following identification of the operations $\mathcal{X}$ and $\Omega$
\begin{align}\label{eq: decomp Omega}
\begin{aligned}\includegraphics[height=3cm]{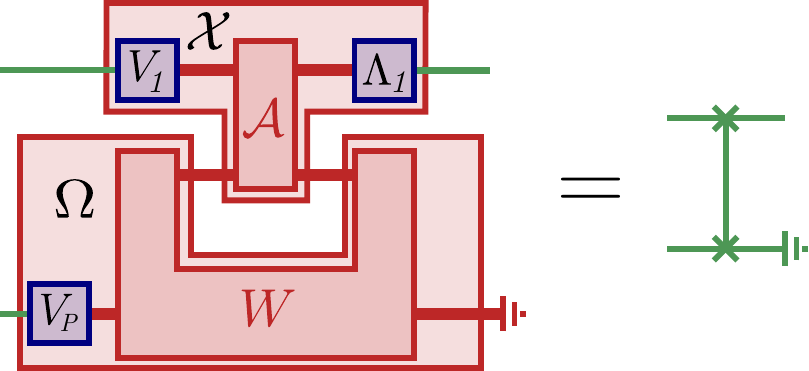}\ \ ,
\end{aligned}
\end{align}
which allows us to apply Lemma~\ref{lemma1} (with a trivial lower system and a trivial $U$) in order to obtain that~\cref{eq: lem12} holds. Plugging this back into \cref{st channel Wid}, and using \cref{eq: choi swap 1}, we obtain
\begin{align}
\begin{aligned}\includegraphics[height=2.25
cm]{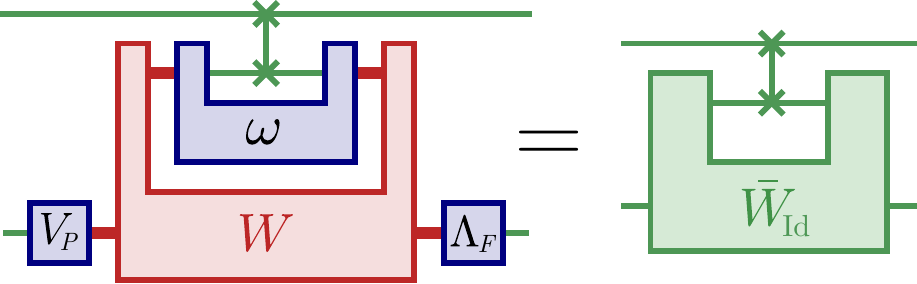} \ \ ,
\end{aligned}
\end{align}
which is equivalent to the rigidity statement given in Eq.~\eqref{eq:st up to w} by virtue of the supermap-channel isomorphism of \cref{diag: Choi chanel}, as desired.

\subsection{Self-testing an error-correcting comb up to local maps}
\label{sec:ECcomb}

As a second example of supermap self-testing, we consider the error-correction comb $\bar W_{\rm EC}$ defined as
\begin{align}\label{eq:EC_comb}
\begin{aligned}
    \includegraphics[height=2cm]{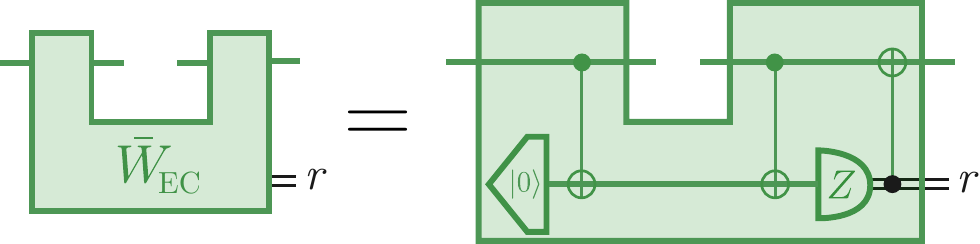} \ \ ,    
\end{aligned}
\end{align}
which also outputs a binary classical register $r\in\{0,1\}$. The comb $\bar W_{\rm EC}$ detects if a bit-flip error ($X$-type) occurs on the transmitted qubit, corrects the error if it occurs, and outputs the value $r$ indicating if the correction was performed. This comb plays an important role, for instance in error-correction assisted quantum metrology~\cite{dur2014improved,sekatski2017quantum} and in channel cloning~\cite{sekatski2025cloning}.

The Choi channel of the error-correcting comb $\mathcal{J}(\Bar{W}_{\text{EC}})=\Bar{W}_{\text{EC}}[\SWAP]$ corresponds to a quantum instrument, due to the classical register $r$. Accordingly, the task is to combine the self-tests of the quantum-classical state $\mathcal{C}(\mathcal{J}(\Bar{W}_{\text{EC}}))$ of four qubits plus a classical register, and of the four-qubit state $(\Phi^+)^{\otimes 2}$, prepared in the experimental rounds where the comb was not applied. With straightforward algebra one finds these states to be 
\begin{align}
\begin{aligned}
    \includegraphics[height=2.25cm]{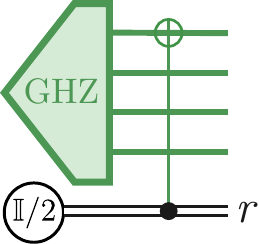}\  \quad \text{and}\quad \includegraphics[height=2.25cm]{FigSec4/1comb_state2.pdf} \ \ ,
\end{aligned}
\end{align}
where on the left the classical register is prepared as the uniform mixture $\mathds{I}/2=\frac{1}{2}(\ketbra{0}+\ketbra{1})$, the four qubits start in the state ${\rm GHZ}=\frac{1}{2}\sum_{i,j=0,1}\ket{iiii}\bra{jjjj}$ and the bit-flip gate, controlled on the classical register, acts on the other half of the singlet sent into the slot.
In Appendix~\ref{app:self_testing_EC_comb}, we show that the state $\mathcal{C}(\mathcal{J}(\Bar{W}_{\text{EC}}))$ can be self-tested using Bell inequalities tailored to GHZ states introduced in~\cite{Barizien24} and a symmetry argument to take into account the extra bit-flip gate in the $r=1$ branch. Crucially, this can be done while keeping the same measurements, and thus extraction maps, on the ``channel-free'' (external) systems as for the CHSH tests that self-test $(\Phi^+)^{\otimes 2}$.

This allows us to obtain a rigidity statement, in the form of ~\cref{selftesting_channels}, for the Choi channel $\mathcal{J}(\Bar{W}_{\text{EC}})$ of the EC-comb. Under the strong network connectivity assumption, this is sufficient to provide a rigidity statement of the form~\cref{eq: Rig up to maps}, i.e.
\begin{align}
    \begin{aligned}
        \includegraphics[height=1.5cm]{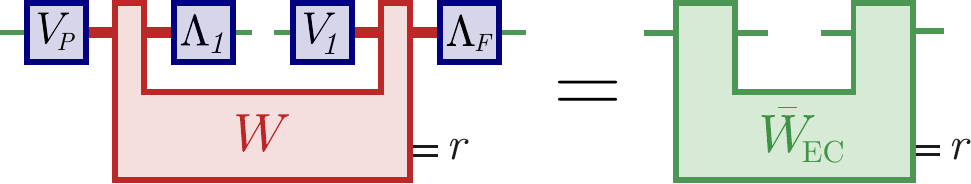} \ .
    \end{aligned} 
\end{align}

%%%%%%%%%%%%%%%%%%%%%%%%%%%%%%%%%%%%%%%%%%%%%%%%%%%%%%%%
\subsection{Self-testing Grover's algorithm}
\label{sec:self_test_Grover}
%%%%%%%%%%%%%%%%%%%%%%%%%%%%%%%%%%%%%%%%%%%%%%%%%%%%%%%%

Grover's algorithm~\cite{grover1996fast} is a quantum search algorithm that provides the optimal quadratic speedup~\cite{zalka1999grover} over classical search methods for unstructured databases. For a database with $N=2^n$ elements the algorithm operates on a quantum systems consisting of $n$ qubits. It is an oracle algorithm, where each of the $M=O(\sqrt{N})$ queries of the database is described by the application of an unknown, but fixed, unitary channel $U_{\bm x}[\cdot]= (\id - 2\ketbra{\bm x})\, \cdot \, (\id - 2\ketbra{\bm x})$ marking the $n$-bitstring $\bm x$ that has to be found. After $M$ queries of the unitary $U_{\bm x}$, the algorithm outputs an estimator $\hat {\bm x}$  satisfying ${\rm Pr}(\hat {\bm x}=\bm x)=1-O(N^{-1})$. Grover's algorithm can be represented by a quantum supermap, or comb, with $M$ slots and a classical output $\hat {\bm x}=\hat x_1 \hat x_2\dots \hat x_n$
\begin{align} 
\begin{aligned}
   \includegraphics[height=1.9cm]{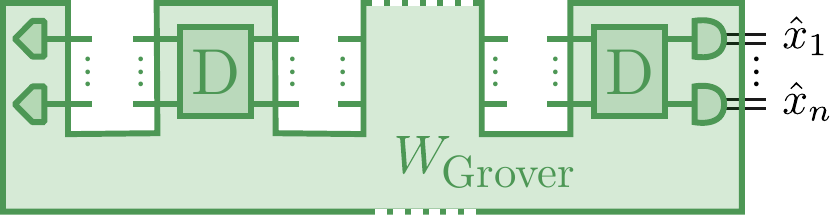} .
\end{aligned}
\end{align}
Here the qubit sources each prepare the state $\ketbra{+}{+}=\frac12(\ket{0}+\ket{1})(\bra{0}+\bra{1})$, the final measurements are in the computational basis $\{\ket{0},\ket{1}\}$, and 
\begin{equation}\label{eq: diffusion}
{\rm D}[\,\cdot\,] =(2 \ketbra{+}^{\otimes n}-\id) \cdot (2 \ketbra{+}^{\otimes n}-\id)
\end{equation}
is the Grover diffusion operation (a unitary $n$-qubit channel) applied after each slot.

To self-test Grover's algorithm, we self-test the reference comb $\bar W$ which corresponds to $W_{\rm Grover}$, where the input sources, respectively the final measurements, have been moved outside, and replaced by $n$ input, respectively $n$ output, qubit wires. Its physical realization $W$, depicted in red in Eq.~\eqref{eq: grover rigid},  has $n$ inputs, $M$ slots with $n$ inputs and outputs each, and $n$ outputs.  For simplicity, we consider its self-testing under the stronger network connectivity assumptions discussed in Sec.~\ref{sec:ST_approach2}. More precisely,  we assume that all the $2n$-system channels $\mathcal{A}_i$ inserted in each slot of the comb are wired like $n$ independent \SWAP operations
\begin{align}\label{eq: grover wiring}
\begin{aligned}
   \includegraphics[height=2.2cm]{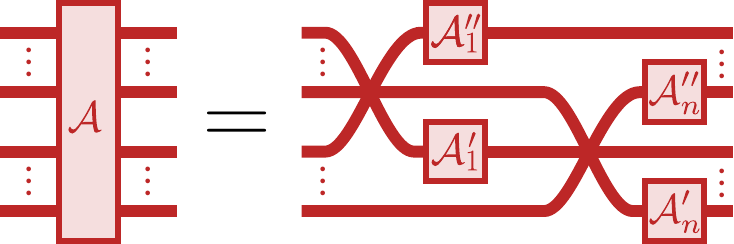}.
\end{aligned}
\end{align}

Following the discussion of Sec.~\ref{sec:ST_approach2}, we know that the rigidity of the supermap $W$ up to local maps 
\begin{widetext}
    \begin{align}\label{eq: grover rigid}
\begin{aligned}
   \includegraphics[width=\textwidth]{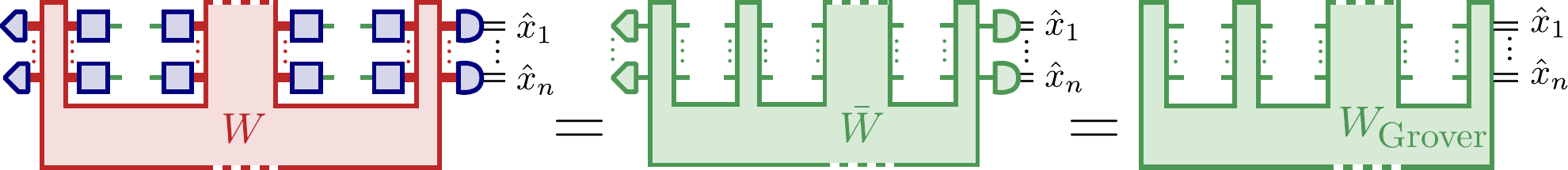}
\end{aligned}
\end{align}
\end{widetext}
 can be guaranteed, if one self-tests the Choi channel $\mathcal{J}(\bar W)$ obtained by combining the comb with $n \times M$ qubit \SWAP gates. To do so, notice that the target comb $\bar W$ has no memory and applies the same $n$-qubit Grover diffusion channel $\rm D$  after every slot. Therefore, its Choi channel $\mathcal{J}(\bar W)$ is product:
\begin{align}\label{eq: grover Choi}
\begin{aligned}
    \includegraphics[height=3cm]{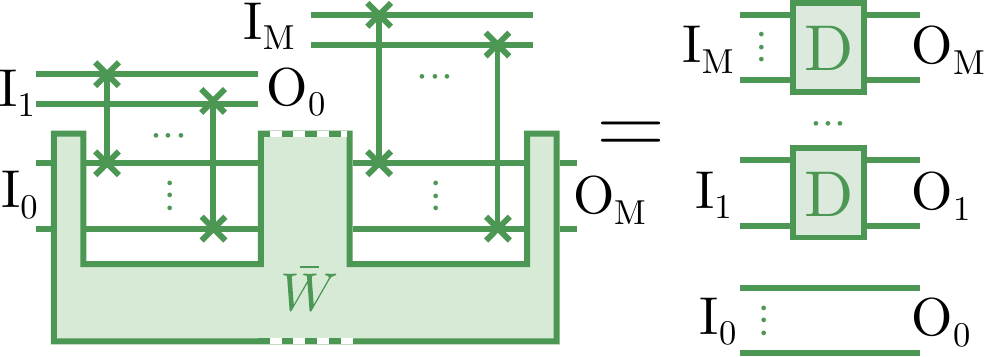}.
\end{aligned}
\end{align}
Consequently the Choi state $\mathcal{C}(\mathcal{J}(\Bar{W}))$ is also the product of the state 
\begin{align}
\begin{aligned}
    \includegraphics[height=2cm]{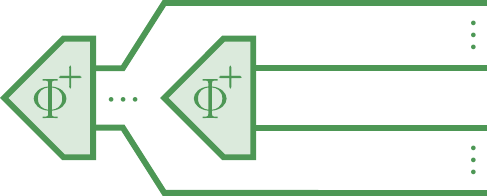}
\end{aligned}
\end{align}
corresponding to the wires $\rm{I}_0\rm{O}_0$ left untouched by $\mathcal{J}(\bar W)$, and $M$ states 
\begin{align}
\begin{aligned}
    \includegraphics[height=2cm]{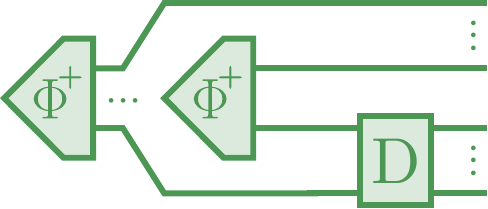}
\end{aligned}
\end{align}
 corresponding to the wires $\rm{I}_i\rm{O}_i$ with $i=1,\dots,M$. The remaining challenge is thus to self-test the pair of $2n$-qubit states $(\Phi^{+})^{\otimes n}$ and $({\rm id}\otimes {\rm D})[(\Phi^{+})^{\otimes n}]$ using the same measurements on the auxiliary systems. In  Appendix~\ref{sec:STGrover}, we show that this can be done based on the maximal violation of appropriate Bell inequalities. The former utilizes parallel self-testing with the CHSH inequality~\cite{Supic20}; while for the latter we construct new Bell inequalities tailored to $({\rm id}\otimes {\rm D})[(\Phi^{+})^{\otimes n}]$ by leveraging the method introduced in~\cite{Barizien24}. This method relies on finding a complete set of nullifying operators for the target state, which can be constructed here by applying the unitary channel ${\rm id}\otimes {\rm D}$ to nullifying operators of $(\Phi^{+})^{\otimes n}$.

%%%%%%%%%%%%%%%%%%%%%%%%%%%%%%%%%%%%%%%%%%%%%%%%%%%%%%%%
\subsection{Self-testing the quantum switch}
\label{sec:self_test_QS}
%%%%%%%%%%%%%%%%%%%%%%%%%%%%%%%%%%%%%%%%%%%%%%%%%%%%%%%%

The quantum switch is a two-slot quantum supermap which (in the version we consider here) has two quantum inputs and two quantum outputs~\cite{Chiribella13}. These two systems at both the input and output are referred to as the target ($T$) and control ($C$) systems. The latter is a qubit that controls the order in which the operations $\mathcal{A}_1$ and $\mathcal{A}_2$ plugged into the two slots are to be performed on the target system entering in the state $\rho_T$: if  the control qubit is in state $\ketbra{0}_C$, the output target system is $(\mathcal{A}_1 \circ \mathcal{A}_2)[\rho_T]$, while if the control is in state $\ketbra{1}_C$, the output target system is $(\mathcal{A}_2 \circ \mathcal{A}_1)[\rho_T]$. Crucially, the control qubit can be in a superposition, thus leading to a superposition of orders in which the operations are applied.  

The capabilities of the quantum switch as a resource in various tasks~\cite{chiribella12,araujo14,costa26}, as well as recent debate about the status of physical implementations of quantum switch~\cite{rozema24}, means that certifying the causal indefiniteness of the quantum switch is an important problem. In contrast to some causally indefinite supermaps~\cite{Oreshkov12}, however, the correlations generated by the parties in the quantum switch cannot violate any so-called ``causal inequality''~\cite{Branciard16,oreshkov16} and can hence be given a causal description. As a result, approaches to certifying its causal indefiniteness require some additional assumptions, typically by relaxing the black-box assumption (at least partially) of the parties' operations~\cite{araujo15,branciard16a,bavaresco19} or, e.g., providing them with trusted quantum, instead of classical, inputs~\cite{dourdent22}.

Recently, however, it was shown that one can certify the causal indefiniteness of the quantum switch device-independently if one assumes a certain network structure. In particular, by considering additional parties who are assumed to be non-signaling with respect to the parties in the quantum switch, one can violate some inequalities derived under assumptions that the causal (signaling) structure is determined by a hidden variable obeying this network structure~\cite{gogioso23,van-der-lugt23,van-der-lugt23a}. A conceptually related but different type of device-independent (but theory dependent) certification was obtained in Ref.~\cite{Dourdent23} by self-testing the quantum inputs used in the semi device-independent certification of~\cite{dourdent22}.

Here we propose to self-test directly the quantum switch (with a qubit target system), using the approaches we introduced for general supermaps. This not only guarantees that a test process  must be causally indefinite, but further implies that it must be precisely identical to the quantum switch, at least up to a restricted set of local operations.

For this, we start by self-testing the quantum switch's Choi channel, obtained by plugging a \SWAP into each slot. The latter channel
\begin{align}\label{eq: choi swap}
    \begin{aligned}
        \includegraphics[height=3cm]{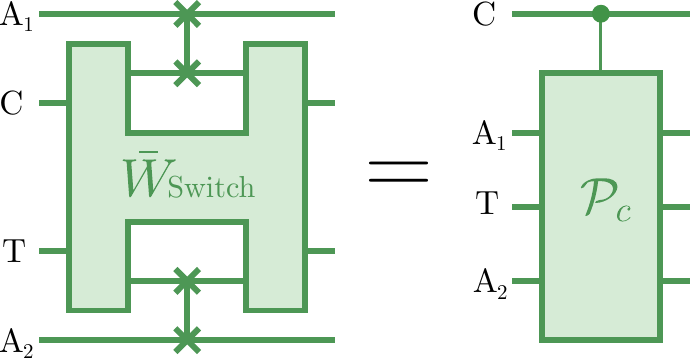} \ \ 
    \end{aligned}
\end{align}
is a permutation $\mathcal{P}_c$ of the $A_1TA_2$ systems, coherently controlled by the value of the control qubit $\ket{c}_C$. More precisely, one finds $\mathcal{P}_0: \ket{ijk}_{A_1TA_2} \mapsto \ket{jki}_{A_1TA_2}$ and $\mathcal{P}_1: \ket{ijk}_{A_1TA_2} \mapsto \ket{kij}_{A_1TA_2}$. 

For concreteness, let us now choose the systems $A_1$, $T$ and $A_2$ as qubits, and show how to derive the rigidity of the Choi channel in Eq.~\eqref{eq: choi swap}, by combining the self-tests of the states $(\Phi^+)^{\otimes 4}$ and $\mathcal{C}(\mathcal{J}(\Bar W_{\text{Switch}}))$. To self-test the four Bell pairs it is sufficient to combine four CHSH self-tests. To self-test the Choi state of the channel, remark that when the control qubit is measured in the computational basis, revealing the values $c$, the channel $\tr_C \left(\ketbra{c}_C \mathcal{J}(\Bar W_{\text{Switch}})\right)$ realizes a fixed permutation $\mathcal{P}_c$ of the systems $A_1TA_2$. The Choi states for these permutations can also be self-tested by combining three CHSH tests between appropriate systems. The remaining challenge is to show that the two permutations are performed in a superposition coherently controlled by the qubit $C$. This can be done by performing another CHSH test of the control qubit with its Choi partner, conditional on the systems $A_1TA_2$ being projected on the same states, thus erasing the which-permutation information; see Appendix~\ref{app:self_testing_channel} for the full details.

From the rigidity of the Choi channel, one can get that
\begin{align}\label{eq:STJC_Switch}
    \begin{aligned}
        \includegraphics[height=3.25cm]{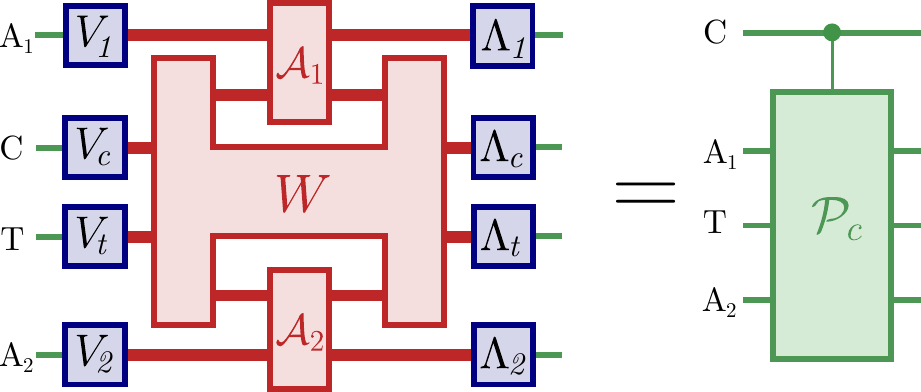}\ . 
    \end{aligned}
\end{align}
Again, we argued in~\cref{sec:ST_approach2} that this is enough to get a rigidity statement \cref{eq: Rig up to maps} up to local maps when leveraging network connectivity assumptions and considering independent access to the input and output of each slot, i.e.~that
\begin{align}
    \begin{aligned}
        \includegraphics[height=2.25cm]{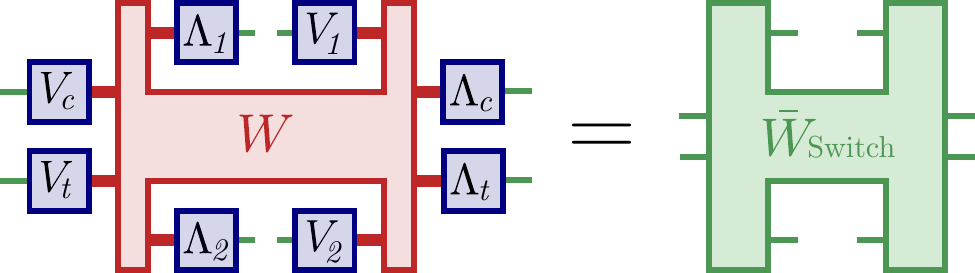} \ .
    \end{aligned}
\end{align}

In the one-box per slot approach, a rigidity statement can still be obtained. Starting from~\cref{eq:STJC_Switch}, our goal is to apply \cref{lemma1} in order to prove the rigidity up to local embedding combs. For this, we first need to identify a unitary $U$ that allows us to extract a \textsc{Swap}: this is done by noting that
\begin{align}
    \begin{aligned}
        \includegraphics[height=3cm]{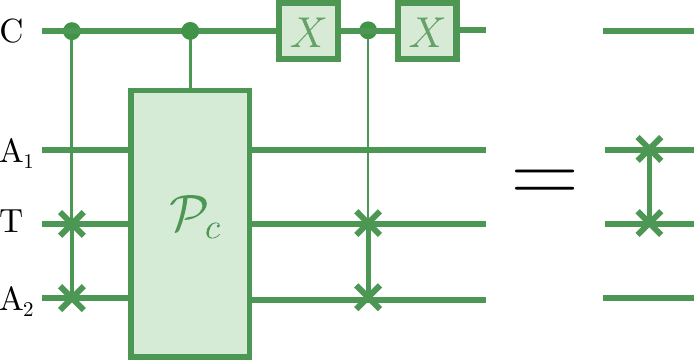} \ \ .
    \end{aligned}
\end{align}
Thus, applying the controlled-\SWAP gate  on the input systems $CTA_2$ in the Choi channel rigidity (Eq.~\ref{eq:STJC_Switch}), allows us to get the  following identity
\begin{align}
    \begin{aligned}
        \includegraphics[height=4.15cm]{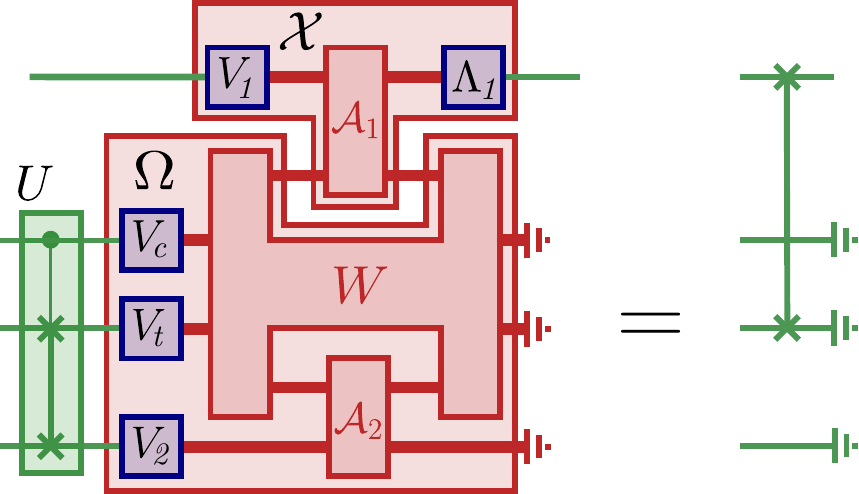} .
    \end{aligned}
\end{align}
We can now  use \cref{lemma1} to certify that, when inserted in the comb slot, the channel $\mathcal{A}_1$ realizes a \SWAP gate up to some local embedding comb $\omega_1$. That is, we get the identity
\begin{align}
    \begin{aligned}
        \includegraphics[height=3.65cm]{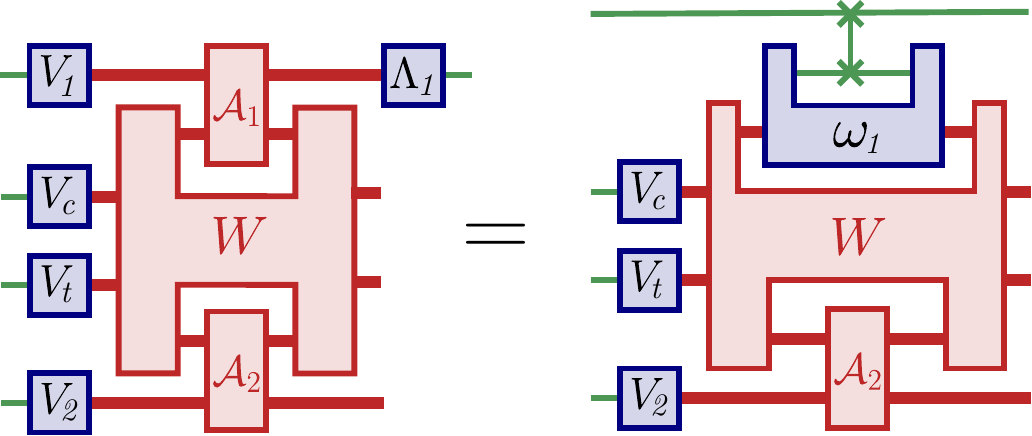}.
    \end{aligned}
\end{align}
By symmetry, one can conduct the exact same reasoning on the channel $\mathcal{A}_2$ inserted in the second slot, to finally obtain the desired identity 
\begin{align}\label{eq: ST switch comb final}
    \begin{aligned}
        \includegraphics[height=3.8cm]{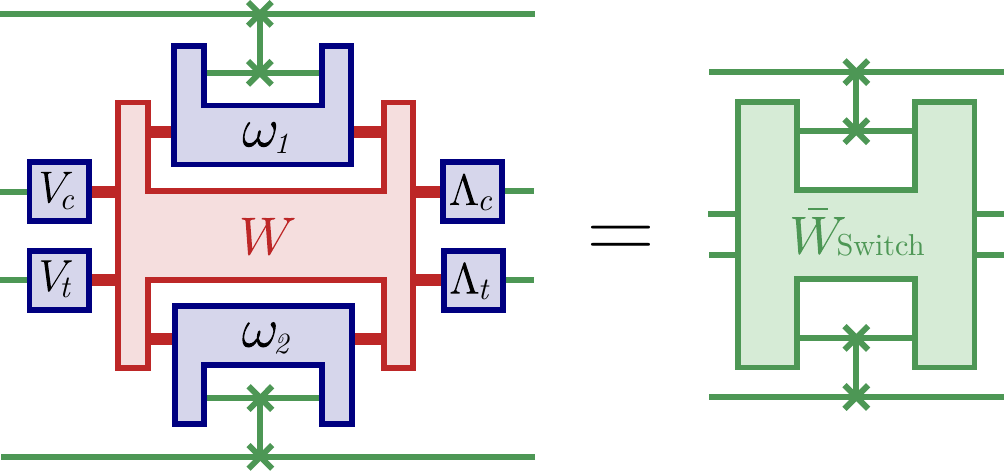} \ .
    \end{aligned}
\end{align}
Due to the supermap-channel isomorphism of \cref{diag: Choi chanel}, this is exactly the rigidity statement of \cref{eq:st up to w}, up to local embedding combs. In particular, Eq.~\eqref{eq: ST switch comb final} implies that the physical process  $W$ is causally indefinite. Indeed, local operations and embedding combs preserve causal definiteness; if $W$ on the lhs above was causally separable, then so would its composition with the blue operations, and the resulting supermap could not possibly be equal to the quantum switch. 

In conclusion, our results allow one to self-test the quantum switch under the network connectivity assumptions, i.e.~if one assumes that the experiment is well described by either of the diagrams in Fig.~\ref{fig:supermaps}. We emphasize that this assumption is a prerequisite for self-testing, which can not be used to deduce the network connectivity of the experiment in the first place. In fact, interpreting the network connectivity of experiments aiming to implement non-causal processes, such as the quantum switch, can be particularly subtle. These experiments (see e.g.~\cite{rozema24,costa26} for reviews) are typically based on photon interference. There, indefinite causal order is found when quantum systems are associated with degrees of freedom carried by a single photon, such as polarization and path.
However, when quantum systems are associated with traveling optical modes, whose state space includes the vacuum state, such experiments admit a quantum circuit description without causal indefiniteness. The latter picture is also consistent with a causally ordered notion of spacetime~\cite{vilasini2024fundamental}. Our self-testing results are not concerned by this debate, nor do they seek to resolve it.

%%%%%%%%%%%%%%%%%%%%%%%%%%%%%%%%%%%%%%%%%%%%%%%%%%%%%%%%
\section{Discussion}
\label{sec:discussion}
%%%%%%%%%%%%%%%%%%%%%%%%%%%%%%%%%%%%%%%%%%%%%%%%%%%%%%%%

In this paper, we showed how self-testing of a quantum supermap can be obtained by self-testing its Choi channel, obtained by plugging \SWAP gates in each of the supermap's slots. Specifically, we introduced two approaches leading to different kinds of rigidity statements for the supermap, depending on the network connectivity of the experiment, as pictured in~\cref{fig:supermaps}. When a single unstructured physical box is inserted in each slot, rigidity can be obtained only up to local embedding combs. 
When independent boxes interact with the input and output systems of each slot and are wired so as to realize \SWAP operations, the rigidity statements can be strengthened up to local maps.

To illustrate these notions we presented a series of examples. We first showed how to self-test the identity 1-comb, mapping each CP map onto itself, and the error-correcting 1-comb, detecting and correcting bit-flip errors. Then, we constructed a self-test for the $M$-slot comb describing Grover's search algorithm. Finally, we showed that it is also possible to self-test the quantum switch, a supermap which implements the operations inserted in the slots in a coherent superposition of causal orders.

These results highlight several fundamental aspects of the device-independent program and raise a number of open questions. 

First, our analysis underlines the central role of network connectivity in self-testing, i.e.~the causal structure of the underlying experiment. This is not specific to supermaps since the self-testing of quantum states, measurement and channel also assume a particular network structure. For instance, it is crucial in all cases that no hidden communication violate the assumed network structure.  The independence of all sources preparing either the quantum states, or of the classical measurement settings, are equally important. In the case of supermaps, the network structure may include a condition on the wiring of the channels inserted in the supermap. We underline that such assumptions on the causal structure of the experiment are not a limitation but a founding feature of any device-independent analysis of experiments, and are already present in Bell's pioneering analysis of bipartite correlations~\cite{bell1964einstein}.

Next, our discussion on self-testing of quantum channels and supermaps naturally raises the question of the generality of the proposed approach. While we believe that it could in principle provide self-tests for all extremal channels and supermaps, we leave this general question open. However, we note that there seems to be a clear path towards a provable self-testing for all channels and supermaps operating on qubits. Indeed, the possibility to self-test all pure multipartite qubit states has been established recently~\cite{balanzo2026all}, which suggests that the Choi state of any extremal multi-qubit channel is similarly certifiable. Additionally, the ad-hoc method, which allows adding new measurements to an existing self-test in order to self-test them together with other elements of the setup~\cite{Chen_2024}, suggests that Bell-pairs can be self-tested alongside the measurement used on the systems when the channel is applied. Combined together, these two elements could allow one to self-test any extremal multi-qubit channel (Sec.~\ref{sec:self-test_channels}), which would also lead to self-testing any extremal qubit supermap (Sec.~\ref{sec:ST_quantum_supermaps}). One prominent example would be to study the self-testing of the (purification of the~\cite{araujo17}) Lugano process, a unitary process violating causal inequalities~\cite{Baumeler16}.
In turn, it could also be relevant to study whether rigidity statements could be obtained based on maximal violations of causal inequalities, the analog of Bell inequalities for causal indefiniteness.

Finally, we recall that the results we presented focus on the exact case of perfect experimental statistics. Extending them to the approximate setting where noisy statistics lead to approximate rigidity statements is an important practical question. Robust self-testing for states has seen tremendous progress in the past years, see e.g.~\cite{Kaniewski16,coopmans2019robust,valcarce2022self}, and Ref.~\cite{Sekatski18} showed that this robustness can be generically lifted to self-testing of channels via the triangle inequality. Given that our method follows a similar construction, we believe that robust versions of our supermap self-tests can be derived, but leave this question for future work.

%%%%%%%%%%%%%%%%%%%%%%%%%%%%%%%%%%%%%%%%%%%%%%%%%%%%%%%%
\acknowledgments

This research was funded in part by l’Agence Nationale de la Recherche (ANR) projects ANR-15-IDEX-02 and ANR-22-CE47-0012, and the PEPR integrated project EPiQ ANR-22-PETQ-0007 as part of Plan France 2030.
For the purpose of open access, the authors have applied a CC-BY public copyright license to any Author Accepted Manuscript (AAM) version arising from this submission.

\bibliography{STSwitch}

@article{Quintino2019,
  title = {Probabilistic exact universal quantum circuits for transforming unitary operations},
  author = {Quintino, Marco T\'ulio and Dong, Qingxiuxiong and Shimbo, Atsushi and Soeda, Akihito and Murao, Mio},
  journal = {Phys. Rev. A},
  volume = {100},
  issue = {6},
  pages = {062339},
  numpages = {18},
  year = {2019},
  month = {Dec},
  publisher = {American Physical Society},
  doi = {10.1103/PhysRevA.100.062339},
  url = {https://link.aps.org/doi/10.1103/PhysRevA.100.062339}
}

@misc{gogioso23,
      title={The Geometry of Causality}, 
      author={Stefano Gogioso and Nicola Pinzani},
      year={2023},
      eprint={2303.09017},
      archivePrefix={arXiv},
      primaryClass={quant-ph},
}

@article{araujo17,
	author = {Mateus Ara\'{u}jo and Adrien Feix and Miguel Navascu{\'e}s and {\v{C}}aslav Brukner},
	doi = {10.22331/q-2017-04-26-10},
	journal = {Quantum},
	pages = {10},
	primaryclass = {quant-ph},
	title = {A purification postulate for quantum mechanics with indefinite causal order},
	volume = {1},
	year = {2017}}

@article{chiribella12,
	author = {Giulio Chiribella},
	doi = {10.1103/PhysRevA.86.040301},
	journal = {Phys. Rev. A},
	pages = {040301(R)},
	primaryclass = {quant-ph},
	title = {Perfect discrimination of no-signalling channels via quantum superposition of causal structures},
	volume = {86},
	year = {2012}}

@misc{costa26,
	archiveprefix = {arXiv},
	author = {Fabio Costa and Giulia Rubino and Cyril Branciard and {\v{C}}. Brukner and Marco T{\'u}lio Quintino},
	eprint = {2606.19438},
	primaryclass = {quant-ph},
	title = {Indefinite Quantum Causality},
	year = {2026}}

@article{chiribella08,
	author = {Giulio Chiribella and Giacomo M. D'Ariano and Paolo Perinotti},
	doi = {10.1103/PhysRevLett.101.180501},
	journal = {Phys. Rev. Lett.},
	pages = {180501},
	title = {Memory Effects in Quantum Channel Discrimination},
	volume = {101},
	year = {2008}}

@misc{dewolf2023quantum,
	author = {Ronald de Wolf},
	primaryclass = {quant-ph},
	title = {Quantum Computing: Lecture Notes},
	year = {2023},
    eprint={1907.09415},
    archivePrefix={arXiv},
    primaryClass={quant-ph},}

@article{valcarce2022self,
	author = {Valcarce, Xavier and Zivy, Julian and Sangouard, Nicolas and Sekatski, Pavel},
	doi = {doi.org/10.1103/PhysRevResearch.4.013049},
	journal = {Phys. Rev. Research},
	number = {1},
	pages = {013049},
	publisher = {APS},
	title = {Self-testing two-qubit maximally entangled states from generalized {C}lauser-{H}orne-{S}himony-{H}olt tests},
	volume = {4},
	year = {2022}}

@article{coopmans2019robust,
	author = {Coopmans, Tim and Kaniewski, J{\k{e}}drzej and Schaffner, Christian},
	doi = {doi.org/10.1103/PhysRevA.99.052123},
	journal = {Phys. Rev. A},
	number = {5},
	pages = {052123},
	publisher = {APS},
	title = {Robust self-testing of two-qubit states},
	volume = {99},
	year = {2019}}

@article{neves2025experimentally,
	author = {Neves, Simon and dos Santos Martins, Laura and Yacoub, Verena and Lefebvre, Pascal and {\v{S}}upi{\'c}, Ivan and Markham, Damian and Diamanti, Eleni},
	doi = {https://doi.org/10.1103/PRXQuantum.6.030312},
	journal = {PRX Quantum},
	number = {3},
	pages = {030312},
	publisher = {APS},
	title = {Experimentally Certified Transmission of a Quantum Message through an Untrusted and Lossy Quantum Channel via {B}ell's Theorem},
	volume = {6},
	year = {2025}}

@article{vilasini2024fundamental,
	author = {Vilasini, Venkatesh and Renner, Renato},
	doi = {10.1103/PhysRevLett.133.080201},
	journal = {Phys. Rev. Lett.},
	number = {8},
	pages = {080201},
	publisher = {APS},
	title = {Fundamental limits for realizing quantum processes in spacetime},
	volume = {133},
	year = {2024}}

@article{zalka1999grover,
	author = {Zalka, Christof},
	doi = {10.1103/PhysRevA.60.2746},
	issue = {4},
	journal = {Phys. Rev. A},
	month = {Oct},
	numpages = {0},
	pages = {2746--2751},
	publisher = {American Physical Society},
	title = {Grover's quantum searching algorithm is optimal},
	url = {https://link.aps.org/doi/10.1103/PhysRevA.60.2746},
	volume = {60},
	year = {1999}}

@inproceedings{grover1996fast,
	address = {New York, NY, USA},
	author = {Grover, Lov K.},
	booktitle = {Proceedings of the Twenty-Eighth Annual ACM Symposium on Theory of Computing},
	doi = {10.1145/237814.237866},
	isbn = {0897917855},
	location = {Philadelphia, Pennsylvania, USA},
	numpages = {8},
	pages = {212--219},
	publisher = {Association for Computing Machinery},
	series = {STOC '96},
	title = {A fast quantum mechanical algorithm for database search},
	url = {https://doi.org/10.1145/237814.237866},
	year = {1996}}

@misc{sekatski2025cloning,
	author = {Sekatski, Pavel and Guryanova, Yelena and Kothakonda, Naga Bhavya Teja and Skotiniotis, Michalis},
	eprint = {arXiv:2509.08059},
	primaryclass = {quant-ph},
	title = {Cloning Quantum Channels},
	year = {2025}}

@article{Supic23quantum,
	author = {{\v S}upi{\'c}, Ivan and Bowles, Joseph and Renou, Marc-Olivier and Ac{\'\i}n, Antonio and Hoban, Matty J.},
	copyright = {2023 The Author(s), under exclusive licence to Springer Nature Limited},
	doi = {10.1038/s41567-023-01945-4},
	issn = {1745-2481},
	journal = {Nature Phys.},
	langid = {english},
	month = may,
	number = {5},
	pages = {670--675},
	publisher = {Nature Publishing Group},
	title = {Quantum Networks Self-Test All Entangled States},
	urldate = {2024-02-05},
	volume = {19},
	year = 2023}

@article{storz2025complete,
	author = {Storz, Simon and Kulikov, Anatoly and Sch{\"a}r, Josua D. and Barizien, Victor and Valcarce, Xavier and Berterotti{\`e}re, Florence and Sangouard, Nicolas and Bancal, Jean-Daniel and Wallraff, Andreas},
	doi = {10.1103/nv7d-k3wr},
	journal = {Phys. Rev. Lett.},
	month = jul,
	number = {3},
	pages = {030801},
	publisher = {American Physical Society},
	title = {Complete {{Self-Testing}} of a {{System}} of {{Remote Superconducting Qubits}}},
	urldate = {2025-07-16},
	volume = {135},
	year = 2025}

@article{bell1964einstein,
	author = {Bell, John S},
	doi = {10.1103/PhysicsPhysiqueFizika.1.195},
	journal = {Physics},
	number = {3},
	pages = {195},
	publisher = {APS},
	title = {On the {E}instein {P}odolsky {R}osen paradox},
	volume = {1},
	year = {1964}}

@article{balanzo2026all,
	author = {Balanz{\'o}-Juand{\'o}, Maria and Coladangelo, Andrea and Augusiak, Remigiusz and Ac{\'\i}n, Antonio and {\v{S}}upi{\'c}, Ivan},
	doi = {https://doi.org/10.1038/s41467-026-70829-x},
	journal = {Nature Commun.},
	pages = {4463},
	publisher = {Nature Publishing Group UK London},
	title = {All pure multipartite entangled states of qubits can be self-tested},
	volume = {17},
	year = {2026}}

@article{coladangelo2016parallel,
	address = {Paramus, NJ},
	author = {Coladangelo, Andrea},
	issn = {1533-7146},
	issue_date = {August 2017},
	journal = {Quantum Inf. Comput.},
	month = aug,
	number = {9--10},
	numpages = {35},
	pages = {831--865},
	publisher = {Rinton Press, Incorporated},
	title = {Parallel self-testing of (tilted) {EPR} pairs via copies of (tilted) {CHSH} and the magic square game},
	url = {https://dl.acm.org/doi/abs/10.5555/3179561.3179567},
	volume = {17},
	year = {2017}}

@article{sekatski2017quantum,
	author = {Sekatski, Pavel and Skotiniotis, Michalis and Ko{\l}ody{\'n}ski, Janek and D{\"u}r, Wolfgang},
	doi = {10.22331/q-2017-09-06-27},
	journal = {Quantum},
	pages = {27},
	publisher = {Verein zur F{\"o}rderung des Open Access Publizierens in den Quantenwissenschaften},
	title = {Quantum metrology with full and fast quantum control},
	volume = {1},
	year = {2017}}

@article{dur2014improved,
	author = {D{\"u}r, Wolfgang and Skotiniotis, Michalis and Froewis, Florian and Kraus, Barbara},
	doi = {10.1103/PhysRevLett.112.080801},
	journal = {Phys. Rev. Lett.},
	number = {8},
	pages = {080801},
	publisher = {APS},
	title = {Improved quantum metrology using quantum error correction},
	volume = {112},
	year = {2014}}

@article{sekatski2023toward,
	author = {Sekatski, Pavel and Bancal, Jean-Daniel and Ioannou, Marie and Afzelius, Mikael and Brunner, Nicolas},
	doi = {10.1103/physrevlett.131.170802},
	journal = {Phys. Rev. Lett.},
	number = {17},
	pages = {170802},
	publisher = {APS},
	title = {Toward the device-independent certification of a quantum memory},
	volume = {131},
	year = {2023}}

@article{bock2024calibration,
	author = {Bock, Matthias and Sekatski, Pavel and Bancal, Jean-Daniel and Kucera, Stephan and Bauer, Tobias and Sangouard, Nicolas and Becher, Christoph and Eschner, J{\"u}rgen},
	doi = {https://doi.org/10.1038/s41534-024-00859-0},
	journal = {npj Quantum Inf.},
	number = {1},
	pages = {63},
	publisher = {Nature Publishing Group UK London},
	title = {Calibration-independent bound on the unitarity of a quantum channel with application to a frequency converter},
	volume = {10},
	year = {2024}}

@article{renou18self,
	author = {Renou, Marc Olivier and Kaniewski, J{\k e}drzej and Brunner, Nicolas},
	doi = {10.1103/PhysRevLett.121.250507},
	journal = {Phys. Rev. Lett.},
	month = dec,
	number = {25},
	pages = {250507},
	publisher = {American Physical Society},
	title = {Self-{{Testing Entangled Measurements}} in {{Quantum Networks}}},
	urldate = {2022-09-29},
	volume = {121},
	year = 2018}

@article{bancal2018noise,
	author = {Bancal, Jean-Daniel and Sangouard, Nicolas and Sekatski, Pavel},
	doi = {10.1103/PhysRevLett.121.250506},
	journal = {Phys. Rev. Lett.},
	month = dec,
	number = {25},
	pages = {250506},
	publisher = {American Physical Society},
	title = {Noise-{{Resistant Device-Independent Certification}} of {{Bell State Measurements}}},
	urldate = {2022-09-29},
	volume = {121},
	year = 2018}

@article{bavaresco2024designing,
	author = {Bavaresco, Jessica and Lipka-Bartosik, Patryk and Sekatski, Pavel and Mehboudi, Mohammad},
	doi = {10.1103/physrevresearch.6.023305},
	issn = {2643-1564},
	journal = {Phys. Rev. Research},
	month = June,
	number = {2},
	pages = {023305},
	publisher = {American Physical Society (APS)},
	title = {Designing optimal protocols in {B}ayesian quantum parameter estimation with higher-order operations},
	url = {http://dx.doi.org/10.1103/PhysRevResearch.6.023305},
	volume = {6},
	year = {2024}}

@article{Wang_2016,
	author = {Wang, Yukun and Wu, Xingyao and Scarani, Valerio},
	doi = {10.1088/1367-2630/18/2/025021},
	issn = {1367-2630},
	journal = {New J. Phys.},
	month = Feb,
	number = {2},
	pages = {025021},
	publisher = {IOP Publishing},
	title = {All the self-testings of the singlet for two binary measurements},
	url = {http://dx.doi.org/10.1088/1367-2630/18/2/025021},
	volume = {18},
	year = {2016}}

@article{Barizien25,
	author = {Barizien, Victor and Bancal, Jean-Daniel},
	doi = {10.1038/s41567-025-02782-3},
	journal = {Nature Phys.},
	number = {4},
	pages = {577--582},
	publisher = {Nature Publishing Group UK London},
	title = {Quantum statistics in the minimal {B}ell scenario},
	volume = {21},
	year = {2025}}

@article{liu23,
	author = {Liu, Qiushi and Hu, Zihao and Yuan, Haidong and Yang, Yuxiang},
	doi = {10.1103/physrevlett.130.070803},
	issn = {1079-7114},
	journal = {Phys. Rev. Lett.},
	month = Feb,
	number = {7},
	pages = {070803},
	publisher = {American Physical Society (APS)},
	title = {Optimal Strategies of Quantum Metrology with a Strict Hierarchy},
	url = {http://dx.doi.org/10.1103/PhysRevLett.130.070803},
	volume = {130},
	year = {2023}}

@article{chiribella2008transforming,
	author = {Chiribella, G. and D'Ariano, G. M. and Perinotti, P.},
	doi = {10.1209/0295-5075/83/30004},
	issn = {1286-4854},
	journal = {Europhys. Lett.},
	month = July,
	number = {3},
	pages = {30004},
	publisher = {IOP Publishing},
	title = {Transforming quantum operations: Quantum supermaps},
	url = {http://dx.doi.org/10.1209/0295-5075/83/30004},
	volume = {83},
	year = {2008}}

@article{wagner2020device,
	author = {Wagner, Sebastian and Bancal, Jean-Daniel and Sangouard, Nicolas and Sekatski, Pavel},
	doi = {10.22331/q-2020-03-19-243},
	issn = {2521-327X},
	journal = {Quantum},
	month = Mar,
	pages = {243},
	publisher = {Verein zur Forderung des Open Access Publizierens in den Quantenwissenschaften},
	title = {Device-independent characterization of quantum instruments},
	url = {http://dx.doi.org/10.22331/q-2020-03-19-243},
	volume = {4},
	year = {2020}}

@article{Barizien24,
	author = {Barizien, Victor and Sekatski, Pavel and Bancal, Jean-Daniel},
	doi = {10.22331/q-2024-05-02-1333},
	issn = {2521-327X},
	journal = {{Quantum}},
	month = may,
	pages = {1333},
	publisher = {{Verein zur F{\"{o}}rderung des Open Access Publizierens in den Quantenwissenschaften}},
	title = {Custom {B}ell inequalities from formal sums of squares},
	url = {https://doi.org/10.22331/q-2024-05-02-1333},
	volume = {8},
	year = {2024}}

@article{branciard16a,
	author = {Cyril Branciard},
	doi = {10.1038/srep26018},
	journal = {Sci. Rep.},
	pages = {26018},
	title = {Witnesses of causal nonseparability: an introduction and a few case studies},
	volume = {6},
	year = {2016}}

@article{van-der-lugt23a,
	author = {van der Lugt, Tein and Ormrod, Nick},
	doi = {10.22331/q-2024-12-03-1543},
	issn = {2521-327X},
	journal = {Quantum},
	month = Dec,
	pages = {1543},
	publisher = {Verein zur Forderung des Open Access Publizierens in den Quantenwissenschaften},
	title = {Possibilistic and maximal indefinite causal order in the quantum switch},
	url = {http://dx.doi.org/10.22331/q-2024-12-03-1543},
	volume = {8},
	year = {2024}}

@article{van-der-lugt23,
	author = {Tein {van der Lugt} and Jonathan Barrett and Giulio Chiribella},
	doi = {10.1038/s41467-023-40162-8},
	journal = {Nature Commun.},
	pages = {5811},
	primaryclass = {quant-ph},
	title = {Device-independent certification of indefinite causal order in the quantum switch},
	volume = {14},
	year = {20223}}

@article{dourdent22,
	author = {Hippolyte Dourdent and Alastair A. Abbott and Nicolas Brunner and Ivan {\v S}upi{\'c} and Cyril Branciard},
	doi = {10.1103/PhysRevLett.129.090402},
	journal = {Phys. Rev. Lett.},
	pages = {090402},
	primaryclass = {quant-ph},
	title = {Semi-Device-Independent Certification of Causal Nonseparability with Trusted Quantum Inputs},
	volume = {129},
	year = {2022}}

@article{Bancal15,
   title={Physical characterization of quantum devices from nonlocal correlations},
   volume={91},
   pages= {022115},
   ISSN={1094-1622},
   url={http://dx.doi.org/10.1103/PhysRevA.91.022115},
   DOI={10.1103/physreva.91.022115},
   number={2},
   journal={Phys. Rev. A},
   publisher={American Physical Society (APS)},
   author={Bancal, Jean-Daniel and Navascués, Miguel and Scarani, Valerio and Vértesi, Tamás and Yang, Tzyh Haur},
   year={2015},
   month=Feb }

@article{bavaresco19,
	author = {Jessica Bavaresco and Mateus Ara\'{u}jo and Brukner, {\v{C}}aslav and Quintino, Marco T\'{u}lio},
	doi = {10.22331/q-2019-08-19-176},
	journal = {Quantum},
	pages = {176},
	primaryclass = {quant-ph},
	title = {Semi-device-independent certification of indefinite causal order},
	volume = {3},
	year = {2019}}

@article{araujo14,
	author = {{Ara{\'u}jo}, M. and {Costa}, F. and {Brukner}, {\v C}.},
	doi = {10.1103/PhysRevLett.113.250402},
	journal = {Phys. Rev. Lett.},
	month = dec,
	number = 25,
	pages = {250402},
	primaryclass = {quant-ph},
	title = {{Computational Advantage from Quantum-Controlled Ordering of Gates}},
	volume = 113,
	year = 2014}

@article{Abbott24,
	author = {Abbott, Alastair A. and Mhalla, Mehdi and Pocreau, Pierre},
	doi = {10.1103/physrevresearch.6.l032020},
	issn = {2643-1564},
	journal = {Phys. Rev. Research},
	month = July,
	number = {3},
	pages = {L032020},
	publisher = {American Physical Society (APS)},
	title = {Quantum query complexity of {B}oolean functions under indefinite causal order},
	url = {http://dx.doi.org/10.1103/PhysRevResearch.6.L032020},
	volume = {6},
	year = {2024}}

@article{guerin16,
	author = {Gu{\'e}rin, Philippe Allard and Feix, Adrien and Ara{\'u}jo, Mateus and Brukner, {\v C}aslav},
	doi = {10.1103/PhysRevLett.117.100502},
	issue = {10},
	journal = {Phys. Rev. Lett.},
	month = {Sep},
	numpages = {5},
	pages = {100502},
	primaryclass = {quant-ph},
	publisher = {American Physical Society},
	title = {Exponential Communication Complexity Advantage from Quantum Superposition of the Direction of Communication},
	url = {https://link.aps.org/doi/10.1103/PhysRevLett.117.100502},
	volume = {117},
	year = {2016}}

@article{araujo15,
	author = {M. Ara\'{u}jo and C. Branciard and F. Costa and A. Feix and C. Giarmatzi and {\v{C}}. Brukner},
	doi = {10.1088/1367-2630/17/10/102001},
	journal = {New J. Phys.},
	pages = {102001},
	primaryclass = {quant-ph},
	title = {Witnessing causal nonseparability},
	volume = {17},
	year = {2015}}

@article{oreshkov16,
	author = {Ognyan Oreshkov and Christina Giarmatzi},
	doi = {10.1088/1367-2630/18/9/093020},
	isbn = {1367-2630},
	journal = {New J. Phys.},
	number = {9},
	pages = {093020},
	title = {Causal and causally separable processes},
	ty = {JOUR},
	volume = {18},
	year = {2016}}

@article{Chiribella13,
	author = {Chiribella, Giulio and D'Ariano, Giacomo Mauro and Perinotti, Paolo and Valiron, Benoit},
	doi = {10.1103/PhysRevA.88.022318},
	issue = {2},
	journal = {Phys. Rev. A},
	month = {Aug},
	numpages = {15},
	pages = {022318},
	primaryclass = {quant-ph},
	publisher = {American Physical Society},
	title = {Quantum computations without definite causal structure},
	url = {https://link.aps.org/doi/10.1103/PhysRevA.88.022318},
	volume = {88},
	year = {2013}}

@article{Oreshkov12,
	author = {Oreshkov, Ognyan and Costa, Fabio and Brukner, {\v{C}}aslav},
	doi = {10.1038/ncomms2076},
	journal = {Nature Commun.},
	pages = {1092},
	title = {{Quantum correlations with no causal order}},
	url = {https://doi.org/10.1038/ncomms2076},
	volume = {3},
	year = {2012}}

@article{Dourdent23,
	author = {Dourdent, Hippolyte and Abbott, Alastair A. and {\v S}upi{\'c}, Ivan and Branciard, Cyril},
	doi = {10.22331/q-2024-10-30-1514},
	issn = {2521-327X},
	journal = {Quantum},
	month = Oct,
	pages = {1514},
	publisher = {Verein zur Forderung des Open Access Publizierens in den Quantenwissenschaften},
	title = {Network-Device-Independent Certification of Causal Nonseparability},
	url = {http://dx.doi.org/10.22331/q-2024-10-30-1514},
	volume = {8},
	year = {2024}}

@article{Branciard16,
	author = {Cyril Branciard and Mateus Ara{\'u}jo and Adrien Feix and Fabio Costa and {\v C}aslav Brukner},
	doi = {10.1088/1367-2630/18/1/013008},
	journal = {New J. Phys.},
	month = {dec},
	number = {1},
	pages = {013008},
	publisher = {IOP Publishing},
	title = {The simplest causal inequalities and their violation},
	url = {https://dx.doi.org/10.1088/1367-2630/18/1/013008},
	volume = {18},
	year = {2015}}

@inproceedings{Magniez05,
	address = {Berlin, Heidelberg},
	author = {Magniez, Fr{\'e}d{\'e}ric and Mayers, Dominic and Mosca, Michele and Ollivier, Harold},
	booktitle = {Automata, Languages and Programming},
	editor = {Bugliesi, Michele and Preneel, Bart and Sassone, Vladimiro and Wegener, Ingo},
	isbn = {978-3-540-35905-0},
	pages = {72--83},
	publisher = {Springer Berlin Heidelberg},
	title = {Self-testing of Quantum Circuits},
	year = {2006}}

@article{Supic20,
	author = {{\v{S}}upi{\'{c}}, Ivan and Bowles, Joseph},
	doi = {10.22331/q-2020-09-30-337},
	issn = {2521-327X},
	journal = {{Quantum}},
	month = sep,
	pages = {337},
	publisher = {{Verein zur F{\"{o}}rderung des Open Access Publizierens in den Quantenwissenschaften}},
	title = {{Self-testing of quantum systems: a review}},
	url = {https://doi.org/10.22331/q-2020-09-30-337},
	volume = {4},
	year = {2020}}

@article{Mayers04,
	author = {Mayers D. and Yao, A.},
	doi = {10.48550/arXiv.quant-ph/0307205},
	journal = {Quantum Info. Comput.},
	pages = {273},
	title = {{Self testing quantum apparatus}},
	url = {https://doi.org/10.48550/arXiv.quant-ph/0307205},
	volume = {4},
	year = {2004}}

@article{Sekatski18,
	author = {Sekatski, Pavel and Bancal, Jean-Daniel and Wagner, Sebastian and Sangouard, Nicolas},
	doi = {10.1103/PhysRevLett.121.180505},
	issue = {18},
	journal = {Phys. Rev. Lett.},
	month = {Nov},
	numpages = {5},
	pages = {180505},
	publisher = {American Physical Society},
	title = {Certifying the Building Blocks of Quantum Computers from Bell's Theorem},
	url = {https://link.aps.org/doi/10.1103/PhysRevLett.121.180505},
	volume = {121},
	year = {2018}}

@article{Kaniewski16,
	author = {{Kaniewski, J.}},
	doi = {10.1103/PhysRevLett.117.070402},
	issue = {7},
	journal = {Phys. Rev. Lett.},
	month = {Aug},
	numpages = {6},
	pages = {070402},
	publisher = {American Physical Society},
	title = {Analytic and Nearly Optimal Self-Testing Bounds for the {C}lauser-{H}orne-{S}himony-{H}olt and {M}ermin Inequalities},
	url = {https://link.aps.org/doi/10.1103/PhysRevLett.117.070402},
	volume = {117},
	year = {2016}}

@article{Baumeler16,
	author = {Baumeler, {\"A}min and Wolf, Stefan},
	doi = {10.1088/1367-2630/18/1/013036},
	issn = {1367-2630},
	journal = {New J. Phys.},
	month = Jan,
	number = {1},
	pages = {013036},
	publisher = {IOP Publishing},
	title = {The space of logically consistent classical processes without causal order},
	url = {http://dx.doi.org/10.1088/1367-2630/18/1/013036},
	volume = {18},
	year = {2016}}

@article{taranto2025,
  title = {A Guide to Higher-order quantum operations},
  author = {Taranto, Philip and Milz, Simon and Murao, Mio and Quintino, Marco Túlio and Modi, Kavan},
  journal = {PRX Quantum},
  pages = {},
  year = {2026},
  month = {Jun},
  publisher = {American Physical Society},
  doi = {10.1103/92c4-g7qs},
  url = {https://link.aps.org/doi/10.1103/92c4-g7qs}
}

@article{Coladangelo_2017,
	author = {Coladangelo, Andrea and Goh, Koon Tong and Scarani, Valerio},
	doi = {10.1038/ncomms15485},
	issn = {2041-1723},
	journal = {Nature Commun.},
	month = may,
	number = {1},
	pages = {15485},
	publisher = {Springer Science and Business Media LLC},
	title = {All pure bipartite entangled states can be self-tested},
	url = {http://dx.doi.org/10.1038/ncomms15485},
	volume = {8},
	year = {2017}}

@article{Chen_2024,
	author = {Chen, Ranyiliu and Man{\v c}inska, Laura and Vol{\v c}i{\v c}, Jurij},
	doi = {10.1038/s41567-024-02584-z},
	issn = {1745-2481},
	journal = {Nature Phys.},
	month = aug,
	number = {10},
	pages = {1642--1647},
	publisher = {Springer Science and Business Media LLC},
	title = {All real projective measurements can be self-tested},
	url = {http://dx.doi.org/10.1038/s41567-024-02584-z},
	volume = {20},
	year = {2024}}

@article{Clauser69,
	author = {Clauser, John F. and Horne, Michael A. and Shimony, Abner and Holt, Richard A.},
	doi = {10.1103/PhysRevLett.23.880},
	journal = {Phys. Rev. Lett.},
	month = oct,
	number = {15},
	pages = {880--884},
	publisher = {{American Physical Society}},
	title = {Proposed {{Experiment}} to {{Test Local Hidden-Variable Theories}}},
	urldate = {2022-09-12},
	volume = {23},
	year = {1969}}

@article{Hossenfelder20,
	author = {Hossenfelder, Sabine and Palmer, Tim},
	doi = {10.3389/fphy.2020.00139},
	issn = {2296-424X},
	journal = {Frontiers Phys.},
	month = May,
	pages = {139},
	publisher = {Frontiers Media SA},
	title = {Rethinking Superdeterminism},
	url = {http://dx.doi.org/10.3389/fphy.2020.00139},
	volume = {8},
	year = {2020}}

@article{Jamiolkowski72,
	author = {A. Jamio{\l}kowski},
	doi = {https://doi.org/10.1016/0034-4877(72)90011-0},
	issn = {0034-4877},
	journal = {Reports Math. Phys.},
	number = {4},
	pages = {275-278},
	title = {Linear transformations which preserve trace and positive semidefiniteness of operators},
	url = {https://www.sciencedirect.com/science/article/pii/0034487772900110},
	volume = {3},
	year = {1972}}

@book{heinosaari11,
	author = {Heinosaari, Teiko and Ziman, M{\'a}rio},
	doi = {https://doi.org/10.1017/CBO9781139031103},
	publisher = {Cambridge University Press},
	title = {The mathematical language of quantum theory: from uncertainty to entanglement},
	year = {2011}}

@article{rozema24,
	author = {Rozema, Lee A. and Str{\"o}mberg, Teodor and Cao, Huan and Guo, Yu and Liu, Bi-Heng and Walther, Philip},
	doi = {10.1038/s42254-024-00739-8},
	issn = {2522-5820},
	journal = {Nature Rev. Phys.},
	month = July,
	number = {8},
	pages = {483--499},
	publisher = {Springer Science and Business Media LLC},
	title = {Experimental aspects of indefinite causal order in quantum mechanics},
	url = {http://dx.doi.org/10.1038/s42254-024-00739-8},
	volume = {6},
	year = {2024}}

@article{Sarkar21,
	author = {Sarkar, Shubhayan and Saha, Debashis and Kaniewski, J{\k e}drzej and Augusiak, Remigiusz},
	doi = {10.1038/s41534-021-00490-3},
	issn = {2056-6387},
	journal = {npj Quantum Inf.},
	month = Oct,
	number = {1},
	pages = {151},
	publisher = {Springer Science and Business Media LLC},
	title = {Self-testing quantum systems of arbitrary local dimension with minimal number of measurements},
	url = {http://dx.doi.org/10.1038/s41534-021-00490-3},
	volume = {7},
	year = {2021}}

@misc{Meyer25,
	archiveprefix = {arXiv},
	author = {Uta Isabella Meyer and Ivan {\v S}upi{\'c} and Fr{\'e}d{\'e}ric Grosshans and Damian Markham},
	eprint = {2508.01071},
	primaryclass = {quant-ph},
	title = {Robustly self-testing all maximally entangled states in every finite dimension},
	year = {2025}}

% \clearpage

%%%%%%%%%%%%%%%%%%%%%%%%%%%%%%%%%%%%%%%%%%%%%%%%%%%%%%%%
\appendix

\setcounter{lemma}{0}
\renewcommand{\thelemma}{\thesection.\arabic{lemma}}
\renewcommand{\thetheorem}{\thesection.\arabic{lemma}}

\onecolumngrid

%%%%%%%%%%%%%%%%%%%%%%%%%%%%%%%%%%%%%%%
\section{Self-testing the Choi channels} \label{app:self_testing_1comb}
\setcounter{lemma}{0}

In this section we present the self-testing results for the Choi channels of the supermaps introduced in~\cref{sec:examples} following the scheme proposed for channel self-testing presented in \cref{sec:self-test_channels} and summarized in~\cref{fig:selftesting_channel}. Namely, for each Choi channel $\mathcal{J}(\Bar W)$, we provide a self-testing result for both the Choi state of the channel $\mathcal{C}(\mathcal{J}(\Bar W))$ and the product of maximally entangled states $(\Phi^+)^{\otimes n}$ corresponding to when the channel is not applied.

\subsection{Constructing the extraction map} \label{app:extractionmap}

In the first subsection we start by presenting a general construction of a local extraction map that is tailored to the situation where the reference measurements $\bar A_{a|x}$ of a party are complementary qubit measurements with $a,x=0,1$. This map applies to all examples discussed in the main text and presented afterwards and can be used to guarantee that the Choi states with and without to the application of the are self-tested with identical extraction maps on Alice's (channel-free) side, see \cref{sec:self-test_channels}. For completeness, at  the end of the subsection we discuss how the construction generalizes to the case of arbitrary measurements in all finite dimensions.

The idea, pioneered by Mayers and Yao in their seminal work \cite{Mayers04}, consists of combining the measurement boxes of the party under consideration $A_{a|x}$ into a black-box quantum circuit that performs the extraction.  Remark that the physical measurement box described by the POVM elements $A_{a|x}$ does not have a quantum output, and has a very limited usage inside a quantum circuit. Nevertheless, for each input setting $x$ the measurements box {\it defines} a unitary channel $\hat{\mathcal{A}}_x$
     \begin{equation}
      \includegraphics[height=1cm]{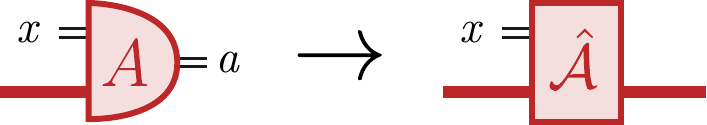}.
    \end{equation}
The latter $\hat {\mathcal{A}}_x [\cdot] =\hat A_x \, \cdot \, \hat A_x^\dag$, can be defined from the POVM elements via a regularization step\footnote{In addition the zero eigenvalues of $A_{0|x}-A_{1|x}$ are reset to 1, which we do not write explicitly.} $\hat A_x := |A_{0|x}-A_{1|x}|^{-1} (A_{0|x}-A_{1|x})$, such that the operator $\hat A_x$ has eigenvalues $\pm 1$ and the same eigenstates as $A_{0|x}$ and $A_{1|x}$. This allows one to define the extraction map $\Lambda_A$ through the circuit
  \begin{equation} \label{eq:explicitextraction}
      \includegraphics[height=2.5cm]{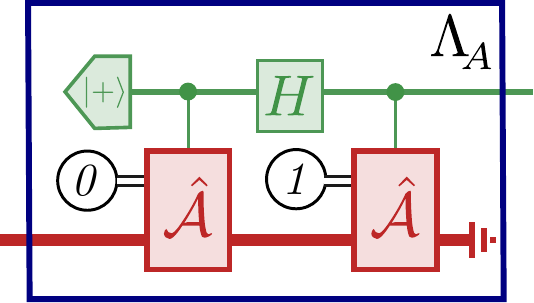},
    \end{equation}
by combining the unitaries obtained from the two measurements controlled on an auxiliary qubit (green wire). Here, this qubit is prepared in the state $\ket{+}=(\ket{0}+\ket{1})/\sqrt{2}$ and $H$ is the Hadamard gate. To understand the idea of this construction, remark that when the red wire carries a qubit and an auxiliary system and $\hat A_0$, $\hat A_1$ are complementary operators ($\{\hat A_0, \hat A_1\} = \id$) acting only on a qubit system, e.g~ $\hat A_0 = \sigma_z \otimes \id$ with $\hat A_1 = \sigma_x\otimes \id$, the above circuit swaps this qubit onto the green wire. See Section~4.3 of Ref.~\cite{Supic20} for a more detailed discussion.
 When the reference measurements are not necessarily complementary and/or act on systems of dimension $d\geq 2$, this construction can be generalized~\cite{Bancal15}.

For pedagogical reason, let us first start by only lifting the complementary assumption while keeping $d=2$. In this case, the extraction circuit of~\cref{eq:explicitextraction} does not realize a \SWAP operation even when the physical measurements are exactly the reference ones. To circumvent this, one needs to build a new operator $\hat{\mathcal{A}}_\perp$, complementary to $\hat{\mathcal{A}}_0$ which is inserted in the circuit of Eq.~\eqref{eq:explicitextraction} instead of  $\hat{\mathcal{A}}_1$. Constructing this new operator can be done using only resources accessible within the physical devices as 
\begin{align}
   \hat {\mathcal A}_\perp = |{\mathcal A}_\perp|^{-1} {\mathcal A}_\perp \qquad {\mathcal A}_\perp = \alpha_0 (A_{0|0}-A_{1|0})+ \alpha_1 (A_{0|1}-A_{1|1}).
\end{align}
Here, the regularization step guarantees that $\hat{\mathcal{A}}_\perp$ is unitary,
while $\alpha_0, \alpha_1$ are real parameters chosen such that when $\hat{\mathcal{A}}_0$, $\hat{\mathcal{A}}_1$ are the reference measurements, $\hat{\mathcal{A}}_0$ and $\hat{\mathcal{A}}_1$ are complementary, i.e.~$\{\hat{\mathcal{A}}_\perp, \hat{\mathcal{A}}_0\}=0$. There always exist such parameters as long as the reference measurements are not the same. \{Note that even though for all physical implementations the operator $\hat{\mathcal{A}}_\perp$ is not unitary  (hence the map $\Lambda_a$ is not necessarily trace-preserving), self-testing of the complete realization (state and measurements) ensures that this is the case when acting on the physical state $\rho$.\}

For high-dimensional systems, Ref.~\cite{Bancal15} provides a general construction for a \SWAP circuit, similar to the one of~\cref{eq:explicitextraction}. The basic idea follows the observation that the generalized Pauli, or Weyl, operators $Z$ and $X$ in any dimension $d$ can be used to define a high-dimensional \SWAP gate when combined in a circuit resembling to Eq.~\eqref{eq:explicitextraction}, but is not restricted to this case. In general, such a circuit relies on operators $\hat{\mathcal{A}}$ acting on the physical system, constructed as polynomials of the measurements operators $\hat{\mathcal{A}}_0$, $\hat{\mathcal{A}}_1$, i.e.~$\hat{\mathcal{A}}:=Q(\hat{\mathcal{A}}_0, \hat{\mathcal{A}}_1)$. Again, self-testing of the complete realization (state and measurements) ensures that the action of any such polynomial on the state is equal to the action of the reference measurements on the reference state. This is thus a sufficient condition to guarantee that the overall extraction circuit performs a \SWAP on the support of $\rho$ when the target experimental statistics are achieved.

\subsection{Self-testing the sources}

We start by treating the latter for all examples together. Indeed, when the local dimension is 2 (qubit case), the state $\Phi^+$ can be self-tested using the maximal violation of the CHSH inequality~\cite{Clauser69} in a bipartite Bell scenario as per the following lemma.

\begin{lemma}[Section 4 of \cite{Supic20}]
    \textbf{CHSH self-test.} \label{LemmaSelftest_CHSH} \ \\
    Let's consider a bipartite Bell scenario with binary inputs and outputs. The two parties, labeled $A$ and $B$, thus share a quantum state $\rho_{AB}$ and perform two binary outcome measurements -- $\pm1$-valued observables -- $A_x, B_y$ for $x,y \in \{0,1\}$. Let 
    \begin{equation}
        \beta_\text{CHSH}^{(AB)}=A_0B_0 + A_0B_1 +A_1B_0 - A_1B_1.
    \end{equation} 
    If the observed statistics are such that 
    \begin{equation}
        \Tr(\rho_{AB} \cdot \beta_\text{CHSH}^{(AB)})=2\sqrt{2}, 
    \end{equation} 
    then the states and measurements are self-tested (in the sense of the rigidity relations~\cref{eq: st states} and \cref{eq: st meas}) with reference realization
    \begin{align}\label{eq:refchsh}
        \begin{aligned}
            & \Phi^+ = \frac{1}{\sqrt{2}}(\ket{00} + \ket{11}),\\
            & A_0 = \sigma_z, \quad A_1 = \sigma_x, \\
            & B_y = \frac{1}{\sqrt{2}}(\sigma_z + (-1)^y \sigma_x),
        \end{aligned}
    \end{align}
    where $\sigma_z$ and $\sigma_x$ are Pauli $z$ and $x$ operators. 
\end{lemma}

A few notes are in order. First, note that in the case of larger dimension, similar results exist for certifying the maximally entangled state of two qudits (see e.g.~\cite{coladangelo2016parallel,Sarkar21,Meyer25}). Second, this self-test can be made robust to noise up to a CHSH violation of approximately $2.11$, see~\cite{Kaniewski16}. Finally, self-testing of the maximally entangled two qubit state is possible with other choices of reference measurements. For example, in the binary input/output Bell scenario, all such self-tests are known and given in Refs.~\cite{Wang_2016, Barizien24}. 

\cref{LemmaSelftest_CHSH} can directly be used to self-test an arbitrary product of maximally entangled two qubit state between distinct pairs of parties. 
\begin{lemma}\textbf{Simultaneous CHSH self-tests.} \label{LemmaSelftest_CHSH_parallel} \ \\
    Let's consider a $2n$-partite Bell scenario with binary inputs and outputs in which quantum states $\rho_{A^{(i)}B^{(i)}}$ are shared between pairs of parties $A^{(i)}$ and $B^{(i)}$ for $i\in\{1,...,n\}$. Suppose that for each $i$, the observed statistics are such that 
    \begin{equation}
        \Tr(\rho_{A^{(i)}B^{(i)}} \cdot \beta_\text{CHSH}^{(A^{(i)}B^{(i)})})=2\sqrt{2}, 
    \end{equation} 
    then the states and measurements are self-tested (in the sense of the rigidity relations \cref{eq: st states} and \cref{eq: st meas}) with the reference state 
    \begin{align}
        \begin{aligned}
            & \bigotimes_{i=1}^n \Phi^+_{A^{(i)}B^{(i)}},
        \end{aligned}
    \end{align}
    and reference measurements as in~\cref{eq:refchsh} for each pair of parties. 
\end{lemma}

\subsection{Self-testing the SWAP}\label{app:self_testing_SWAP}

In this section, we propose a self-testing result for the qubit \SWAP operation, seen as a channel acting on two systems. This amounts to self-test both the product of two maximally entangled qubits states $\Phi^+_{A^{(1)}B^{(1)}} \otimes \Phi^+_{A^{(2)}B^{(2)}}$ as well as the Choi state $\mathcal{C(\SWAP)}$ obtained by applying the \SWAP operator on the parties $B^{(1)}, B^{(2)}$, i.e.~$\mathcal{C(\SWAP)} = \Phi^+_{A^{(1)}B^{(2)}} \otimes \Phi^+_{A^{(2)}B^{(1)}}$. This can be done by requiring that the observed statistics verify~\cref{LemmaSelftest_CHSH_parallel} between parties $A^{(1)}, B^{(1)}$ and $A^{(2)},B^{(2)}$ when the channel is not applied and between parties $A^{(1)}, B^{(2)}$ and $A^{(2)},B^{(1)}$ when the channel is applied. Since both self-tests have the same reference measurements for parties $A^{(i)}$ (as given by~\cref{eq:refchsh}), this ensures self-testing of the $\SWAP$ channel.

\subsection{Self-testing the Choi channel of the error-correcting comb} \label{app:self_testing_EC_comb}

In this section we propose a self-testing result for the Choi channel $\mathcal{J}(\bar W_{\text{EC}})$ of the error-correcting comb, discussed in Sec.~\ref{sec:ECcomb} of the main text. This amounts to certifying 
\begin{equation}
    \bar \Psi = \Phi^+_{A^{(1)}B^{(1)}}\otimes \Phi^+_{A^{(2)}B^{(2)}},
\end{equation}
when the channel is not applied, which can be done using~\cref{LemmaSelftest_CHSH_parallel}, and the Choi state 
\begin{equation}
\begin{aligned}
    \mathcal{C}(\mathcal{J}(\bar W_{\text{EC}})) & = (\id_{A^{(1)}A^{(2)}}\otimes \mathcal{J}(\bar W_{\text{EC}})_{B^{(1)}B^{(2)}})[\ketbra{\bar \Psi}] = \frac{1}{2} \big[ \ketbra{G_0} \otimes \ketbra{0}_r + \ketbra{G_1} \otimes \ketbra{1}_r \big]
\end{aligned}
\end{equation}
when the channel is applied, where $\ket{G_0} = \ket{GHZ^4}$, $\ket{G_1} = \sigma_x^{A^{(1)}}\ket{GHZ^4}$, and
\begin{equation}
    \ket{GHZ^4} = \frac{1}{\sqrt{2}}\big( \ket{0000} + \ket{1111}\big).
\end{equation}
In the following, we detail how $\mathcal{C}(\mathcal{J}(\bar W_{\text{EC}}))$ can be self-tested with the same measurement as in by~\cref{LemmaSelftest_CHSH_parallel} for the parties $A^{(i)}$. First, let us state the following lemma.

\begin{lemma}[Section 3.3 of \cite{Barizien24}] \textbf{Self-testing $\ketbra{G_0}$.}\label{lem:GHZ} \\ 
Consider the following 4-partite Bell inequality, in a binary settings scenario:
\begin{equation}\label{eq:BellineqGHZ}
\begin{split}
    \langle S^4\rangle := \frac{1}{3}\big(\langle & (A_0^{(1)} + A_0^{(2)} + B_0^{(1)}) (B_0^{(2)}+B_1^{(2)})\rangle \big)   + \langle A_1^{(1)} A_1^{(2)} B_1^{(1)} (B_0^{(2)} - B_1^{(2)}) \rangle  \leq 2\sqrt{2}.
\end{split}
\end{equation}
Reaching the Tsirelson bound of this inequality self-test the following reference realization:
\begin{equation} \label{eq:measGHZ}
    \begin{split}
    & \ketbra{G_0} := \ketbra{GHZ^4}, \\
        & A_0^{(1)} = A_0^{(2)} = B_0^{(1)} = \sigma_z, \quad A_1^{(1)} = A_1^{(2)} = B_1^{(1)} = \sigma_x, \\
        & B_0^{(2)} = \frac{\sigma_z+\sigma_x}{\sqrt{2}}, \quad B_1^{(2)} = \frac{\sigma_z-\sigma_x}{\sqrt{2}}.
    \end{split}
\end{equation}
\end{lemma}

Now notice that measuring the state $\ket{G_1} = \sigma_x^{A^{(1)}}\ket{GHZ^4}$ with the same measurements will give the same correlations but with $A_0^{(1)} \to -A_0^{(1)}$, which is a symmetry of the quantum set. Therefore, these correlations are still extremal and self-test the state (see~\cite{Barizien25} for a discussion on symmetries and self-testing relations) $\ket{G_1} = \sigma_x^{A^{(1)}}\ket{GHZ^4}$ with the exact same measurements, so we have the following result. 

\begin{lemma}\textbf{Self-testing $\ketbra{G_1}$.}\label{lem:GHZsim} \\
Consider the following symmetry of the Bell inequality~\cref{eq:BellineqGHZ}:
\begin{equation}\label{eq:BellineqGHZsim}
\begin{split}
    \langle S^4_{A_0^{(1)} \to -A_0^{(1)}}\rangle := \frac{1}{3}\big(\langle & (-A_0^{(1)} + A_0^{(2)} + B_0^{(1)}) (B_0^{(2)}+B_1^{(2)})\rangle \big) + \langle A_1^{(1)} A_1^{(2)} B_1^{(1)} (B_0^{(2)} - B_1^{(2)}) \rangle \leq 2\sqrt{2}.
\end{split}
\end{equation}
The Tsirelson bound of this inequality self-tests the state $\ketbra{G_1}:=\sigma_x^{A^{(1)}}\ketbra{GHZ^4}\sigma_x^{A^{(1)}}$ and the measurements of~\cref{eq:measGHZ}.
\end{lemma}

Note that both self-testing results implies the existence of local maps on $A^{(1)}, A^{(2)}, B^{(1)}, B^{(2)}$ that extract the target state (either $\ketbra{G_0}$ or $\ketbra{G_1}$) from any state achieving the quantum bound. Since the measurements used in both lemmas are fixed, these local maps can be taken identical in both cases. This leads to the following result.

\begin{theorem} \label{thm:outputbitflipstate}
    \textbf{Self-testing $\mathcal{C}(\mathcal{J}(\bar W_{\text{EC}}))$.} \\
    Consider a physical realization of the considered self-testing experiment, two sources and an instrument applied on half of each source between the four parties $A^{(1)},A^{(2)},B^{(1)},B^{(2)}$. Each party performs two binary outcomes measurements and there is an additional binary classical register $r\in\{0,1\}$, accessible only to $B^{(1)}$ and $B^{(2)}$. 

    Denote respectively by $\langle \cdot \rangle_{r=0}$ and $\langle \cdot \rangle_{r=1}$ the statistics between the four parties conditioned on the classical register. The following conditions are sufficient to self-test the target output state $\mathcal{C}(\mathcal{J}(\bar W_{\text{EC}}))$:
    \begin{enumerate}
        \item $P(r=0) = P(r=1) = 1/2$, 
        \item $\langle S^4\rangle_{r=0} = 2\sqrt{2}$,
        \item $\langle S^4_{A_0^{(1)} \to -A_0^{(1)}}\rangle_{r=1} = 2\sqrt{2}$,
    \end{enumerate}    
    along with reference measurements given by~\cref{eq:measGHZ}.  
\end{theorem}

\begin{proof}
    For a given physical realization of the statistics, the physical output state is given by:
\begin{equation}
    \rho =\frac{1}{2} \left(\rho_0 \otimes \ketbra{0}_r + \rho_1 \otimes \ketbra{1}_r\right)
\end{equation}
where $\rho_0$ and $\rho_1$ are normalized states, and we already used the first condition $P(r=0) = P(r=1) = 1/2$. The second and third conditions imply that we can use \cref{lem:GHZ} and \cref{lem:GHZsim} to certify that there exist local maps, which do not depend on the classical register $r$, such that the states $\rho_0$ and $\rho_1$ are mapped to $\ketbra{G_0}$ and $\ketbra{G_1}$ respectively. Thus, under the action of these local maps $\rho$ is mapped to $\mathcal{C}(\mathcal{J}(\bar W_{\text{EC}}))$.
\end{proof}

Finally, since the reference measurements and the extraction maps for parties $A^{(i)}$ are the same in~\cref{thm:outputbitflipstate} and~\cref{LemmaSelftest_CHSH_parallel}, these two results can be used to self-test $\mathcal{J}(\bar W_{\text{EC}})$.

\subsection{Self-testing the Choi channel of Grover's algorithm} \label{sec:STGrover}

In this section we propose a self-testing result for the Choi channel of Grover's algorithm. As discussed in~\cref{sec:self_test_Grover}, this amounts to certifying the state preparation -- when the channel is not applied, and the Choi state of the channel -- when the channel is applied.
The former state, product of $(M+1)$ copies of the state $(\Phi^+)^{\otimes n}$, can be certified using $(M+1)n$ parallel CHSH self-tests, see~\cref{LemmaSelftest_CHSH_parallel}.
The latter is the product of one $(\Phi^+)^{\otimes n}$ state, certified again by parallel CHSH tests, and $M$ states $({\rm id}\otimes {\rm D})[(\Phi^{+})^{\otimes n}]$ arising from each slot, certified by $M$ parallel ``Grover'' tests, that we will now derive and summarize in~\cref{TheoremSelftestGrover}. We emphasize that all the reference measurements used in these self-test are identical to the ones required by~\cref{LemmaSelftest_CHSH_parallel}, and can thus be combined to self-test the Grover's Choi channel.

Let us first take a closer look at the state that we need to certify. In the following, we label $A =\{A^{(i)}\}$ the ``channel-free'' systems and $B =\{B^{(i)}\}$ the ones on which $\rm D$ is applied. We can write $({\rm id}_A\otimes {\rm D}_B)[(\Phi^{+})^{\otimes n}] = \ketbra{\psi_{\rm G}} $ with
\begin{equation} \label{eq:groverstateST}
    \begin{split}
        \ket{\psi_{\rm G}}_{AB} = \frac{1}{\sqrt{2^n}} \sum_{\bm x \in \{+,-\}^n} (-1)^{\delta_{\bm x \neq \bm +^n}} \ket{\bm x}_A \ket{\bm x}_B. 
    \end{split}
\end{equation}

In order to construct the desired Bell inequality to self-test this state, we use the method presented in~\cite{Barizien24}. This method's aim is to derive a Bell inequality tailored to the target state, i.e.~that is uniquely maximized by this state. This is a necessary requirement for any Bell inequality self-testing the target state. It relies on 4 steps:
\begin{enumerate}
    \item Find a complete set of ``nullifiers'' $\{\hat{N}_i\}$ of the target state $\ket{\psi}$, i.e.~operators such that $\hat{N}_i\ket{\psi}=0$ and $\cap_i \ker \hat{N}_i = \text{Vect}(\ket{\psi})$. 
    \item Make a choice (or a parametrization) of measurements and express the nullifiers in terms of the measurements operators.
    \item Lift the nullifers to the formal polynomial algebra (where commutation rules specific to the ideal Hilbert space are omitted) and compute there SOS: $\sum_i N_i^2$. 
    \item Check wether it is possible to cancel all ``non-measurable'' terms in the SOS, i.e.~all terms whose local degree is strictly larger than $1$. The Bell inequality and its SOS bound are then obtained by looking at measurable terms and identity terms respectively.
\end{enumerate}
We present below how we apply this method to the state in~\cref{eq:groverstateST}. \\

\noindent \underline{\textbf{Step 1:}} In order to find the nullifiers, we first observe that in all generality, if, for some $\ket{\psi}$, an operator $\hat{N}$ verifies $\hat{N}\ket{\psi}=0$, then for any unitary $U$, the state $U\ket{\psi}$ is nullified by the operator $\hat{N}_U := U\hat{N}U^\dagger$. Moreover, the kernel of the new nullifier $\ker(\hat{N}_U)$ is given by $U\ker(\hat{N})$. This simple observation means that for any complete set of nullifiers $\{ \hat N_i \}$ of state $\ket{\psi}$, we can build a complete set of nullfiers $\{\hat{N}_{U,i} \}$ for the state $U \ket{\psi}$. 

We now use this property to construct a complete family of nullifiers for~\cref{eq:groverstateST}. Indeed, as $\ket{\psi_{\rm G}}_{AB}= (\id_A \otimes U_{\rm D}) \ket{\phi^+}^{\otimes n}_{AB}$ with
\begin{equation}
U_{\rm D} =2 \ketbra{+}^{\otimes n}-\id = 2\prod_{j=1}^n \frac{\id+\sigma_x^{B^{(j)}}}{2} - \id,
\end{equation}
it suffices to find a complete set of nullifiers for the state $\ket{\phi^+}^{\otimes n}_{AB}$. Such a set is for example given by:
\begin{equation}
    \hat{N}_{x,i} = {\sigma_x}^{A^{(i)}}-{\sigma_x}^{B^{(i)}}, \quad \hat{N}_{z,i} = {\sigma_z}^{A^{(i)}}-{\sigma_z}^{B^{(i)}}. 
\end{equation}
The complete set of nullifiers for our target state $\ket{\psi_{\rm G}}_{AB}$ is therefore given by:
\begin{subequations}
    \begin{align}
        & \hat{N}_{U_{\rm D}, x,i} := (\id_A \otimes U_{\rm D})\hat{N}_{x,i} (\id_A \otimes U_{\rm D})^\dagger = \hat{N}_{x,i} = {\sigma_x}^{A^{(i)}}-{\sigma_x}^{B^{(i)}}, \label{eq:Xtypenull}\\
        & \hat{N}_{U_{\rm D}, z,i} := (\id_A \otimes U_{\rm D})\hat{N}_{z,i} (\id_A \otimes U_{\rm D})^\dagger = {\sigma_z}^{A^{(i)}} - {\sigma_z}^{B^{(i)}} \left(1- 2 \prod_{j\neq i} \frac{\id + \sigma_x^{B^{(j)}}}{2}\right).
    \end{align}
\end{subequations}
Note that since the state is symmetric under $A \leftrightarrow B$, we will rather use 
\begin{equation} \label{eq:Ztypenull}
    \begin{split}
        \hat{N}_{U_{\rm D}, z,i}' := {\sigma_z}^{A^{(i)}} \left(1- 2 \prod_{j\neq i} \frac{\id + \sigma_x^{A^{(j)}}}{2}\right) - {\sigma_z}^{B^{(i)}}.
    \end{split}
\end{equation}

\noindent \underline{\textbf{Step 2:}} We choose the measurement operators to be the same as in the CHSH test, see~\cref{eq:refchsh}. As such, the Pauli operators express as 
\begin{equation}
    \begin{split}
        & {\sigma_z}^{A^{(i)}} = A_0^{(i)}, \quad {\sigma_x}^{A^{(i)}} = A_1^{(i)},\\
        & {\sigma_z}^{B^{(i)}} = \frac{B_0^{(i)}+B_1^{(i)}}{\sqrt{2}}, \quad {\sigma_x}^{B^{(i)}} = \frac{B_0^{(i)}-B_1^{(i)}}{\sqrt{2}},
    \end{split}
\end{equation}
in terms of the measurement operators for all $i \in \{1,...,n\}$. We can then insert this in~\cref{eq:Xtypenull} and~\cref{eq:Ztypenull} to obtain the expression of the nullifiers in terms of measurement operators to obtain:
\begin{equation}
    \begin{split}
        & N_{U_{\rm D}, x,i} = A_1^{(i)} - \frac{B_0^{(i)}-B_1^{(i)}}{\sqrt{2}},\\
        & N'_{U_{\rm D}, z,i} = A_0^{(i)} \left(1-2\prod_{j\neq i} \frac{1 + A_1^{(j)}}{2}\right) - \frac{B_0^{(i)}+B_1^{(i)}}{\sqrt{2}}.
    \end{split}
\end{equation}

\noindent \underline{\textbf{Step 3:}} We can now lift the nullifiers to the formal polynomial algebra (don't assume anything on the anti-commutation relations of the measurement operators) and compute the SOS
\begin{equation}\label{eq:SOS1}
    P := \sum_{i\in\{1,...,n\}} N_{U_{\rm D}, x,i}^2  + N_{U_{\rm D}, z,i}'^2,
\end{equation}
where the variables $N$ and $N'$ without hats are indeterminates, to be distinguished from the qubit operators $\hat N$, $\hat N'$.
To compute this quantity, we use the identities $(A_x^{(i)})^2 = (B_y^{(i)})^2=1$ (which holds for any projective binary outcome measurement) and commutation between the measurements of different parties. As such:
\begin{equation}
\begin{split}
    N_{U_{\rm D}, x,i}^2 & = \left(A_1^{(i)} - \frac{B_0^{(i)}-B_1^{(i)}}{\sqrt{2}}\right)^2\\
    & = 1 - 2 A_1^{(i)} \frac{B_0^{(i)}-B_1^{(i)}}{\sqrt{2}} + \frac{1}{2} (B_0^{(i)}-B_1^{(i)})^2 \\
    & = 2 - \sqrt{2} A_1^{(i)} (B_0^{(i)}-B_1^{(i)}) - \frac{1}{2} \{B_0^{(i)}, B_1^{(i)}\},
\end{split} 
\end{equation}
while 
\begin{equation}
\begin{split}
    N_{U_{\rm D}, z,i}'^2 &= \left(A_0^{(i)} \left(1-2\prod_{j\neq i} \frac{1 + A_1^{(j)}}{2}\right) - \frac{B_0^{(i)}+B_1^{(i)}}{\sqrt{2}}\right)^2\\
    & = \left(A_0^{(i)} \left(1-2\prod_{j\neq i} \frac{1 + A_1^{(j)}}{2}\right)\right)^2 - 2A_0^{(i)} \left(1-2\prod_{j\neq i} \frac{1 + A_1^{(j)}}{2}\right)\frac{B_0^{(i)}+B_1^{(i)}}{\sqrt{2}} + \frac{1}{2}(B_0^{(i)}+ B_1^{(i)})^2 \\
    & = 1- 2A_0^{(i)} \left(1-2\prod_{j\neq i} \frac{1 + A_1^{(j)}}{2}\right)\frac{B_0^{(i)}+B_1^{(i)}}{\sqrt{2}}  + \frac{1}{2} \{B_0^{(i)}, B_1^{(i)}\},
\end{split}
\end{equation}
where we also used that $\left(\frac{1 + A_1^{(j)}}{2}\right)^2 = \frac{1 + A_1^{(j)}}{2}$.\\

\noindent \underline{\textbf{Step 4:}} Summing the last two terms cancels out the non-measurable terms $\{B_0^{(i)}, B_1^{(i)}\}$. The SOS is then obtained by summing over $i\in\{1,...,n\}$ and is given by
\begin{equation} \label{eq:SOSgrover}
\begin{split}
    P = 3n - \sqrt{2}S_{\text{Grover}}
\end{split}
\end{equation}
where 
\begin{equation} \label{eq:BellopgroverST}
\begin{split}
    S_{\text{Grover}} := \sum_{i\in\{1,...,n\}}\left[  A_1^{(i)} (B_0^{(i)}-B_1^{(i)}) + A_0^{(i)}(B_0^{(i)}+B_1^{(i)})  \left(1-2\prod_{j\neq i} \frac{1 + A_1^{(j)}}{2}\right)\right]
\end{split}
\end{equation}
is the candidate Bell expression. Note that the SOS in \cref{eq:SOS1} immediately implies that $S_{\text{Grover}} \preceq 3n/\sqrt{2}$. Since the ideal realization (with state~\cref{eq:groverstateST} and CHSH measurements~\cref{eq:refchsh}) saturates this bounds, this is indeed the quantum bound of this operator. We can now move to our main result.

\begin{theorem} \textbf{Self-testing $\ketbra{\psi_{\rm G}}$.} \label{TheoremSelftestGrover} \\
    Let us consider a 2n-partite Bell scenario where each party performs two binary measurements. Consider the Bell expression~\cref{eq:BellopgroverST}. If the observed statistics verify  
\begin{equation}
    \Tr(\rho S_{\text{Grover}}) = 3n/\sqrt{2},
\end{equation} the state $\ketbra{\psi_{\rm G}}$ (see~\cref{eq:switchstateout}) is self-tested. 
\end{theorem}

\begin{proof} 
Without loss a generality (dilation to larger spaces), we assume that the shared state $\rho = \ketbra{\psi}$ is pure. 

First, due to the SOS decomposition~\cref{eq:SOSgrover}, it holds that $N_{U_{\rm D}, x,i} \ket{\psi} = N_{U_{\rm D}, z,i}'\ket{\psi} = 0 $ for all $i\in\{1,...,n\}$, i.e.
\begin{subequations}
    \begin{align}
        & A_1^{(i)} \ket{\psi} =  \frac{B_0^{(i)}-B_1^{(i)}}{\sqrt{2}} \ket{\psi}, \label{eq:nullif1G}\\
        & A_0^{(i)} \left(1-2\prod_{j\neq i} \frac{1 + A_1^{(j)}}{2}\right) \ket{\psi} =  \frac{B_0^{(i)}+B_1^{(i)}}{\sqrt{2}}\ket{\psi}. \label{eq:nullif2G}
    \end{align}
\end{subequations}
\cref{eq:nullif1G} guarantees for all $i$ that 
\begin{equation}
    \ket{\psi} = (A_1^{(i)})^2\ket{\psi} = \left(\frac{B_0^{(i)}-B_1^{(i)}}{\sqrt{2}} \right)^2\ket{\psi} = \left(1 - \frac{1}{2}\{B_0^{(i)}, B_1^{(i)}\} \right)\ket{\psi},
\end{equation}
or equivalently
\begin{equation} \label{eq:Banticomm}
\{B_0^{(i)}, B_1^{(i)}\} \ket{\psi} = 0.
\end{equation}

Next, notice that the operator $ 1-2\prod_{j\neq i} \frac{1 + A_1^{(j)}}{2}$ is an involution, i.e.
\begin{equation} \label{eq:Ginvolution}
    \left(1-2\prod_{j\neq i} \frac{1 + A_1^{(j)}}{2}\right)^2 = 1.
\end{equation}
We can thus rewrite~\cref{eq:nullif2G} as
\begin{equation}
    A_0^{(i)} \ket{\psi} =  \frac{B_0^{(i)}+B_1^{(i)}}{\sqrt{2}} \left(1-2\prod_{j\neq i} \frac{1 + A_1^{(j)}}{2}\right)\ket{\psi} \label{eq:nullif2Gbis}
\end{equation}
using the commutation across distinct parties. Now the $z$ and $x$ nullifying conditions can be be combined to obtain
\begin{equation}
\begin{split}
    A_1^{(i)} A_0^{(i)} \ket{\psi} & \underset{\ref{eq:nullif2Gbis} }{=} A_1^{(i)} \frac{B_0^{(i)} + B_1^{(i)}}{\sqrt{2}} \left(1-2\prod_{j\neq i} \frac{1 + A_1^{(j)}}{2}\right)\ket{\psi} \\
    & = \left(1-2\prod_{j\neq i} \frac{1 + A_1^{(j)}}{2}\right) \frac{B_0^{(i)} + B_1^{(i)}}{\sqrt{2}} A_1^{(i)}  \ket{\psi} \\
    & \underset{\ref{eq:nullif1G}}{=} \left(1-2\prod_{j\neq i} \frac{1 + A_1^{(j)}}{2}\right) \frac{B_0^{(i)} + B_1^{(i)}}{\sqrt{2}} \frac{B_0^{(i)} - B_1^{(i)}}{\sqrt{2}}   \ket{\psi} \\
    & = - \left(1-2\prod_{j\neq i} \frac{1 + A_1^{(j)}}{2}\right) \frac{B_0^{(i)} - B_1^{(i)}}{\sqrt{2}} \frac{B_0^{(i)} + B_1^{(i)}}{\sqrt{2}} \ket{\psi} \\
    & \underset{\ref{eq:nullif2Gbis}}{=} -\frac{B_0^{(i)} - B_1^{(i)}}{\sqrt{2}}  A_0^{(i)} \ket{\psi} \\
    & = -A_0^{(i)} \frac{B_0^{(i)} - B_1^{(i)}}{\sqrt{2}}  \ket{\psi} \underset{\ref{eq:nullif1G}}{=}-A_0^{(i)}A_1^{(i)} \ket{\psi},
\end{split}
\end{equation}
or equivalently:
\begin{equation} \label{eq:Aanticomm}
    \{A_0^{(i)}, A_1^{(i)}\}  \ket{\psi} = 0.
\end{equation}

So far, we have thus established that both $A_0^{(i)}, A_1^{(i)}$ and $B_0^{(i)}, B_1^{(i)}$ must be complementary observables on $\ket{\psi}$.
There are now several possible ways to proceed: one could for example build explicitly the local extraction maps (see e.g.~the swap gate isometry presented in Section 4.3 of Ref.~\cite{Supic20}); or, as we will do below, utilize Jordan's lemma (see e.g. Section~7.1.2 of Ref.~\cite{Supic20}). 

Since every party holds two binary measurements, Jordan's lemma implies that we can decompose the local Hilbert spaces into qubit blocks, with the guarantee that measurement operators act on each block independently. The overall statistics can then be expressed as the convex combination of the statistics obtained on each qubit-...-qubit blocks, and the bound of $S_{\text{Grover}}$ can only be saturated if all qubit-...-qubit contributions reach the maximal value $3n/\sqrt{2}$. It thus suffices to show that the only qubit realization that can reach $3n/\sqrt{2}$ is the reference one up to local unitaries. From now on, we thus assume that $\ket{\psi}$ is a normalized $2n$ qubit state and every measurement can be expressed as a norm 1 linear combination of the Pauli operators, i.e.~$A_x^{(i)} = \bm{a}_x^{(i)} \cdot \bm{\sigma}$ with $\bm{\sigma} = (\sigma_x,\sigma_y,\sigma_z)$ and $|\bm{a}_x^{(i)}|=1$ (likewise for B's systems).

Since for each $i$ qubit operators verify $\{A_0^{(i)}, A_1^{(i)}\} \propto \id$, \cref{eq:Aanticomm} implies that $\{A_0^{(i)}, A_1^{(i)}\} = 0$. Likewise, \cref{eq:Banticomm} implies $\{B_0^{(i)}, B_1^{(i)}\} = 0$. This means that there exist a local choice of basis such that the measurements are the ideal CHSH ones (see~\cref{eq:refchsh}). Therefore, the nullifiers are equivalent to the ideal ones in \cref{eq:Xtypenull} and \cref{eq:Ztypenull} up to a choice of local basis and thus identify a unique quantum state, equivalent to $\ket{\psi_{\text{G}}}$ up to local unitaries. This concludes the proof.
\end{proof}

\subsection{Self-testing the Choi channel of the quantum switch} \label{app:self_testing_channel}

In this section we propose a self-testing result for the Choi channel $\mathcal{J}(\bar W_{\text{Switch}})$. This amounts to certifying 
\begin{equation}
    \bar \Psi = \Phi^+_{A^{(C)}B^{(C)}}\otimes \Phi^+_{A^{(T)}B^{(T)}}\otimes \Phi^+_{A^{(1)}B^{(1)}}\otimes \Phi^+_{A^{(2)}B^{(2)}},
\end{equation}
when the channel is not applied, which can be done using~\cref{LemmaSelftest_CHSH_parallel}, and the Choi state $\mathcal{C}(\mathcal{J}(\bar W_{\text{Switch}}))=\ketbra{\psi_\text{out}}$, where
\begin{align} \label{eq:switchstateout}
    \ket{\psi_\text{out}} = \frac{1}{4} \sum_{t,a,b} \Big( & \ket{0,0}_{A^{(C)}B^{(C)}}\ket{t,a,b}_{A^{(T)}A^{(1)}A^{(2)}}\ket{b,t,a}_{B^{(T)}B^{(1)}B^{(2)}} \notag \\[-3mm]
    & + \ket{1,1}_{A^{(C)}B^{(C)}}\ket{t,a,b}_{A^{(T)}A^{(1)}A^{(2)}}\ket{a,b,t}_{B^{(T)}B^{(1)}B^{(2)}} \Big).
\end{align}
The idea to self-test this state is in two steps. First, one observes that when projecting the control system $A^{(C)}$ on the computational basis $\{\ket{0}, \ket{1}\}$, the remaining target and auxiliary systems are product of three maximally entangled states. Second, when simultaneously projecting the target and auxiliary systems in the $\ket{0}$ state, the control systems are left in a maximally entangled state. We use this intuition to derive a Bell inequality, whose quantum bound self-test the target state~\cref{eq:switchstateout}. The method to derive this inequality, the SOS method, is presented in the previous subsection~\ref{sec:STGrover}, while the self-testing result is formalized in~\cref{TheoremSelftest}. We emphasize that all the measurements used here are identical to the ones required by~\cref{LemmaSelftest_CHSH_parallel}, allowing to overall self-test the channel $\mathcal{J}(\bar W_{\text{Switch}})$.\\

\noindent We present below how we apply the SOS method (see~\cref{sec:STGrover} for more details) to the state~\cref{eq:switchstateout}. \\

\noindent \underline{\textbf{Step 1:}} In order to find the nullifiers, we first observe that when projecting the control system $A^{(C)}$ on the computational basis $\{\ket{0}, \ket{1}\}$, the remaining target and auxiliary systems are product of three maximally entangled states. This leads to the following nullifiers for $I\in\{T,1,2\}$:
\begin{equation} \label{eq:Cprojectednullifiers}
    \begin{split}
        & \hat{N}_{z,0,1} = \hat{\Pi}_0^{(C)} ({\sigma_z}^{A^{(T)}}-{\sigma_z}^{B^{(1)}}), \quad \hat{N}_{z,0,2} = \hat{\Pi}_0^{(C)} ({\sigma_z}^{A^{(1)}}-{\sigma_z}^{B^{(2)}}), \quad \hat{N}_{z,0,3} = \hat{\Pi}_0^{(C)} ({\sigma_z}^{A^{(2)}}-{\sigma_z}^{B^{(T)}}), \\
        & \hat{N}_{x,0,1} = \hat{\Pi}_0^{(C)} ({\sigma_x}^{A^{(T)}}-{\sigma_x}^{B^{(1)}}), \quad \hat{N}_{x,0,2} = \hat{\Pi}_0^{(C)} ({\sigma_x}^{A^{(1)}}-{\sigma_x}^{B^{(2)}}), \quad \hat{N}_{x,0,3} = \hat{\Pi}_0^{(C)} ({\sigma_x}^{A^{(2)}}-{\sigma_x}^{B^{(T)}}), \\
        & \hat{N}_{z,1,1} = \hat{\Pi}_1^{(C)} ({\sigma_z}^{A^{(T)}}-{\sigma_z}^{B^{(2)}}), \quad \hat{N}_{z,1,2} = \hat{\Pi}_1^{(C)} ({\sigma_z}^{A^{(1)}}-{\sigma_z}^{B^{(T)}}), \quad \hat{N}_{z,1,3} = \hat{\Pi}_1^{(C)} ({\sigma_z}^{A^{(2)}}-{\sigma_z}^{B^{(1)}}), \\
        & \hat{N}_{x,1,1} = \hat{\Pi}_1^{(C)} ({\sigma_x}^{A^{(T)}}-{\sigma_x}^{B^{(2)}}), \quad \hat{N}_{x,1,2} = \hat{\Pi}_1^{(C)} ({\sigma_x}^{A^{(1)}}-{\sigma_x}^{B^{(T)}}), \quad \hat{N}_{x,1,3} = \hat{\Pi}_1^{(C)} ({\sigma_x}^{A^{(2)}}-{\sigma_x}^{B^{(1)}}), 
    \end{split}
\end{equation}
where ${\sigma_z}^I, {\sigma_x}^I$ denote the Pauli operators on party $I$ and $\hat{\Pi}_s^{(i)} := (1+(-1)^s {\sigma_z}^{A^{(i)}})/2$ is the projector on the computational basis of $A^{(i)}$. 

The next nullifiers come from the fact that whenever parties $A^{(T)}$, $A^{(1)}$ and $A^{(2)}$ are projected simultaneously on $\ket{0}$, the marginal state between $A^{(C)}$ and $B^{(C)}$ becomes the maximally entangled state $\ket{\phi^+}$. This translates to 
\begin{equation} \label{eq:Cnullifiers}
    \begin{split}
        & \hat{N}_{z,C} = \hat{\Pi}_0^{(T)} \hat{\Pi}_0^{(1)} \hat{\Pi}_0^{(2)} ({\sigma_z}^{A^{(C)}} - {\sigma_z}^{B^{(C)}}), \\
        & \hat{N}_{x,C} =\hat{\Pi}_0^{(T)} \hat{\Pi}_0^{(1)} \hat{\Pi}_0^{(2)} ({\sigma_x}^{A^{(C)}} - {\sigma_x}^{B^{(C)}})
    \end{split}
\end{equation}
nullifying the state~\cref{eq:switchstateout}. 

We now have to check that the set of nullifiers is complete, i.e.~will be sufficient to uniquely identify the target state. To do so, notice that an arbitrary 8 qubit pure state can always be written $\ket{\psi} = \ket{0}_{A^{(C)}} \otimes \ket{\psi_0} + \ket{1}_{A^{(C)}} \otimes \ket{\psi_1}$. If all operators in~\cref{eq:Cprojectednullifiers} nullify $\ket{\psi}$, one can easily see that it implies
\begin{equation}
    \begin{split}
        & \ket{\psi_0} = \ket{\xi_0}_{B^{(C)}}\otimes \ket{\phi^+}_{A^{(T)}B^{(1)}}\otimes \ket{\phi^+}_{A^{(1)}B^{(2)}}\otimes \ket{\phi^+}_{A^{(2)}B^{(T)}}, \\
        & \ket{\psi_1} = \ket{\xi_1}_{B^{(C)}}\otimes \ket{\phi^+}_{A^{(T)}B^{(2)}}\otimes \ket{\phi^+}_{A^{(1)}B^{(T)}}\otimes \ket{\phi^+}_{A^{(2)}B^{(1)}},
    \end{split}
\end{equation}
where for $s\in\{0,1\}$, $\ket{\xi_s}$ is a sub-normalized state on $B^{(C)}$. Then, one can see that the projection on $\hat{\Pi}_0^{(T)} \hat{\Pi}_0^{(1)} \hat{\Pi}_0^{(2)}$ of the overall state sends the marginal state on $A^{(C)}$ and $B^{(C)}$ to $\ket{0}_{A^{(C)}} \ket{\xi_0}+ \ket{1}_{A^{(C)}} \ket{\xi_1}$. Now, imposing that operators in~\cref{eq:Cnullifiers} are nullifiers, we get that $\ket{\xi_s} = \ket{s}/\sqrt{2}$ and thus $\ket{\psi}=\ket{\psi_{\text{out}}}$ as desired.\\

\noindent \underline{\textbf{Step 2:}} We choose the measurement operators to be the same as in the CHSH test, see~\cref{eq:refchsh}. As such, the Pauli operators express as 
\begin{equation}
    \begin{split}
        & {\sigma_z}^{A^{(i)}} = A_0^{(i)}, \quad {\sigma_x}^{A^{(i)}} = A_1^{(i)},\\
        & {\sigma_z}^{B^{(i)}} = \frac{B_0^{(i)}+B_1^{(i)}}{\sqrt{2}}, \quad {\sigma_x}^{B^{(i)}} = \frac{B_0^{(i)}-B_1^{(i)}}{\sqrt{2}},
    \end{split}
\end{equation}
in terms of the measurement operators for all $i \in \{C,T,1,2\}$. We can then insert this in~\cref{eq:Cprojectednullifiers} and~\cref{eq:Cnullifiers} to obtain the expression of the nullifiers in terms of measurement operators.\\

\noindent \underline{\textbf{Step 3:}} We can now lift the nullifiers to the formal polynomial algebra and compute the SOS
\begin{equation}
    P := \sum_{r\in\{z,x\}} \left(\sum_{s\in \{0,1\}, \ t\in \{1,2,3\}} N_{r,s,t}^2 \right) + N_{r,C}^2. 
\end{equation}
To compute this we use the following identities: $(A_x^{(i)})^2 = (B_y^{(i)})^2=1$, $(\Pi_s^{(i)})^2 =\Pi_s^{(i)}$ and measurements across different parties commute. Note in particular that
\begin{equation}
    N_{z,0,1}^2 = \Pi_0^{(C)} (2 - \sqrt{2} A_0^{(T)} (B_0^{(1)}+B_1^{(1)}) + \frac{1}{2} \{B_0^{(1)}, B_1^{(1)}\}), 
\end{equation}
while 
\begin{equation}
    N_{x,0,1}^2 = \Pi_0^{(C)} (2 - \sqrt{2} A_1^{(T)} (B_0^{(1)}-B_1^{(1)}) - \frac{1}{2} \{B_0^{(1)}, B_1^{(1)}\}).
\end{equation}
The main idea for building the SOS in this way is that when summing the last two terms over $r\in\{z,x\}$ we cancel out the non-measurable terms $\Pi_0^{(C)} \{B_0^{(1)}, B_1^{(1)}\}$ while the measurement part is proportional to $\beta^{(A^{(T)}B^{(1)})}_\text{CHSH}$. \\

\noindent \underline{\textbf{Step 4:}} A direct algebraic computation of the full SOS gives: 
\begin{equation} \label{eq:SOSswitch}
\begin{split}
    P = 12 - \sqrt{2}S_{\text{Switch}}
\end{split}
\end{equation}
where we used the fact that $\Pi_0^{(C)}+\Pi_1^{(C)}=1$ and 
\begin{equation} \label{eq:BellopswitchST}
\begin{split}
    S_{\text{Switch}} := \Pi_0^{(C)}&  \left[ \beta^{(A^{(T)}B^{(1)})}_\text{CHSH}+\beta^{(A^{(1)}B^{(2)})}_\text{CHSH}+\beta^{(A^{(2)}B^{(T)})}_\text{CHSH}\right]  + \Pi_1^{(C)} \left[ \beta^{(A^{(T)}B^{(2)})}_\text{CHSH}+\beta^{(A^{(1)}B^{(T)})}_\text{CHSH}+\beta^{(A^{(2)}B^{(1)})}_\text{CHSH}\right] \\
    & + \Pi_0^{(T)} \Pi_0^{(1)} \Pi_0^{(2)} \left[\beta^{(A^{(C)}B^{(C)})}_\text{CHSH}-2\sqrt{2}\right], 
\end{split}
\end{equation}
is the candidate Bell expression. Note that the SOS immediately implies that $S_{\text{Switch}} \preceq 6\sqrt{2}$. Since the ideal realization (with state~\cref{eq:switchstateout} and CHSH measurements~\cref{eq:refchsh}) saturates this bounds, this is indeed the quantum bound of this polynomial. We can now move to our main result.

\begin{theorem} \textbf{Self-testing $\mathcal{C}(\mathcal{J}(\bar W_{\text{Switch}}))$.} \label{TheoremSelftest} \\
    Let us consider an 8-partite Bell scenario where each party performs two binary measurements. For every ``channel-free'' system, we denote $\Pi_s^{(i)} = (\id+(-1)^s A_0^{(i)})/2$ the projection on the eigenspaces of the first measurement operator. \\
    Consider the Bell expression~\cref{eq:BellopswitchST}. If the observed statistics verify  
\begin{equation}
    \Tr(\rho S_{\text{Switch}}) = 6\sqrt{2},
\end{equation} the state $\mathcal{C}(\mathcal{J}(\bar W_{\text{Switch}}))=\ketbra{\psi_{\text{out}}}$ (see~\cref{eq:switchstateout}) is self-tested. 
\end{theorem}

\begin{proof}
The proof then follows the same steps to the one of~\cref{TheoremSelftestGrover}. Namely, we first use the SOS decomposition to proove anti-commutation relations for observables on the states. Then, we utilize Jordan's lemma to reduce to qubit realizations and conclude using the completeness of the set of nullifiers.

Specifically, without loss a generality (dilation to larger spaces), we assume that the shared state $\rho = \ketbra{\psi}$ is pure. Due to the SOS decomposition~\cref{eq:SOSswitch}, it holds that $N_{r,s,t} \ket{\psi} = N_{r,C}\ket{\psi} = 0 $ for all $r,s,t$. Now this implies that the nullifier operators are the ideal ones when acting on $\ket{\psi}$. For example, using the first relation for $r,s,t=z,0,1$ and $r,s,t=z,1,3$, we obtain that $\Pi_s^{(C)} \{B_0^{(1)}, B_1^{(1)}\} \ket{\psi} = 0 $ and summing over $s$ we get $\{B_0^{(1)}, B_1^{(1)}\} \ket{\psi} = 0 $. Similar relations can be used to derive anti-commutation relations for parties $A/B^{(T/1/2)}$, i.e.
\begin{equation}\label{eq:12TcommSwitch}
    \begin{split}
        & \{A_0^{(i)}, A_1^{(i)}\} \ket{\psi} = \{B_0^{(i)}, B_1^{(i)}\} \ket{\psi} = 0, 
    \end{split}
\end{equation}
for $i \in \{1,2,T\}$. For the control systems, $N_{r,C}\ket{\psi} = 0$ gives rise to
\begin{equation}\label{eq:Ccommswitch}
    \{A_0^{(C)}, A_1^{(C)}\}\Pi_0^{(T)} \Pi_0^{(1)} \Pi_0^{(2)} \ket{\psi} = \{B_0^{(C)}, B_1^{(C)}\}\Pi_0^{(T)} \Pi_0^{(1)} \Pi_0^{(2)} \ket{\psi}=0.
\end{equation}
Note that one can also make a statement one the norm of $\Pi_0^{(T)} \Pi_0^{(1)} \Pi_0^{(2)} \ket{\psi}$. Indeed, observe that:
\begin{equation} \label{eq:painfulderivation}
    \begin{split}
        \bra{\psi}\Pi_0^{(C)}\Pi_0^{(T)} \Pi_0^{(1)} \Pi_0^{(2)} \ket{\psi} & \underset{xx}{=} \bra{\psi}\Pi_0^{(C)}\left(\frac{B_0^{(1)}-B_1^{(1)}}{\sqrt{2}}\right)^2\Pi_0^{(T)} \Pi_0^{(1)} \Pi_0^{(2)} \ket{\psi}  \\
        & = \bra{\psi}\Pi_0^{(C)}\left(\frac{B_0^{(1)}-B_1^{(1)}}{\sqrt{2}}\right)\Pi_0^{(T)} \Pi_0^{(1)} \Pi_0^{(2)} \left(\frac{B_0^{(1)}-B_1^{(1)}}{\sqrt{2}}\right) \Pi_0^{(C)} \ket{\psi}\\
        & \underset{N_{x,0,1}\ket{\psi}=0}{=} \bra{\psi}\Pi_0^{(C)}A_1^{(T)}\Pi_0^{(T)} \Pi_0^{(1)} \Pi_0^{(2)} A_1^{(T)} \Pi_0^{(C)} \ket{\psi}\\
        & \underset{\ref{eq:12TcommSwitch}}{=} \bra{\psi}\Pi_0^{(C)}\Pi_1^{(T)} \Pi_0^{(1)} \Pi_0^{(2)} \ket{\psi}, 
    \end{split}
\end{equation}
and thus using $\Pi_0^{(T)} + \Pi_1^{(T)} = \id$, one has
\begin{equation}
    \bra{\psi}\Pi_0^{(C)}\Pi_0^{(T)} \Pi_0^{(1)} \Pi_0^{(2)} \ket{\psi} = \frac{1}{2} \bra{\psi}\Pi_0^{(C)} \Pi_0^{(1)} \Pi_0^{(2)} \ket{\psi}.
\end{equation}
Doing the same as in \cref{eq:painfulderivation} but inserting operators on $B^{(2)}$ and using $N_{x,0,2} \ket{\psi} = 0$, we obtain: 
\begin{equation}
    \bra{\psi}\Pi_0^{(C)} \Pi_0^{(1)} \Pi_0^{(2)} \ket{\psi} = \frac{1}{2} \bra{\psi}\Pi_0^{(C)}  \Pi_0^{(2)} \ket{\psi}.
\end{equation}
Last, doing the same again inserting operators on $B^{(T)}$ and using $N_{x,0,3} \ket{\psi} = 0$, we get: 
\begin{equation}\label{eq:lastpainful}
    \bra{\psi}\Pi_0^{(C)} \Pi_0^{(2)} \ket{\psi} = \frac{1}{2} \bra{\psi}\Pi_0^{(C)} \ket{\psi}.
\end{equation}
Combining Eqs.~(\ref{eq:painfulderivation}-\ref{eq:lastpainful}) together with their analog for $\Pi_1^{(C)}$, and using the fact that $\Pi_0^{(C)}+\Pi_1^{(C)}=\id$, we finally obtain
\begin{equation} \label{eq:projectionnorm}
    \bra{\psi}\Pi_0^{(T)} \Pi_0^{(1)} \Pi_0^{(2)} \ket{\psi} = \bra{\psi}\ket{\psi} = 1/8.
\end{equation}

Next, since every party holds two binary measurements, we can use Jordan's lemma~\cite{Supic20} to decompose the local Hilbert spaces to qubit blocks. The overall statistics can then be expressed as the convex combination of the statistics obtained on each qubit-...-qubit blocks, and the bound of $S_{\text{Switch}}$ can only be saturated if all qubit-...-qubit contributions reach the maximal value $\sqrt{2}n$. It thus suffices to show that the only qubit realization that can reach $\sqrt{2}n$ is the reference one up to local unitaries.

Since we can  now assume that the realization is on local qubit Hilbert spaces, the measurements  verify $\{A_0^{(i)}, A_1^{(i)}\}~\propto~\id$. \cref{eq:12TcommSwitch} then implies that $\{A_0^{(i)}, A_1^{(i)}\}=0$ for $i\in\{1,2,T\}$. Moreover, since \cref{eq:projectionnorm} implies that $\Pi_0^{(T)} \Pi_0^{(1)} \Pi_0^{(2)} \ket{\psi}$ is a non-zero vector, \cref{eq:Ccommswitch} proves that $\{A_0^{(C)}, A_1^{(C)}\}=0$. Likewise, $\{B_0^{(i)}, B_1^{(i)}\}=0$ for $i\in\{1,2,T,C\}$. This means that there exist a local choice of basis such that all measurements are the ideal CHSH ones (see~\cref{eq:refchsh}). Therefore, the nullifiers are equivalent to the ideal ones in \cref{eq:Cprojectednullifiers} and \cref{eq:Cnullifiers} up to a choice of local basis and thus identify a unique quantum state, equivalent to $\ket{\psi_{\text{out}}}$ up to local unitaries. This concludes the proof. 
\end{proof}

%%%%%%%%%%%%%%%%%%%%%%%%%%%%%%%%%%%%%%%%%
\section{Proof of~\cref{lemma1}} \label{app:lemmas}

In this appendix, we present the proof of \cref{lemma1} of the main text. We first introduce a number of preliminary lemmas that will be used to establish the desired result. We prove these through the diagrammatic notations already used in the main text, where green wires and boxes represent trusted systems and operations (whose dimensions are known, assumed to be finite), red wires and boxes represent untrusted systems and operations (whose dimensions are unknown, allowed to be either finite or countably infinite), while the introduction of a blue wire or a blue box is to be interpreted as ``there exists a system or an operation such that...''. All boxes in the diagrams below represent CPTP (deterministic) operations; the ordering of the wires, from top to bottom, on each side of an inequality is the same.\footnote{These conventions are the same as those used in the main text, except that occasionally we changed the ordering of the wires in an equation---in which case we were using labels to clearly identify these: see e.g.~\cref{eq: grover Choi,eq: choi swap,eq:STJC_Switch}.}

\begin{lemma}\label{prop0_1} For any CPTP map $\mathcal{E}$, 
\begin{align}
    \begin{aligned}
        \includegraphics[width= 8cm]{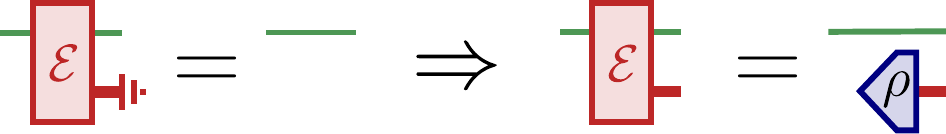} \ \ .
    \end{aligned}
\end{align}
\end{lemma}
\begin{proof} Let us plug in half a maximally entangled state $\Phi^+$ into $\mathcal{E}$. By the hypothesis of the lemma, the Choi state $\mathcal{C}(\mathcal{E})$ thus obtained is equal to $\Phi^+$ after tracing out the untrusted red system. It must therefore be product, i.e.~there must exist a state $\rho$ such that
\begin{align}
    \begin{aligned}
        \includegraphics[height= 1.6cm]{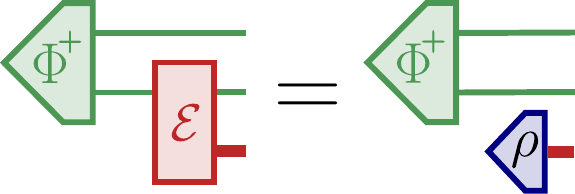} \ \ .
    \end{aligned}
\end{align}
The conclusion of the lemma is then obtained after applying the inverse \ChJa isomorphism to remove $\Phi^+$ in the equation above.
\end{proof}

\begin{lemma} \label{prop1} For any CPTP map $\mathcal{T}$,
    \begin{align}
    \begin{aligned}
        \includegraphics[height= 1.5cm]{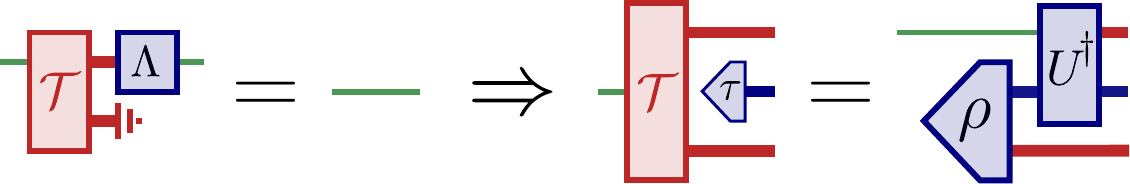} \ \ ,
    \end{aligned}
\end{align} 
where $U^\dag$ is a unitary operation. 
\end{lemma}

\begin{proof}
    Let us consider a unitary Stinespring dilation of the channel $\Lambda$, i.e.
    \begin{align}
    \begin{aligned}
        \includegraphics[height= 1cm]{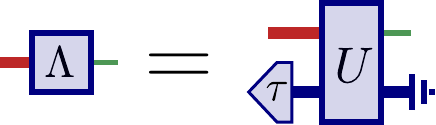} \ \ ,
    \end{aligned}
\end{align}
for some unitary channel $U$ and some state $\tau$.
The hypothesis of the lemma can then be written
\begin{align}
    \begin{aligned}
        \includegraphics[height= 1.5cm]{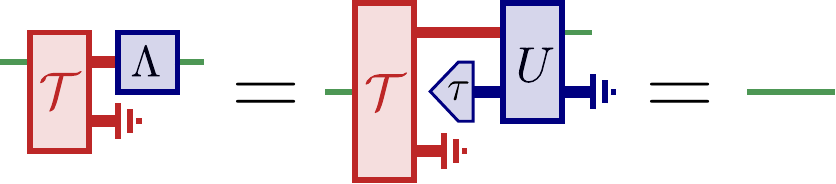} \ \ ,
    \end{aligned}
\end{align}
and it immediately follows from~\cref{prop0_1} that 
\begin{align} 
    \begin{aligned}
        \includegraphics[height= 1.5cm]{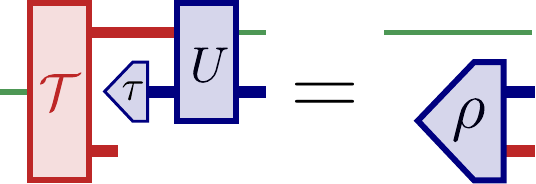} \ \ .
    \end{aligned}
\end{align}
Applying the inverse unitary channel $U^\dag$ on the top two output wires cancels $U$ in the lhs of the equation above and concludes the proof.
\end{proof}

\begin{lemma}\label{prop0_2} For any CPTP map $\mathcal{E}$ and any state $\rho$,
    \begin{align}
    \begin{aligned}
        \includegraphics[height= 1.5 cm]{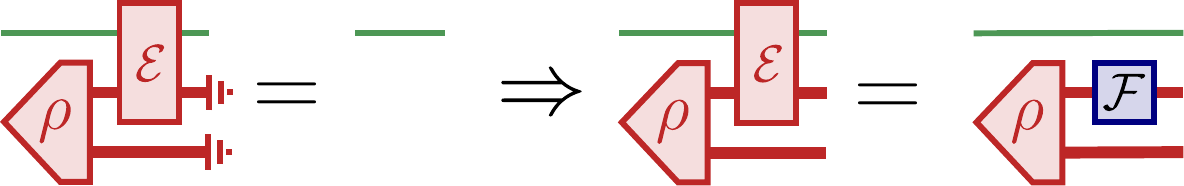} \ \ .
    \end{aligned}
\end{align}
\end{lemma}

\begin{proof}
    By virtue of the \ChJa isomorphism, we do everything at the level of the Choi states. By hypothesis of the lemma, the state
    \begin{align}\label{app: lemm 2 eq1}
    \begin{aligned}
       \includegraphics[height=2cm]{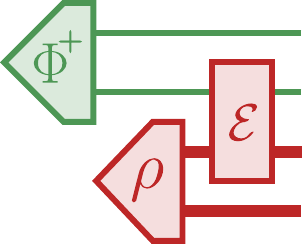} \ \ ,
    \end{aligned}
\end{align}
is equal to $\Phi^+$ after tracing the untrusted red systems. It must therefore be product, i.e.~there exists a state $\tau$ such that 
\begin{align}\label{app: lemm 2 eq2}
    \begin{aligned}
        \includegraphics[height=2cm]{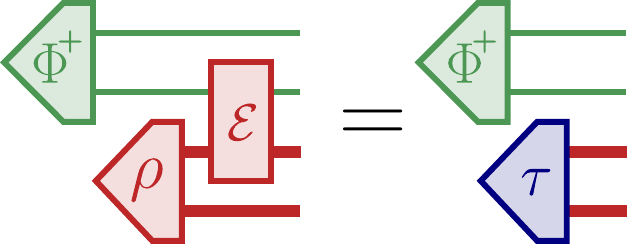} \ \ .
    \end{aligned}
\end{align}
Furthermore, tracing out the trusted green systems, we obtain 
\begin{align}
\begin{aligned}
    \includegraphics[height=2.1cm]{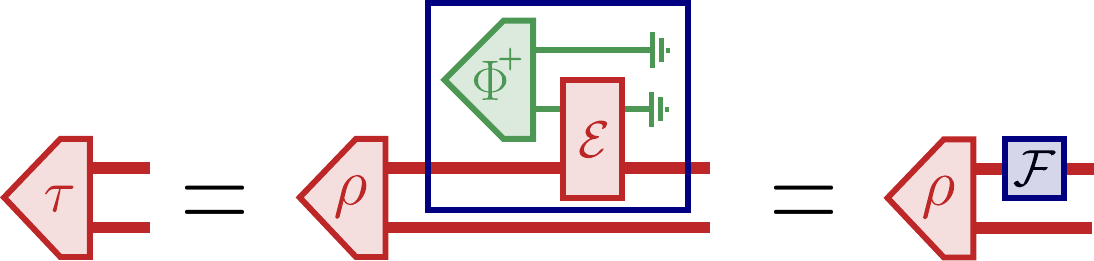} \ \ .
\end{aligned}
\end{align}
Finally, after plugging this expression of the state $\tau$ into Eq.~\eqref{app: lemm 2 eq2}, the conclusion of the lemma is obtained  by the \ChJa isomorphism.
\end{proof}

\begin{lemma}    
\label{cor0} For any CPTP map $\mathcal{E}$ and any state $\rho$
    \begin{align}
    \begin{aligned}
        \includegraphics[height= 2cm]{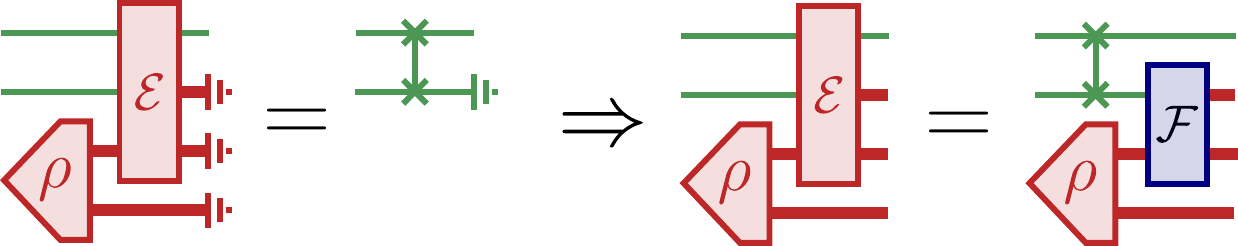} \ \ .
    \end{aligned}
\end{align}
\end{lemma}
\begin{proof}
    First, we apply the \SWAP gate on both green wires to define the channel $\cE'$ satisfying 
       \begin{align}
    \begin{aligned}
        \includegraphics[height= 2cm]{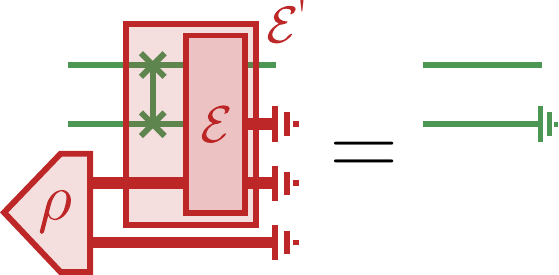} \ \ .
    \end{aligned}
\end{align}
Next, we insert half of a maximally entangled state into the second green wire to obtain the identity
\begin{align}
    \begin{aligned}
        \includegraphics[height= 2.3cm]{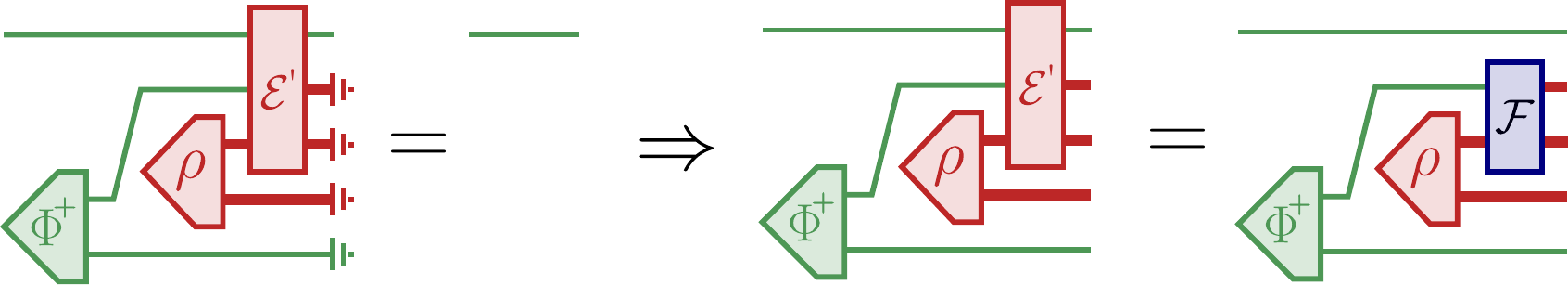} \ \,
    \end{aligned}
\end{align}
where the implication follows by applying the \cref{prop0_2} to this specific configuration with a product input state. By the \ChJa isomorphism, the last identity can be rewritten as 
\begin{align}
    \begin{aligned}
        \includegraphics[height= 2cm]{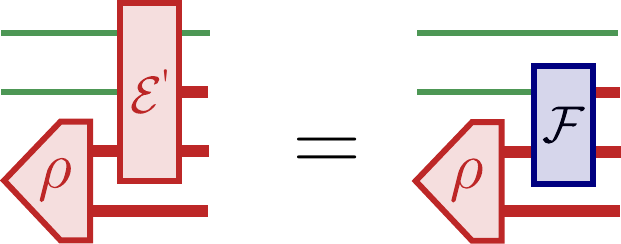} \ \ ,
    \end{aligned}
\end{align}
from which the conclusion of the lemma follows by applying a \SWAP gate onto the two green input wires to recover the original channel $\cE$ from $\cE'$.
\end{proof}

\begin{lemma}\label{lemma0}
    For any CPTP maps  $\cX$ and $\mathcal{T}$
\begin{align}\label{eq: lem01}
\begin{aligned}\includegraphics[height=1.5cm]{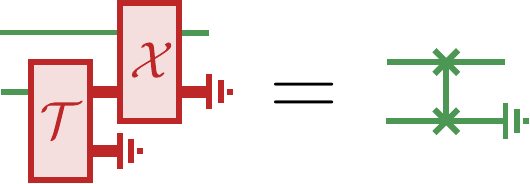}
\end{aligned} \quad \mathlarger{\mathlarger{\mathlarger{\Rightarrow}}} \ \ \  
\begin{aligned}\includegraphics[height= 1.85cm]{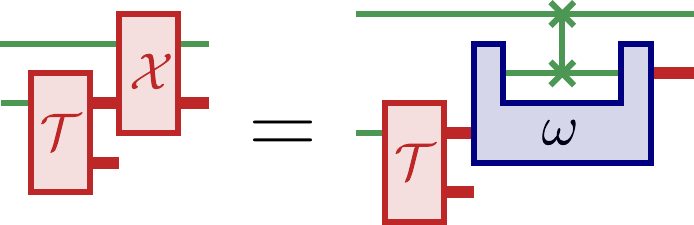} \ .
\end{aligned}
\end{align}
\end{lemma}

\begin{proof}
    The first thing is to see that the hypothesis directly implies that 
    \begin{align*}
        \begin{aligned}
            \includegraphics[height= 1cm]{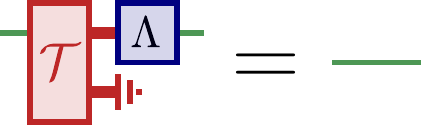} \end{aligned} \quad \text{, with } \quad  \begin{aligned}\includegraphics[height= 1cm]{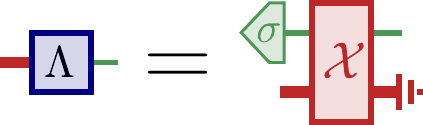} \ ,
        \end{aligned}
    \end{align*}
    for any state $\sigma$. It then follows from~\cref{prop1} that 
    \begin{align}\label{eqapp:prop1consequence}
    \begin{aligned}
        \includegraphics[height= 1.5cm]{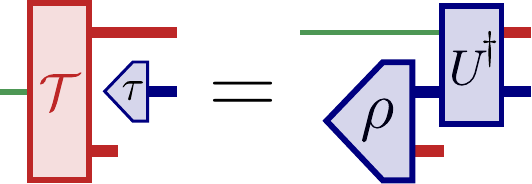} \ \ .
    \end{aligned}
\end{align}
We can then rewrite the hypothesis of the lemma as
\begin{align}
    \begin{aligned}
        \includegraphics[height= 2.1cm]{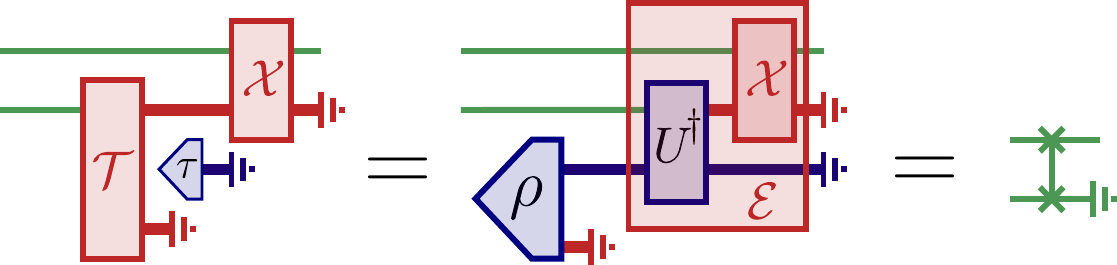} \ \ .
    \end{aligned}
\end{align}
From the last equality, applying the \cref{cor0} for the state $\rho$ and the channel $\mathcal{E}$ identified in the above diagram we obtain
\begin{align}\label{appeq: temp1}
    \begin{aligned}
        \includegraphics[height= 1.7cm]{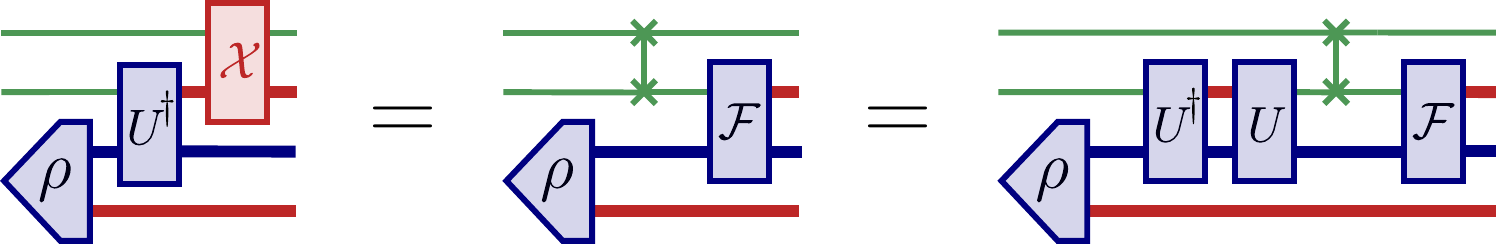} \ \ ,
    \end{aligned}
\end{align}
where to obtain the last equality we used the fact that $U^\dag$ is unitary to insert the identity $U\circ U^\dag ={\rm id}$. Finally, plugging  Eq.~\eqref{eqapp:prop1consequence} into Eq.~\eqref{appeq: temp1} we obtain 
\begin{align}
    \begin{aligned}
        \includegraphics[height= 1.7cm]{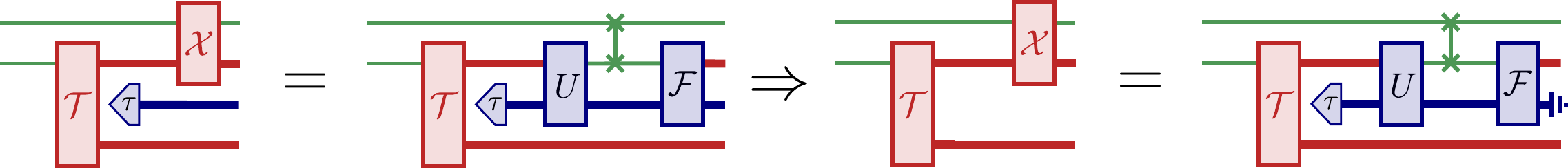} \ \,
    \end{aligned}
\end{align}
which proves the lemma upon defininig the embedding comb $\omega$ as
\begin{align}
    \begin{aligned}
        \includegraphics[height= 1cm]{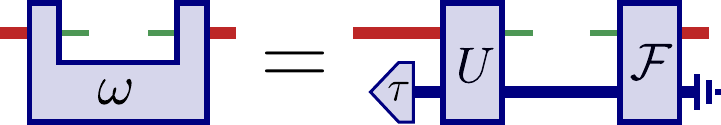} \ \ .
    \end{aligned}
\end{align}
\end{proof}

\begin{lemma}\label{lemma1old}
    For any bipartite channel $\mathcal{X}$ and any one-slot comb $\Omega$ such that
\begin{align}\label{eq: lem61}
\begin{aligned}\includegraphics[height=2cm]{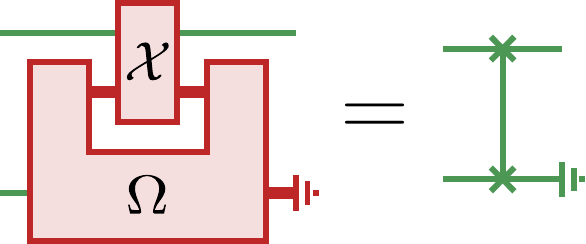}\ \ ,
\end{aligned}
\end{align}
    it holds true that
\begin{align}\label{eq: lem62}
\begin{aligned}\includegraphics[height=2.5cm]{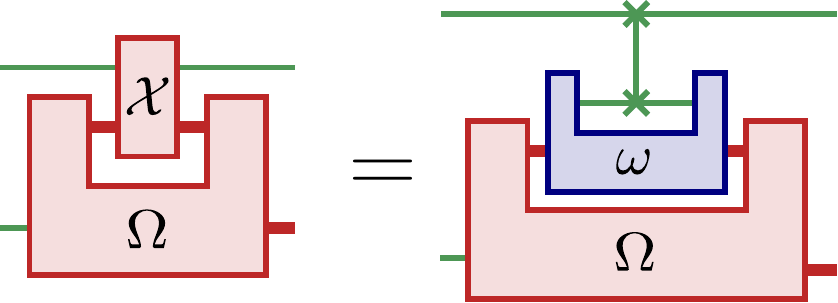}\ \ .
\end{aligned}
\end{align}
\end{lemma}

\begin{proof} For any 1-comb it holds true that \cite{taranto2025}:
    \begin{equation}
        \includegraphics[height=1.5cm]{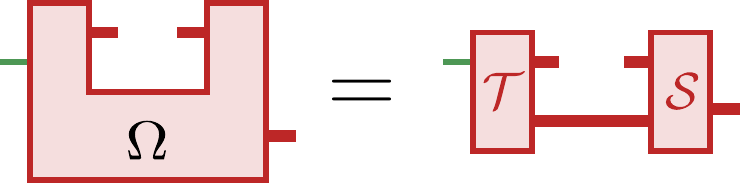}
    \end{equation}
    for some CPTP maps $\mathcal{T}, \mathcal{S}$. Hence the hypothesis becomes  
    \begin{equation}
        \includegraphics[height=2cm]{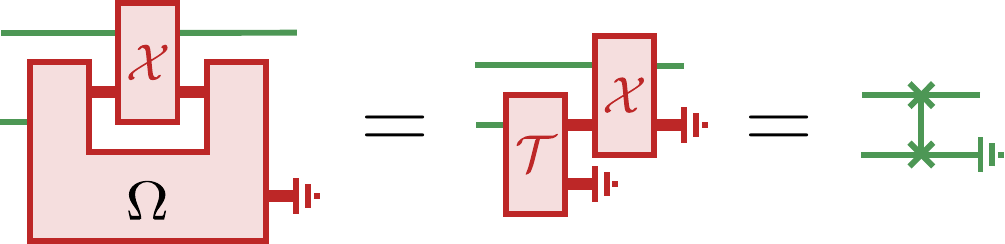} \ ,
    \end{equation}
    from which we can apply~\cref{lemma0}. After composing with $\mathcal{S}$, we get Eq.~\eqref{eq: lem12}.
\end{proof}

\LemmaOne*

\begin{proof} First, let us define a new comb 
     \begin{equation}
        \includegraphics[height=3cm]{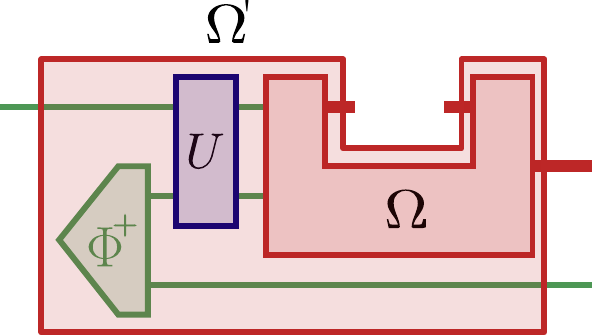},
    \end{equation}
    for which the statement of the lemma guarantees
       \begin{equation}
        \includegraphics[height=2.1cm]{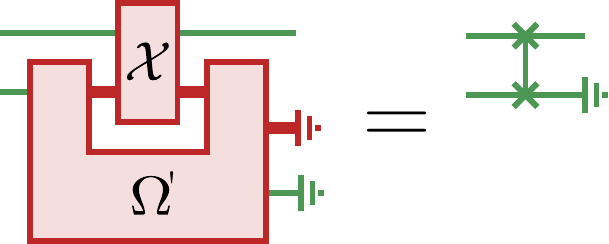}.
    \end{equation}
    By \cref{lemma1old} this identity implies
     \begin{equation}
        \includegraphics[height=3cm]{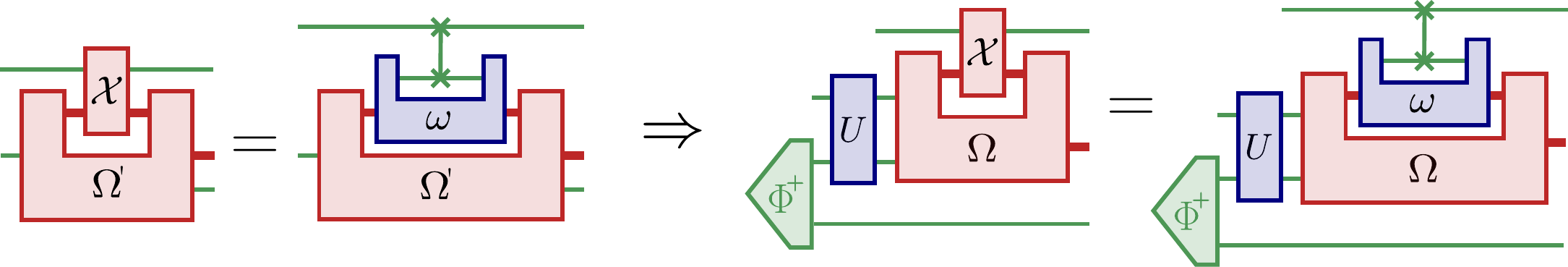}.
    \end{equation}
    We conclude the proof by using the \ChJa isomorphism to remove the $\Phi^+$ state and applying the inverse $U^\dagger$ operation on the two input wires of $\Omega$.
\end{proof}

\end{document}